\documentclass[11pt]{article}
\usepackage{graphicx,lscape,rotate,subcaption}
\usepackage{float,captions2}
\usepackage[usenames]{color}
\usepackage{amsmath,amssymb,amsfonts,mathrsfs}
\usepackage{amsthm}
\usepackage{euscript}
\usepackage{natbib}
\usepackage{pifont}
\usepackage{multirow}
\usepackage{verbatim}
\usepackage[outdir=./]{epstopdf}
\usepackage{fullpage}
\usepackage{booktabs}  
\usepackage{setspace}
\usepackage[geometry]{ifsy m}
\usepackage{cases}
\usepackage{tikz}
\usepackage{lipsum}
\usepackage{longtable} 
\usepackage{booktabs}
\usepackage{shuffle}
\usepackage{longtable}
\usepackage{csquotes}
\usepackage{color, colortbl}
\definecolor{Gray}{gray}{0.9}

\usepackage{bbm}
\usepackage{cprotect}
\usepackage{fancyvrb}
\usepackage{caption}

\RequirePackage[colorlinks,citecolor=blue!70!green,urlcolor=blue, linkcolor=blue!70!green]{hyperref}

\newcommand{\tonde}[1]{\left(#1\right)}
\allowdisplaybreaks
\newtheorem{thm}{Theorem}
\newtheorem{cor}{Corollary}
\newtheorem{definition}{Definition}
\newtheorem{lemma}{Lemma}
\newtheorem{remark}{Remark}

\newlength{\defaultbaselineskip}
\newcommand{\R}{{\mathbb R}}
\newcommand{\E}{{\mathbb E}}

\newcommand{\Tcal}{{\mathcal T}}
\newcommand{\intphitilde}[1]{ \frac{\int_0^{\tau}\tilde{\phi}(#1)d W_{#1}}{\sqrt{\tau}}}

\newcommand{\intphitildesqminus}[1]{ 
e^{-\frac{u^2}{2\tau}\int_0^{\tau}\tilde{\phi}(#1)^2 d #1}}

\newcommand{\intphitildesqplus}[1]{ 
e^{\frac{u^{2}}{2\tau}\int_0^{#1}\tilde{\phi}(u)^2d u}}

\definecolor{verdescuro}{HTML}{006400}  
\definecolor{rossoval}{HTML}{D60029}  

\usepackage{todonotes}

\begin{document}

\begin{titlepage}

\title{Ultra-short-term volatility surfaces\thanks{\footnotesize{We thank the participants in the XXVI Workshop in Quantitative Finance (Palermo, April 15-17, 2025), the Quantitative Finance and Financial Econometrics Conference (Marseille, June 3-6, 2025), the Conference in Honor of David Bates (Iowa City, October 10, 2025) and the Derivatives and Asset Pricing Conference (Cozumel, February 26-28, 2026) for their comments. We also thank seminar participants at the University of Amsterdam, the University of Rotterdam, the University of Verona and the University of Vienna. We are grateful to Oleg Bondarenko (the Cozumel discussant), Yannick Dillschneider, Gustavo Freire and Paola Pederzoli (the Iowa discussant) for their helpful suggestions.}}}

\author{Federico M. Bandi{\thanks{\noindent Johns Hopkins University, Carey Business School, 555 Pennsylvania Avenue, Washington, DC 20001, USA;  e-mail:\   fbandi1@jhu.edu.}~~~~~~~~Nicola
 Fusari\thanks{\noindent Johns Hopkins University, Carey Business School, 100 International Drive, Baltimore, MD 21202, USA;  e-mail:\   nicola.fusari@jhu.edu.}~~~~~~~~Guido Gazzani{\thanks{\noindent University of Verona, Dipartimento Scienze Economiche, via Cantarane 24, Verona, Italy;  e-mail:\  guido.gazzani@univr.it}~~~~~~~~Roberto
 Ren\`o\thanks{\noindent Essec Business School, 3 Av. Bernard Hirsch, 95000 Cergy, France;  e-mail:\ reno@essec.edu. ~}}}}

\end{titlepage}
 
\maketitle

\begin{abstract}
\singlespacing Options with maturities below one week, hereafter {\it ultra-short-term} options, have seen a sharp increase in trading activity in recent years. Yet, these instruments are difficult to price \textit{jointly} using classical pricing models due to the pronounced oscillations observed in the at-the-money implied-volatility term structure across {\it ultra-short-term} tenors. We propose {\it Edgeworth++}, a parsimonious jump–diffusion model featuring a nonparametric stochastic volatility component, which provides flexibility in capturing the implied-volatility smiles for \textit{each} tenor, combined with a deterministic shift extension, which allows the model to fit rich at-the-money implied-volatility term structures \textit{across} tenors. We derive a local (in tenor) expansion of the process characteristic function suited to value \textit{ultra-short-term} options. The expansion leads to fast and accurate option pricing in closed form via standard Fourier inversion. We discuss the benefits of the proposed approach relative to natural benchmarks.
\end{abstract}



\newpage

\section{Introduction}
The option market has undergone significant recent changes, with trading activity increasingly concentrated in very short tenors. Fig.~\ref{fig:short_trading} reports trading volume, in terms of number of contracts, across S\&P 500 (SPX) options with different maturities. Since 2018, the segment with expiries lower than 7 calendar days has been the most actively traded. By 2023, the share of volume in these instruments is about 70\% of total option volume. Throughout this article, we refer to tenors up to 7 calendar days as being \textit{ultra-short-term} tenors.

\begin{figure}[h!]
	\centering
	\includegraphics[width=0.8\textwidth]{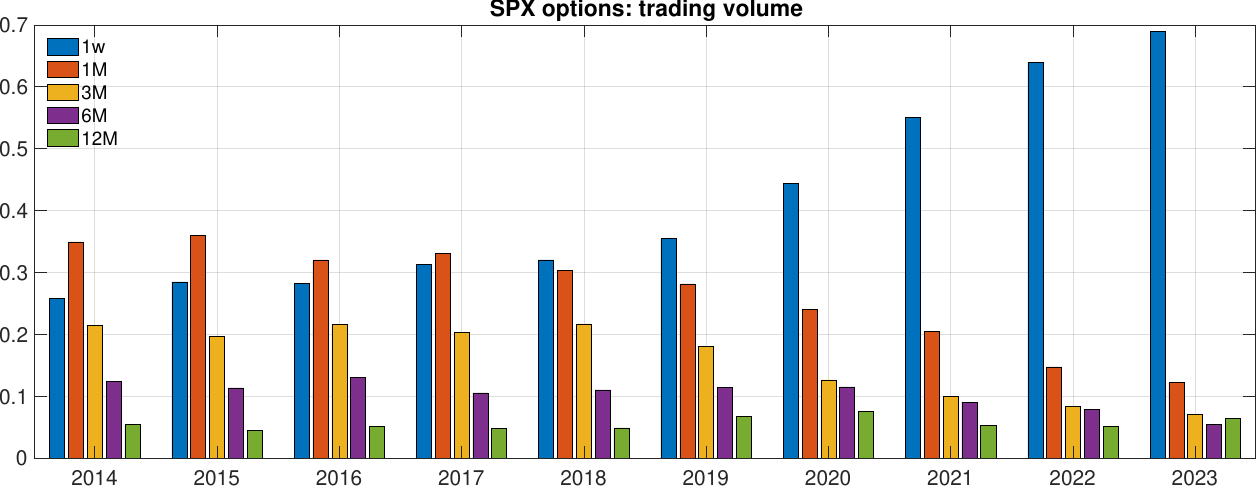}
	\caption{The figure reports volume (i.e., number of contracts traded) in SPX options as a fraction of total volume associated with various tenors: between 0 and 7 days (1w), between 8 and 30 days (1m), between 31 and 60 days (2m), between  61 and 120 days (6m) and between 121 and 365 days (12m). The data covers the period from 2014 to 2023. Data source: OptionMetrics.}
	\label{fig:short_trading}
\end{figure}
On May 11, 2022 contracts with Thursday expiries were introduced. These expiries completed the menu of daily expiries offered by the CBOE. As a result, \textit{ultra-short-term} option trading in 2022 began to include tenors ranging from a few hours (the so-called 0-Days-To-Expiration options or 0DTEs) to 7 seven days (for contracts expiring on the same weekday of the following week). 

A key stylized fact motivates this article: the term structure of \textit{ultra-short-term} at-the-money (ATM) implied volatilities often exhibits pronounced oscillations. Fig.~\ref{fig:iv_term_structure} plots the ATM implied-volatility term structure for the first six consecutive maturities observed on Thursday, August 18, 2022, a typical day in our sample. While the overall shape of the term structure is decreasing, we observe notable sudden changes in its slope associated with maturities between 1 and 2 days and between roughly 4 and 6 days.

\begin{figure}[t!]
\begin{center}
		\includegraphics[width=0.6\linewidth]{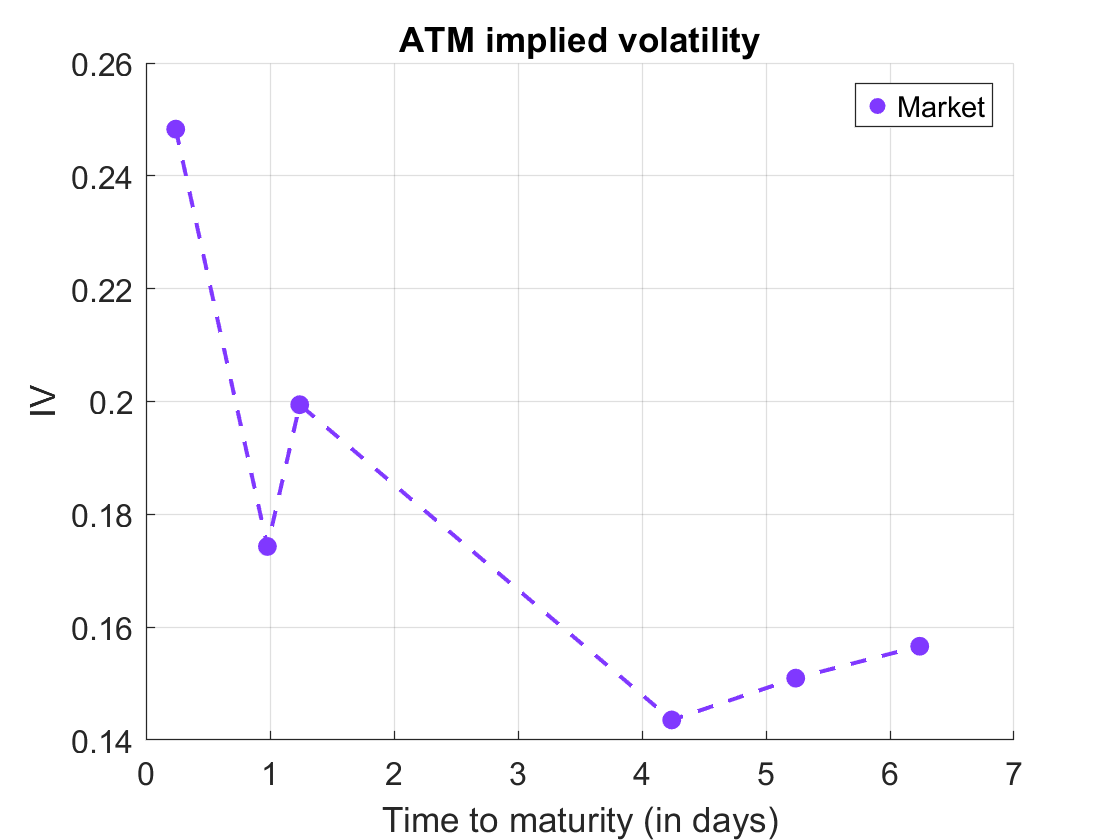}
	\caption{ATM implied-volatility term structure on a typical day: Thursday, August 18, 2022. Data source: CBOE.}
	\label{fig:iv_term_structure}
	\end{center}
\end{figure}
This behavior contrasts with that of longer-maturity options, which typically exhibit considerably smoother ATM implied-volatility shapes. To illustrate, Fig.~\ref{fig:smoothness} visualizes two statistics representing the smoothness of the ATM implied-volatility term structure and its variability: the averages of the absolute values of its first and second derivative (computed from the available ATM implied volatilities with discrete approximations). In the figure, we report information on two segments of the SPX option market: \textit{ultra-short-term} options (i.e., maturities between 0 an 7 calendar days) and options with maturities between 8 and 30 calendar days. The \textit{ultra-short} term structures display noticeably higher variability, and variability of variability, than their longer-dated counterparts. 

The peculiar behavior of the implied-volatility surfaces at very short maturities may arise due to various factors (and their interaction), e.g., heterogeneity in risk premia (c.f. \citealp{AFH:24}), different levels of liquidity reflecting different preferences of market participants for each tenor (retail traders are known to concentrate in 0DTEs, c.f. \citealp{BBG:23}) and institutional features, the joint presence of options maturing at 16:00 (those issued with a weekly tenor) and options maturing at 9:30 (those issued with a monthly tenor), among others.

Regardless of their justification, these features suggest that the \textit{ultra-short-term} ATM implied-volatility term structures are substantially more difficult to replicate using classical option pricing models than their longer-term counterparts. Traditional option pricing models are, in fact, designed for more regular term structures.\footnote{In a 1-factor Heston model, the forward variance can only be monotonically increasing/descreasing. In a 2-factor Heston model, the forward variance can have humps. We will return to this point.} Consistent with this observation, typical option pricing filters, such as the one proposed by \citet{bakshi1997empirical} and used, e.g., in \citet{christoffersen2008option} exclude \textit{ultra-short-term} options explicitly.\footnote{\citet{bakshi1997empirical} write: \enquote{... as options with less than six days to expiration may induce liquidity-related biases, they are excluded from the sample.}} 

\begin{figure}[t!]
\centering
    \begin{tabular}{cc}
    Panel A & Panel B \\.
            \includegraphics[width=0.43\textwidth]{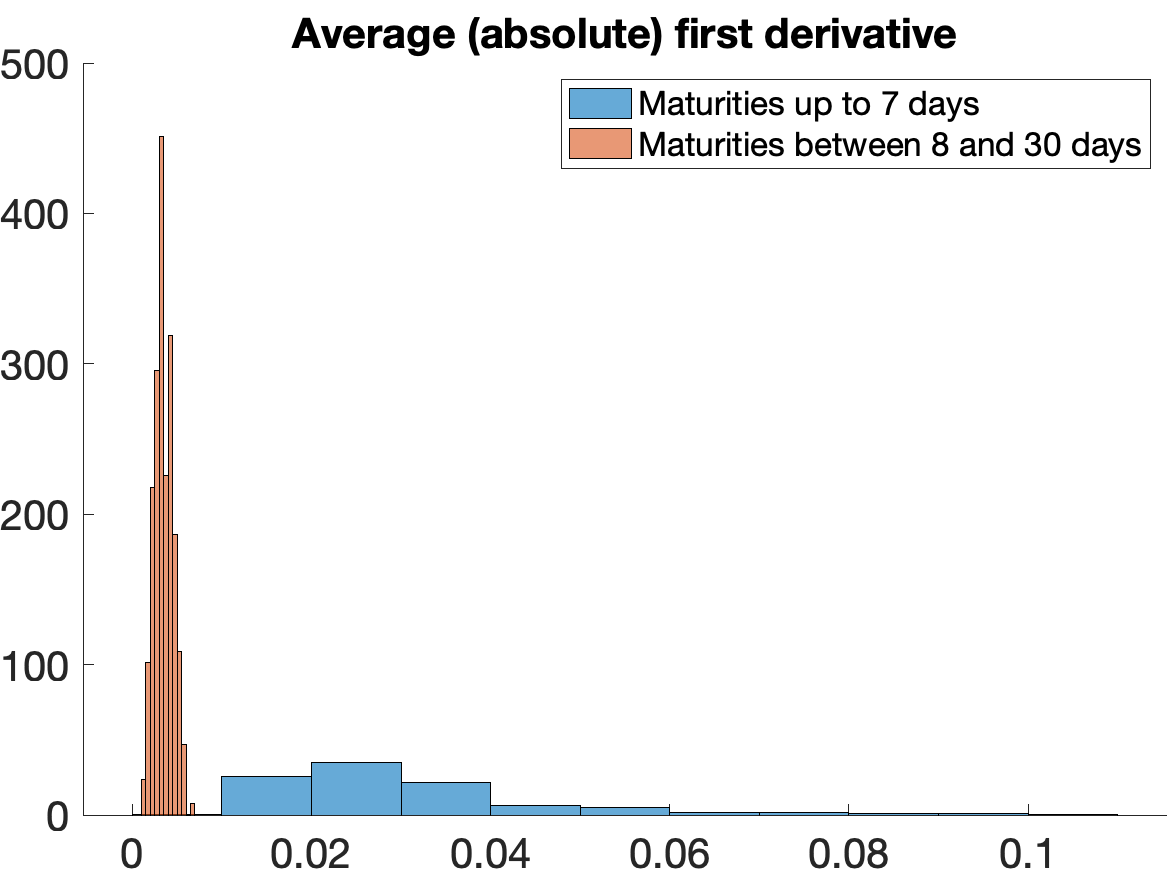}
  &        \includegraphics[width=0.43\textwidth]{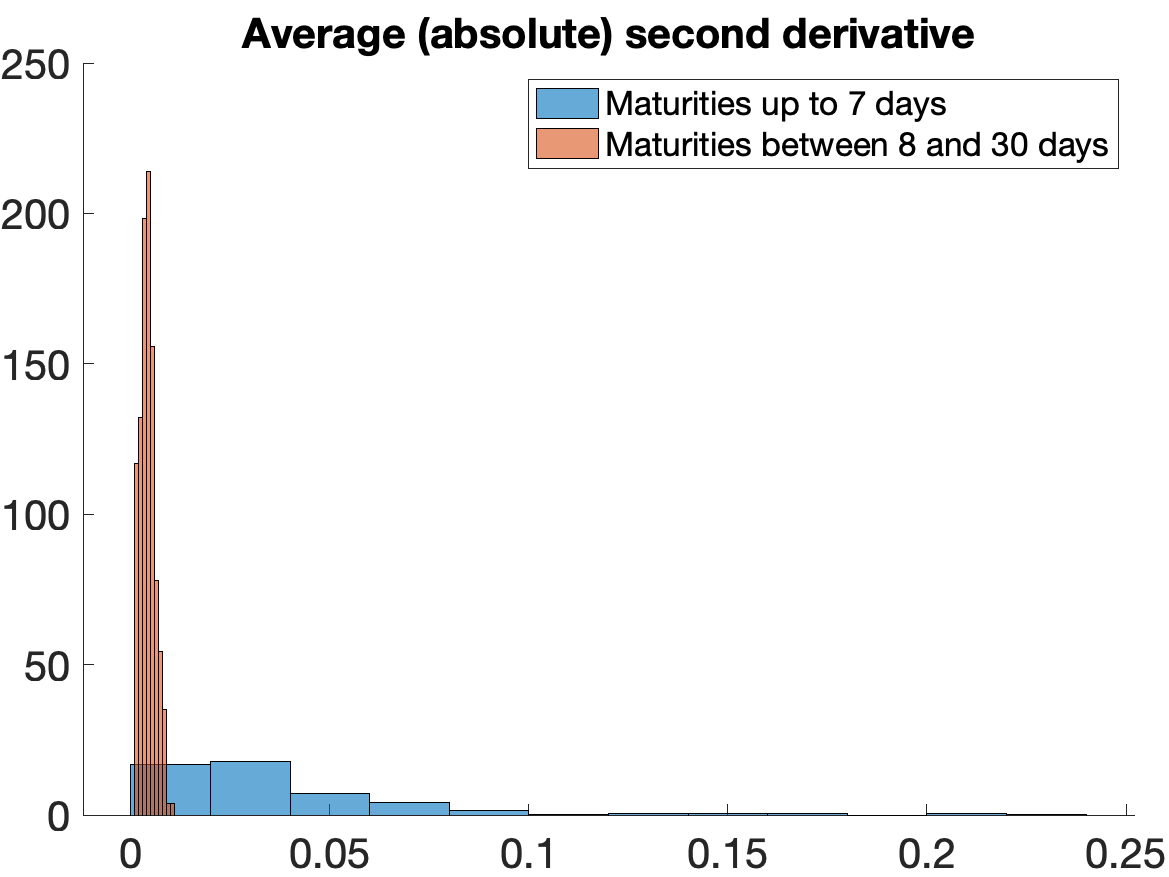}
  \end{tabular}
      \caption{We report the histograms of two summary statistics for the ATM implied-volatility term structures associated with options with expiries between 0 and 7 calendar days (\textit{ultra-short-term} options) and expiries between 8 and 30 calendar days. Data are from May 6, 2022 to May 11, 2023. Data source: CBOE. Panel A: Average absolute value of the first derivative of the ATM implied-volatility term structure, defined as \sloppy\mbox{$\frac{1}{n-1} \sum_{i=1}^{n-1} \left| \frac{\sigma_{i+1} - \sigma_i}{\tau_{i+1} - \tau_i} \right|$}, where \(\tau_1, \tau_2, \ldots, \tau_n\) denote the expiries and \(\sigma_i = \sigma(\tau_i)\) denote the corresponding implied volatilities. Panel B: Average absolute value of the second derivative of the ATM implied-volatility term structure defined using a discrete approximation based on triplets of consecutive implied volatilities.} 
          \label{fig:smoothness}
\end{figure}

Because of current volume (as reported in Fig. \ref{fig:short_trading}), our focus is - precisely - on the \textit{ultra-short-term} segment of the SPX option market. We propose a novel option pricing model explicitly designed to (\textit{i}) calibrate the \textit{ultra-short-term} ATM implied-volatility term structure \textit{across} tenors while (\textit{ii}) pricing the implied-volatility smile for \textit{each} tenor. The model is a continuous-time jump-diffusion specification with two volatility factors. The first is a nonparametric \textit{stochastic} volatility factor allowing for arbitrary dynamics in the \textit{processes} governing the evolution of volatility (e.g., the volatility of volatility) as well as the interaction between volatility and the price process (e.g., leverage), following the approach in \citet{BFR:24}. This stochastic factor is introduced to capture the right level of skew and convexity of the implied-volatility smile for \textit{each} tenor. The second is a \textit{deterministic} volatility factor, representing a shift extension in volatility, a modeling device widely used in the mathematical finance literature since the work of \citet{BM:01} on interest-rate term-structure modeling. In our context, this deterministic factor - also referred to as a \textit{displacement} - is employed to flexibly capture the level of the term structure of forward variances and, therefore, the ATM implied-volatility term structure. We operationalize the model by deriving its characteristic function up to a second-order term in the (square root of the) option’s tenor, thereby extending the characteristic function expansion(s) in \citet{BR:24}, as well as the pricing formulae in \citet{BFR:24}, in order to incorporate a volatility displacement absent in these earlier works. We refer to the resulting model specification as \textit{Edgeworth++}.\footnote{The terminology \enquote{\textit{Edgeworth}} highlights the use of a characteristic function expansion and is consistent with that adopted in \citet{BFR:24}. The suffix \enquote{\textit{++}} is the standard (since \citealp{BM:01}) suffix in the mathematical finance literature denoting the inclusion of a deterministic shift extension.}

In spite of the current size of the \textit{ultra-short-term} option market, the pricing literature on options with maturities below one week is rather sparse. \citet{andersen2017short} price weekly options using a conditional (on spot volatility) Gaussian model with price discontinuities. \citet{A-JLi:25} discuss the different features of short- versus long-maturity surfaces, the shortest maturity in their empirics being the 1-week maturity. \citet{bates2019crashes} is an early paper on 1-day pricing with results extending to a month. As we discuss below, his adopted model has features which make it rather close to one of our benchmarks. The \textit{intra-daily} pricing of 0DTEs is the focus of \citet{BFR:24} who work with a model specification in which important (for \textit{ultra-short-term} pricing) distributional tilts to a Gaussian model are obtained through rich volatility dynamics. \citet{BFR:24}, however, do not consider maturities between 1 and 7 calendar days, our \textit{joint} focus here. In recent work, \citet{bourgey2026quadratic} have proposed a modified version of the quadratic rough Heston model of \citet{gatheral2020quadratic} in order to also fit multiple \textit{ultra-short-term} tenors. Their suggested modification is not a deterministic shift as in this paper, but a polynomial hump in the variance associated with past returns exceeding a certain threshold (see, also, \citealp{gazzani2025pricing}). Differently from \textit{Edgeworth++}, whose Fourier pricing is essentially in closed form, the model in \citet{bourgey2026quadratic} requires simulation to price. As we discuss in Subsection \ref{sec:speed}, computational speed will be a metric used to evaluate model performance. We will, therefore, compare \textit{Edgeworth++} to widely-used model specifications for which Fourier pricing is also feasible.   

The literature on using a deterministic shift extension (as in \citealp{BM:01}) to enhance stochastic volatility models is rather developed, c.f. e.g. \citet{PRS:14,PPR:18} for affine jump-diffusions, \citet{Tiaetal15} for a general class of stochastic volatility models, and \citet{BFG:16}, \citet{ER:19} and \citet{EGR:19} for rough volatility diffusions in which the displacement is introduced in the form of an initial forward variance curve.

Our empirical work benchmarks \textit{Edgeworth++} against two central contributions in the mathematical finance literature (the first) and in the finance literature (the second): a rough Heston-type diffusion with the same shift extension as in this article (\textit{Rough Heston++}) - in the tradition of \citet{Gatheral2018} - and a 2-factor affine model with jumps in returns and volatility (\textit{2F Heston Merton}) - in the tradition of \citet{DPS:00}. We estimate the models over one year of daily option surfaces. \textit{Edgeworth++} attains an average RMSE of about 1 volatility points, improving fit by more than 60\% relative to \textit{Rough Heston++} and by more than 40\% relative to \textit{2F Heston Merton}. 

The underperformance of the \textit{Rough Heston++} model is consistent with - and extends - recent results in \citet{A-JLi:25}, who show that traditional rough volatility models find it hard to fit option prices with maturities between 7 and 30 days. We further document that rough models have issues capturing the \textit{full} implied-volatility smiles across wide moneyness levels over \textit{ultra-short-term} maturities. This outcome stems, in large part, from the absence of a jump component in the return process, a feature which forces a very parsimonious model, like \textit{Rough Heston++}, to trade off the pricing of out-of-the-money (OTM) and ATM options. The underperformance of the \textit{2F Heston Merton} model reflects, instead, the rigidity of its (implied) ATM implied-volatility term structure, which is excessively smooth relative to the sudden fluctuations often observed in the data over short horizons. For both models, the assumed affine structure hinders somewhat their ability to consistently capture the rich implied-volatility smiles documented for \textit{ultra-short-term} options. 

In essence, \textit{Rough Heston++} is able to capture (through the displacement) the variability of the ATM implied-volatility term structures but, for \textit{each} tenor, finds it challenging to consistently deliver the correct implied-volatility smile across a wide range of moneyness. Because of the explicit presence of discontinuities and its large number of parameters, the \textit{2F Heston Merton} model is superior to \textit{Rough Heston++} along the second dimension. It, however, fails along the first due to the smoothness of its implied forward variances.

Consistent with our benchmark \textit{Rough Heston++} model, the core work on rough volatility modeling has generally excluded jumps in prices. Since its inception, in fact, rough modeling has been deliberate about not including discontinuities, which have generally been viewed as being not needed to capture key stylized facts (e.g., \citealp{BFG:16},\footnote{On page 893, they write: \enquote{It has often been claimed that jumps are required to explain the observed extreme short-dated smile in SPX ... It is apparent from ﬁgures 4 and 8 that the rough Bergomi model (where the price process is continuous) generates smiles consistent with those observed empirically even for very short expirations; there is no need for jumps.}} and \citealp{Gatheral2018}\footnote{On page 935, they write: \enquote{This [i.e., the explosive behaviour of the short-term ATM skew induced by roughness] is interesting in and of itself in that it provides a counterexample to the widespread belief that the explosion of the volatility smile as $\tau \rightarrow 0$ (as clearly seen in ﬁgures 1 and 2) implies the presence of jumps.}}). There are, however, recent exceptions (e.g., \citealp{bondi2024rough}). Similarly, the affine literature has been forceful about the role of multiple volatility factors in fitting implied-volatility smiles for each tenor as well as their term structure (c.f. \citealp{christoffersen2009shape}\footnote{On page 1915, they write: \enquote{We demonstrate that two-factor models have much more flexibility in controlling the level and slope of the smirk [then a one-factor model]. An additional advantage is that two-factor models also provide more flexibility to model the volatility term structure.}}). As a consequence, multivariate stochastic volatility models generally do not include a displacement. Again, there are recent exceptions (\citealp{PRS:14,PPR:18}).  

Cognizant of these new developments - and because of our empirical findings reported above - we enrich both benchmark models along their needed dimension: we add jumps to \textit{Rough Heston++} (which yields \textit{Rough Heston Merton++}) and a displacement to \textit{2F Heston Merton}  (which yields \textit{2F Heston Merton++}). The outcome of these enhancements is an equitable comparison across models in which the sole difference is the specification of the volatility dynamics: univariate and nonparametric - but diffusive - in \textit{Edgeworth++}, multivariate and affine - but, again - diffusive in \textit{2F Heston Merton++}, univariate and affine - but rough - in \textit{Rough Heston Merton++}. While - as expected - the performance of both competitors improves considerably, \textit{Edgeworth++} continues to have an advantage in capturing rich implied-volatility smiles for each tenor as well as complex term structures of ATM implied volatilities across tenors. Importantly, this conclusion is reached even though the constant (across time) parameters of the competing models are re-estimated across days. We recall, in fact, that the stochastic dynamics of \textit{Edgeworth++} are driven by \textit{processes}. In this sense, \textit{Edgeworth++}'s daily estimates are logically consistent with the model's dynamics. The same cannot be said in the case of other model specifications for which daily estimation amounts to daily re-evaluation of constant parameters with the objective of optimizing performance.   

We conclude this Introduction by emphasizing that, while the traditional way to compute the implied volatility is to use calendar-time sampling, there has been recent interest in using business-time sampling.\footnote{Given total variance over the period ($\sigma^2 \tau$), the choice of $\tau$ necessarily affects the computation of $\sigma^2$ and, therefore, of the implied volatility and its term structure. The CBOE has historically used calendar-time sampling to compute the classical implied volatility (over 30 days) as well as its shorter (9 days) and longer (3, 6 and 12 months) counterparts. The 1-day VIX, launched on April 23 2023, has been using business-time sampling since its inception.} In this article, we do not pre-filter data in order to determine business-time sampling - and the corresponding implied volatilities - but use the typical calendar-time sampling. Our methods would however apply to business-time sampling and the resulting implied-volatility term structures. We leave business-time sampling - and the timing choices that it entails - as an extension for future work.    

The remainder of the paper is structured as follows. Section \ref{sec:model} presents \textit{Edgeworth++}. Section \ref{sec:data} discusses the data. In Section \ref{benchmarks} we introduce the benchmark models. Section \ref{sec:numerical_results} discusses relative performance. We compare \textit{Edgeworth++} to the benchmarks on three grounds: computational speed, pricing the shortest tenor, and pricing the \textit{ultra-short-term} ATM implied-volatility term structure along with individual smiles. Section \ref{impbench} is about the comparison of \textit{Edgeworth++} with enriched benchmarks designed to address the limitations documented in Section \ref{sec:numerical_results}. Section \ref{spot} uses spot volatility estimation as an additional metric to evaluate model performance.  Section \ref{sec:conclusions} provides concluding remarks.  Appendix \ref{sec:appendix} contains proofs.

\section{The pricing model}\label{sec:model}

We assume the logarithmic price $X_t$ is a real stochastic process defined on the filtered probability space $(\Omega, \mathcal F, (\mathcal F_t)_{t\ge 0}, \mathbb{Q})$, where $\mathbb{Q}$ is the risk-neutral measure, whose dynamics are specified as follows: 
\begin{eqnarray}\label{dyn}
X_{t} &=& \underbrace{X_0 +\int_0^t\mu_{s} ds + \int_0^t\sigma_s dW_s}_{=:X^c} + \underbrace{\int_0^t x_s dN_s}_{=:X^J}, \\ \nonumber
\sigma_t &=& \sigma_0 + \phi(t) +\int_0^t \alpha_s ds +\int_0^t\beta_s d W_s +\int_0^t\beta_s^\prime d W_s^\prime, \\ \nonumber
\mu_{t} &=& \mu_0+\int_0^t \gamma_{s}ds + \int_0^t\delta_{s}dW_s+ \int_0^t\delta^{\prime}_{s}dW^*_s, \notag \\
\beta_{t} &=& \beta_0+\int_0^t\zeta_{s}ds + \int_0^t\eta_{s}dW_s+ \int_0^t\eta^{\prime}_{s}dW^{**}_s, \notag
\end{eqnarray}
where $\mu_t = r_t - \frac{1}{2}\sigma^2_t$ with $r_t$ defining the risk-free rate. The shift extension or displacement $\phi(t)$ is a deterministic function such that $\phi(0)=0$ and $\sigma_t>0$, almost surely. The quantities $W_t$, $W_t^{\prime}$, $W_t^*$ and $W_t^{**}$ are independent Brownian motions. $N_t$ is a Poisson process, independent of all Brownian motions, with random intensity $\lambda_t,$ and $x_t$ is a random jump size with time-varying moments. When pricing 0DTEs, \cite{BFR:24} show that a model with jumps in volatility (possibly correlated with jumps in prices) would be virtually indistinguishable from a model with jumps in prices only. Relying on their work, we dispense with volatility jumps in pricing \textit{ultra-short-term} surfaces. We assume Gaussian jump sizes $x_t$, but extensions to richer families for which the characteristic function is known would be immediate (c.f., e.g., \citealp{andersen2017short}).

In order to separate the volatility of volatility from leverage, we adopt the classical \cite{H:93} specification and suitably re-label quantities:
\begin{equation}\label{eqz:heston_convention}
    \beta_t = \widetilde{\beta}_t \rho_t,\qquad\qquad\qquad \beta'_t = \widetilde{\beta}_t \sqrt{1-\rho^2_t}.
\end{equation}
As a result, $\widetilde{\beta}_t$ now represents the spot volatility of volatility and $\rho_t$ is spot leverage. 

When $\phi(t)=0$, the process in Eq. (\ref{dyn}) belongs to the family of processes considered in \citet{BR:24} and used by \cite{BFR:24} to price 0DTEs. While it is unconventional to describe the dynamics of $\mu_t$ and $\beta_t,$ in addition to those of the typical objects of interest (namely, price levels $X_t$ and volatility $\sigma_t$), we do so in light of the fact that some of the corresponding processes (e.g., $\delta_t$) will play a role in the characteristic function expansion(s). 

Next, we extend the approach in \cite{BFR:24} by accounting for the displacement $\phi(t)$ which, as shown below, will be important for pricing \textit{beyond} the shortest tenors considered in \cite{BFR:24}.

\subsection{Characteristic function expansion with displacement}

Let $X^c$ satisfy Eq. \eqref{dyn}. Define the demeaned and standardized \textit{continuous} logarithmic price increments as
\begin{equation}\label{zeta}
    Z^c_{\tau}:=\frac{X^c_\tau-X^c_0-\mu_0\tau}{\sigma_0\sqrt{\tau}}.\qquad
\end{equation}
By construction, we have 
\begin{equation*}
    Z^c_{\tau}=\frac{W_\tau}{\sqrt{\tau}}+\sqrt{\tau}Y_\tau,
\end{equation*}
where $$Y_{\tau}=\frac{1}{\sigma_0\tau}\left(\int_0^\tau(\mu_s-\mu_0)ds+\int_0^\tau(\sigma_s-\sigma_0)d W_s\right).$$
Below, we provide an expansion (up to the second order in $\sqrt{\tau}$) of the $\mathbb{Q}$-characteristic function of $Z^c_{\tau}$. Whenever otherwise specified, $\mathbb{Q}$-expectations are taken with respect to $t=0$ information. 

\begin{thm}\label{th:expansion}
    Let $Z_{\tau}^c$ be defined as in Eq. (\ref{zeta}). Assume $X^c$, defined in Eq. \eqref{dyn}, to be $D_W^{(5)}$.\footnote{This is a technical, but innocuous, assumption of differentiability discussed in the Appendix.} Then, for $\tau>0,$ we have
    \begin{equation*}
    \E^{\mathbb{Q}}[e^{iu Z^c_\tau}] = \Psi^c(u,\tau) + O(\tau^{3/2})\psi(u),
    \end{equation*}
    where $\psi(u)$ is an integrable function over $\mathbb{R}$ of order $u^{-3}$, as $|u|\rightarrow\infty$, and
\begin{align*}
\Psi^c(u,\tau) &= \exp\left( 
    -\frac{u^2}{2} \frac{ \int_0^{\tau} \tilde{\phi}^2(s) \, ds }{\tau} 
\right)  \left(
    1 
    - iu^3 \frac{\beta_0}{\sigma_0 \tau^{3/2}} \int_0^{\tau} \tilde{\phi}(s) \left( \int_0^s \tilde{\phi}(s_1) \, ds_1 \right) ds \right. \\
&\quad 
    - u^2 \frac{\delta_0}{\sigma_0 \tau} \int_0^{\tau} \left( \int_0^s \tilde{\phi}(s_1) \, ds_1 \right) ds 
    - u^2  \frac{\alpha_0}{\sigma_0 \tau} \int_0^{\tau} s \tilde{\phi}(s) \, ds \\
&\quad 
    + u^4  \frac{\eta_0}{\sigma_0 \tau^2} \int_0^{\tau} \int_0^s \left( \int_0^{s_1} \tilde{\phi}(s_2) \, ds_2 \right) ds_1 \, ds \\
&\quad 
    - \frac{\beta_0^2 u^2}{8 \sigma_0^2 \tau} \left[
        2\tau^2 
        + 4u^2 \int_0^{\tau} \int_0^s \tilde{\phi}(s_1) \, ds_1 \, \tilde{\phi}(s) \, ds \right. \\
&\qquad\qquad
        - \frac{24u^2}{\tau} \int_0^{\tau} \int_0^s \int_0^{s_1} \tilde{\phi}(s_2) \, ds_2 \, \tilde{\phi}(s_1) \, ds_1 \, \tilde{\phi}(s) \, ds \\
&\qquad\qquad
        \left. - \frac{12u^2}{\tau} \int_0^{\tau} \int_0^s \left( s_1 - \frac{2u^2}{\tau} \int_0^{s_1} \int_0^{s_2} \tilde{\phi}(s_3) \, ds_3 \, \tilde{\phi}(s_2) \, ds_2 \right) \tilde{\phi}(s_1) \, ds_1 \, \tilde{\phi}(s) \, ds
    \right] \\
&\quad 
    \left. - \frac{(\beta_0')^2 u^2}{2 \sigma_0^2 \tau} \left( \frac{\tau^2}{2} - \frac{2u^2}{\tau} \int_0^{\tau} \int_0^s s_1 \tilde{\phi}(s_1) \, ds_1 \, \tilde{\phi}(s) \, ds \right)
\right),
\end{align*}
    with $\tilde{\phi}(s)=1+\phi(s)/\sigma_0,$ for all $s\in\R$.
\end{thm}
\begin{proof}
    See Appendix \ref{appendix:expansion_proof}.
\end{proof}
Theorem \ref{th:expansion} illustrates the impact of the displacement $\phi(t)$ on the expansion of the $\mathbb{Q}$-characteristic function of $Z^c_t$ and, consequently, on the higher $\mathbb{Q}$-moments of continuous (demeaned, standardized) logarithmic returns (relative to the standard Gaussian moments). The following Corollary recalls the mapping between the characteristic function expansion and higher-order moments for the case without displacement ($\phi(t)= 0$), a result which is provided in \citet{BR:24} and is employed in \cite{BFR:24} to price options with only one tenor, i.e., the intra-daily tenor associated with 0DTEs. 
\begin{cor}[to Theorem \ref{th:expansion}]\label{cor:expansion}
    Let $Z_{\tau}^c$ be defined as in Eq. (\ref{zeta}). Assume $X^c$, defined in Eq. \eqref{dyn}, to be $D_W^{(5)}$. Assume, also, $\phi(t)= 0$. Then, for $\tau>0$, we have
\begin{eqnarray*}
\E^{\mathbb{Q}}[e^{iu Z^c_\tau}] &=& e^{- \frac{u^2}{2}}
\left( 1 \underbrace{- iu^3\frac{\widetilde{\beta}_0 \rho_0}{2\sigma_0} \sqrt{\tau}}_{\textrm{third moment adjustment}} \right. \notag \\
&& \left. \underbrace{ -u^2\left(\frac{\alpha_0+ \delta_0}{2\sigma_0} + \frac{\widetilde{\beta}_0^2} {4\sigma_0^2}\right)\tau+\frac{1}{24} \frac{\widetilde{\beta}_0^2} {\sigma_0^2} u^2 \left(4u^2 -\rho_0^2u^2\tonde{3u^2-8}\right)\tau +\frac{\eta_0}{6\sigma_0} u^4\tau}_{\textrm{second, fourth and sixth moment adjustment}} \right) + O(\tau^{3/2}) \psi(u).
\end{eqnarray*}

\end{cor}
\begin{proof}
    The result follows from direct computations, given Theorem \ref{th:expansion}.
\end{proof}

\subsection{Piecewise deterministic displacement}
Let $\phi(t):\R\to\R$ be such that $\phi(0)=0$. We assume that $\phi(t)$ is a \emph{simple function}, i.e., a finite linear combination of indicator functions on measurable sets. In the literature, this choice is generally associated with the modeling of forward variance (c.f. e.g., \citealp{B:05,B:15}). 

Let $\{\tau_0, \tau_1,\dots,\tau_{n}\}$ be a set of increasing tenors. Then,
\begin{equation}\label{eqz:phi_piecewise}
\phi(t)=\sum_{k=0}^{n-1}a_{k}1_{[\tau_k,\tau_{k+1})}(t),
\end{equation}
where $\tau_0 = a_0 = 0,$ and $a_k \in \R$ for all $k=0,\dots, n-1$.
Observe that options with tenors up to $\tau_1$ are priced as in Corollary \ref{cor:expansion} (since $a_0 = 1$), while longer maturity options (within the \textit{ultra-short-term} segment) require one extra parameter $a_k$ for each additional expiry. 

To simplify notation, we introduce the following $m\times n$ matrices (given $\sigma_0>0$):
\begin{align}\label{eqz:delta_phi}
   {\Delta \Tcal}^{(n)}:= \begin{pmatrix}
        \tau_{1} & \cdots & \tau_{n}-\tau_{n-1}\\
        \vdots & \ddots& \vdots\\
        \tau_1^m & \dots & \tau_{n}^m-\tau_{n-1}^m
    \end{pmatrix},\qquad
    \tilde{{\Phi}}^{(n)}:= \begin{pmatrix}
        1 & 1+\frac{a_{1}}{\sigma_0} &\cdots & 1+\frac{a_{n-1}}{\sigma_0}\\
        \vdots &  \vdots& \ddots & \vdots\\
        1 &\left(1+\frac{a_{1}}{\sigma_0}\right)^m & \dots & \left(1+\frac{a_{n-1}}{\sigma_0}\right)^m
    \end{pmatrix}.
\end{align}
For any pair $(j,k)$ with $j=1,\dots,m$ and $k=0,\dots,n-1$, we may write the generic $(j,k)$-element as
\begin{equation*}
    \Delta \Tcal_{j,k}^{(n)}  = \tau_{k+1}^j-\tau_k^j, \qquad \tilde{\Phi}_{j,k}^{(n)}= \left( 1 + \frac{a_k}{\sigma_0}\right)^j,
\end{equation*}
and define, for each $j$ with $j=1,\dots, m$ (given $n>1$), the $n$-vectors
\begin{eqnarray}\label{a}
\Delta \Tcal_{j}^{(n)} = (\Delta\Tcal_{j,0}^{(n)},\ldots,\Delta\Tcal_{j,n-1}^{(n)})^{\top}
\end{eqnarray}
 and 
\begin{eqnarray}\label{b} 
\tilde{{\Phi}}_j^{(n)} = ( \tilde{{\Phi}}_{j,0}^{(n)},\ldots, \tilde{{\Phi}}_{j,n-1}^{(n)})^{\top}.
\end{eqnarray}
Corollary \ref{cor:expansion_piecewise}, below, provides a closed-form expression for the characteristic function expansion $\Psi^c(u,\tau)$ of $Z^c_{\tau}$ given the assumed displacement in Eq. (\ref{eqz:phi_piecewise}). The expression in the Corollary will be used in implementations. The symbol $\langle \cdot ,\cdot\rangle$ denotes the inner product between two vectors in $\mathbb{R}^n$.    

\begin{cor}[to Theorem \ref{th:expansion}]\label{cor:expansion_piecewise}
    Let $Z_{\tau}^c$ be defined as in Eq. (\ref{zeta}). Assume $X^c$, defined in Eq. \eqref{dyn}, to be $D_W^{(5)}$. Let the displacement be as in Eq. (\ref{eqz:phi_piecewise}). Consider a set $\{\tau_0,\dots,\tau_n\}$ of tenors and express the vectors $\Delta\Tcal_j^{(k)}$ and $\tilde{\Phi}_j^{(k)}$, for each $k=1,\dots,n,$ as in Eqs. \eqref{a} and \eqref{b}. Then, for each $\tau_k \in \{\tau_0,\dots,\tau_n\}$, we have
     \begin{equation*}
    \E^{\mathbb{Q}}[e^{iu Z^c_{\tau_k}}] = \Psi^c(u,\tau_k) + O(\tau_k^{3/2})\psi(u),
    \end{equation*}
where
\begin{align*}
\Psi^c(u,\tau_k) 
&= \exp\left(
    -\frac{u^2}{2} \frac{\langle \tilde{\Phi}_2^{(k)}, \Delta \mathcal{T}_1^{(k)} \rangle}{\tau_k}
\right)  \left(
    1 
    - iu^3  \frac{\tilde{\beta}_0 \rho_0}{2 \sigma_0 \tau_k^{3/2}} \langle \tilde{\Phi}_2^{(k)}, \Delta \mathcal{T}_2^{(k)} \rangle \right. \\
&\quad 
    - u^2  \frac{\alpha_0 + \delta_0}{2 \sigma_0 \tau_k} \langle \tilde{\Phi}_1^{(k)}, \Delta \mathcal{T}_2^{(k)} \rangle
    + u^4  \frac{\eta_0}{6 \sigma_0 \tau_k^2} \langle \tilde{\Phi}_1^{(k)}, \Delta \mathcal{T}_3^{(k)} \rangle \\
&\quad 
    - \frac{\tilde{\beta}_0^2 \rho_0^2}{8 \sigma_0^2 \tau_k} u^2 \left(
        2\tau_k^2
        + 2u^2 \langle \tilde{\Phi}_2^{(k)}, \Delta \mathcal{T}_2^{(k)} \rangle
        - \frac{4u^2}{\tau_k} \langle \tilde{\Phi}_2^{(k)}, \Delta \mathcal{T}_3^{(k)} \rangle \right. \\
&\qquad\qquad
        \left.\left.
        - \frac{12u^2}{\tau_k} \left(
            \frac{1}{6} \langle \tilde{\Phi}_2^{(k)}, \Delta \mathcal{T}_3^{(k)} \rangle
            - \frac{u^2}{12 \tau_k} \langle \tilde{\Phi}_4^{(k)}, \Delta \mathcal{T}_4^{(k)} \rangle
        \right)
    \right) \right. \\
&\quad 
    \left. - \frac{\tilde{\beta}_0^2 (1 - \rho_0^2)}{2 \sigma_0^2 \tau_k} u^2 \left(
        \frac{\tau_k^2}{2}
        - \frac{u^2}{3 \tau_k} \langle \tilde{\Phi}_2^{(k)}, \Delta \mathcal{T}_3^{(k)} \rangle
    \right)
\right).
\end{align*}

\end{cor}
\begin{proof}
    See Appendix \ref{appendix:piecewise_shift}.
\end{proof}

Next, we turn to pricing.

\subsection{Edgeworth++}

Pricing requires the characteristic function of the (standardized, de-meaned) full process, inclusive of price jumps. Because of independence between $X^c$ and $X^J$, the full characteristic function conveniently factorizes and can be written as
\begin{equation}
\Psi(u,\tau) =  \Psi^c(u,\tau) \Psi^J(u,\tau).
\end{equation}
Any jump model with known characteristic function $\Psi^J(u,\tau)$ may, of course, be utilized. For simplicity, we employ Gaussian jump sizes with mean $\mu_{J,t}$, variance $\sigma_{J,t}^2$ and intensity $\lambda_t$ and write
\begin{equation*}
\Psi^{J}(u,\tau) = \exp\left\{\tau \lambda_t( e^{i u \frac{\mu_{J,t}}{\sigma_t \sqrt{\tau}} - u^2 \frac{\sigma_{J,t}^2}{2\sigma^2_t \tau}} - 1-iu \bar{\mu}_{J,t} )\right\} \qquad  \bar{\mu}_{J,t}=\exp\left\{\frac{\mu_{J,t}}{\sigma_t \sqrt{\tau}}+\frac{1}{2}\frac{\sigma_{J,t}}{\sigma^2_t \tau}^2\right\}-1,
\end{equation*}
where $\bar{\mu}_{J,t}$ is the $\mathbb{Q}$-compensator.\footnote{With constant parameters, $\Psi^{J}(u,\tau)$ is the exact characteristic function of $X^J$. With time-varying parameters, it can be interpreted as an expansion to the order $\tau$ \citep{BR:24}.}

Given $\Psi(u,\tau),$ Fourier inversion yields prices. The time-0 price of a call option with strike $K$ and maturity $\tau$ is given by:\footnote{Puts are priced analogously or by put/call parity.}
\begin{align}
C_{\tau, K}  = &~e^{X_0} \left(\frac 12 + \frac{1}{\pi}\int_0^{\infty}\Re\left[\frac{e^{iud_{2,\tau}}\Psi(u-i\sigma_t\sqrt{\tau}, \tau)}{iu\Psi(-i\sigma_t\sqrt{\tau},\tau)}\right]du\right) \notag \\
&- Ke^{-r_0\tau} \left(\frac 12 + \frac{1}{\pi}\int_0^{\infty}\Re\left[\frac{e^{iu d_{2,\tau}}\Psi(u,\tau)}{iu}\right]du\right), 
\end{align}
where
$$d_{2,\tau} = \frac{X_0 -\log(K) +(r_0-\frac 12 \sigma_0^2)\tau}{\sigma_0\sqrt{\tau}}.$$

As mentioned, we call the pricing model \textit{Edgeworth++}. $\Psi(u,\tau)$ may be viewed as the expansion (to the second order in $\sqrt{\tau}$) of the characteristic function of a semi-parametric model in which the jump component is parametric and the continuous component is fully nonparametric. The resulting prices are, essentially, closed-form approximations (in the same spirit as pricing via approximations in \citealp{LPP:17}, \textit{inter alia}). 

We note that, in the original formulation of the deterministic shift extension (\citealp{BM:01}), the extension is calibrated to (i.e., implied from) observed data (the term structure of interest rates).\footnote{In some instances, the calibration approach can be extended to forward variance using variance displacements along with the \citet{CM:98} replication formula (c.f. \citealp{B:15}).} In our case, the calibration approach holds, e.g., for suitable sub-cases like the displaced Black-Scholes model. Remark \ref{sec:bsdisplaced} provides details. For our general model, the displacement is, instead, a component of the parameter vector to be estimated, as in \citet{PRS:14,PPR:18}. 

For each time $t$, the estimation of \textit{Edgeworth++} results in the estimation of the parameter vector $\Theta_t = (\sigma_t,\widetilde{\beta}_t,\rho_t,\eta_t,\alpha'_t,\lambda_t,\mu_{J,t},\sigma_{J,t},a_1,\ldots,a_{n-1}),$ where
$$\alpha'_t = \frac{\alpha_t+\delta_t}{2}$$
and $n$ is the number of tenors on the time-$t$ implied-volatility surface.\footnote{We set $r_t =0$.} Thus, \textit{Edgeworth++}  has $7+n$ parameters.   

\begin{remark}[Displaced Black and Scholes model]\label{sec:bsdisplaced}
Let $\phi(t)$ satisfy Eq. \eqref{eqz:phi_piecewise}. Consider the displaced Black and Scholes model (BS++) in which the logarithmic price dynamics are given by
\begin{equation*}
    d X_t = -\frac{1}{2}(\sigma_0+\phi(t))^2 dt + (\sigma_0+\phi(t)) dW_t,
\end{equation*}
where $\sigma_0>0$ is spot volatility. Assume a set $\{\tau_0, \tau_1,\dots,\tau_n\}$ of increasing expiries. It is straightforward to show that
\begin{align*}
   \E^{\mathbb{Q}}[e^{iu (X_{\tau_n}-X_0)}]&=\exp\left\{-\frac{iu}{2}\sigma_0^2 \langle \tilde{\Phi}_2^{(n)},\Delta \Tcal_{1}^{(n)}\rangle\right\}\exp\left\{-\frac{u^2}{2}\sigma_0^2 \langle \tilde{\Phi}_2^{(n)},\Delta \Tcal_{1}^{(n)}\rangle\right\}.
\end{align*}
Thus, the BS++ ATM implied volatility is given by
\begin{align}\label{eqz:iv_bs++}
    \sigma_{\text{BS}}(\tau):=\begin{cases}
    \sigma_0, \qquad &\text{if}\ \tau < \tau_1,\\
    \sqrt{\frac{\sigma_0^2 \langle\tilde{\Phi}_2^{(k)},\Delta \Tcal_1^{(k)}\rangle}{\tau_{k}}}, \qquad &\text{if}\ \tau_{k-1}\le \tau < \tau_k.
    \end{cases}
\end{align}
Given the market ATM implied volatilities, i.e. $(\sigma_{\text{BS}(\tau_i)})_{i=1}^{n},$  we can solve recursively:
\begin{align*}
    \sigma_0&=\sigma_{\text{BS}}(\tau_1)\\
a_{k}&=-\sigma_0 + \sqrt{\frac{\sigma_{\text{BS}}^2(\tau_{k+1})\tau_{k+1} -\sigma_{\text{BS}}^2(\tau_{k})\tau_{k}}{\Delta_{k+1}\tau}},
\end{align*}
 for $k=1,\dots,n-1$, where $\Delta_{k}\tau:=\tau_{k}-\tau_{k-1}$ and $a_0=\tau_0=0$. We conclude that, in this specific case, one can calibrate the shift extension's parameters directly from market ATM volatilities (given absence of calendar-time arbitrage, i.e., if $\sigma_{\text{BS}}^2(\tau_{k+1})\tau_{k+1} -\sigma_{\text{BS}}^2(\tau_{k})\tau_{k}  \geq 0$).\footnote{The solution of the recursive system for the $a_k$ coefficients yields a realistic set of starting values for spot volatility and the shift extension's parameters. We use the starting values to optimize \textit{Edgeworth++}.} Conversely, setting the coefficients $a_k$ as being real implies absence of calendar-time arbitrage as a property of the model.
\end{remark}

\begin{remark}[Interpreting the displacement] By a simple application of It\^o's lemma, the conditional $\mathbb{Q}$-expectation of the spot variance process, i.e., the spot forward variance, can be written as follows:
\begin{eqnarray}\label{risk}
\mathbb{E}^\mathbb{Q}[\sigma^2_{\tau_q}] &=& \sigma_0^2 + 2\int_0^{\tau_q} \mathbb{E}^\mathbb{Q}[\sigma_{s-}]d\phi(s) + \int_0^{\tau_q} \left(d\phi(s)\right)^2+2\int_0^{\tau_q} \mathbb{E}^\mathbb{Q}[\sigma_s \alpha_s]ds + \int_0^{\tau_q} \mathbb{E}^\mathbb{Q}[\widetilde{\beta}^2_s]ds \nonumber \\
&=& \sigma_0^2 + 2\sum_{j=1}^{q} \mathbb{E}^\mathbb{Q}[\sigma_{\tau_{j-}}](a_j - a_{j-1}) + \sum_{j=1}^{q}(a_j - a_{j-1})^2+ 2\int_0^{\tau_q} \mathbb{E}^\mathbb{Q}[\sigma_s \alpha_s]ds + \int_0^{\tau_q} \mathbb{E}^\mathbb{Q}[\widetilde{\beta}^2_s]ds, \nonumber \\
\end{eqnarray}
since $d\phi(t) = \sum_{k=1}^{n-1} (a_k - a_{k-1})\delta(t - \tau_k)dt$ and $\left(d\phi(t)\right)^2 = \sum_{k=1}^{n-1} (a_k - a_{k-1})^2\delta(t - \tau_k)dt,$ where $\delta(.)$ is the dirac delta function. Eq. (\ref{risk}) clarifies that the shifts (which are jumps with deterministic sizes at deterministic times) make up a portion of the (forward) variance risk premium and each tenor will, in general, contribute to it. 

Next, notice that the square of each ATM implied volatility is well-approximated by the average forward variance. Thus,
\begin{eqnarray*}
\sigma^2_{\text{BS}}(\tau_q,S_0) &\sim & \frac{1}{\tau_q}\int_0^{\tau_q} \mathbb{E}^\mathbb{Q}[\sigma^2_{s}]ds \\
& = & \sigma_0^2 + \frac{1}{\tau_q} \sum_{j=1}^{q-1} \left(2\mathbb{E}^\mathbb{Q}[\sigma_{\tau_{j-}}](a_j - a_{j-1})+(a_j - a_{j-1})^2\right)(\tau_q - \tau_j) \nonumber \\
&+ &\frac{2}{\tau_q}\int_0^{\tau_q}  (\tau_q - u)\mathbb{E}^\mathbb{Q}[\sigma_u \alpha_u]du + \frac{1}{\tau_q}\int_0^{\tau_q} (\tau_q - u)\mathbb{E}^\mathbb{Q}[\widetilde{\beta}^2_u]du,
\end{eqnarray*} 
which shows the impact of the displacement on the ATM implied-volatility term structure (for a changing $q$). Short-tenor shifts stay longer in the system and, therefore, receive a higher weight: given $q$, the lower $j$, the larger $(\tau_q - \tau_j)$.
\end{remark}

\section{Data}\label{sec:data}

Our dataset consists of a one-year window of trading days, from May 6, 2022 to May 11, 2023. May 6, 2022 is the first day, a Friday, on which options were traded daily with expiries a week later. 

The data is obtained from the CBOE via the Datashop online platform.\footnote{\url{https://datashop.cboe.com}.} Specifically, we retrieve options' best bid and ask quotes at the 1-minute frequency between the opening and the closing of each trading day. On each day, we sample the data at 10:30 and conduct inference at that time. The choice of 10:30 reflects a trade-off between ensuring a sufficiently high mean volume and limiting volume variability, as excessive variability could compromise the reliability of the estimates (c.f. \citealp{BFR:24}). On each day, we retain the options with the $6$ shortest tenors in order to construct the ATM implied-volatility term structure for that day. This strategy leaves us with a total of just $3.65\%$ of options with time to maturity larger than $7$ days. 
On average, an implied-volatility surface comprises 724 options. Fig.~\ref{fig:data_section1} shows how the number of options in the surface varies over time by expiry.
We apply commonly used filters to clean the option data. Specifically, define log-moneyness $m$ as
\begin{equation}\label{eq:m} 
m = \frac{\log(K)-\log(F)}{\sigma_{BS}\sqrt{\tau}},
\end{equation}
where $\tau$ is the time to maturity, $K$ is the option strike, $F$ is the implied forward price, and $\sigma_{BS}$ is the ATM implied volatility (i.e., the Black-Scholes implied volatility of the option with tenor $\tau$ and strike closer to the implied forward price). We keep options with log-moneyness $-15<m<5,$ a sizable moneyness interval. Second, we eliminate illiquid, or non-traded, contracts by retaining only options with strictly positive bid and ask quotes. Third, we remove days on which Federal Open Market Committee (FOMC) meetings occurred. As documented in the literature (c.f., e.g., \citealp{DJKS:19}, and \citealp{JKS:23}), announcement risk on FOMC days would require the introduction of an additional jump factor realized at deterministic times, something that could be accommodated in our framework but is beyond the objectives of the current paper. Application of these filters results in a final sample of $247$ \textit{ultra-short-term} volatility surfaces to be priced.
 
\begin{figure}[H]
	\begin{center}
		\includegraphics[scale =0.53]{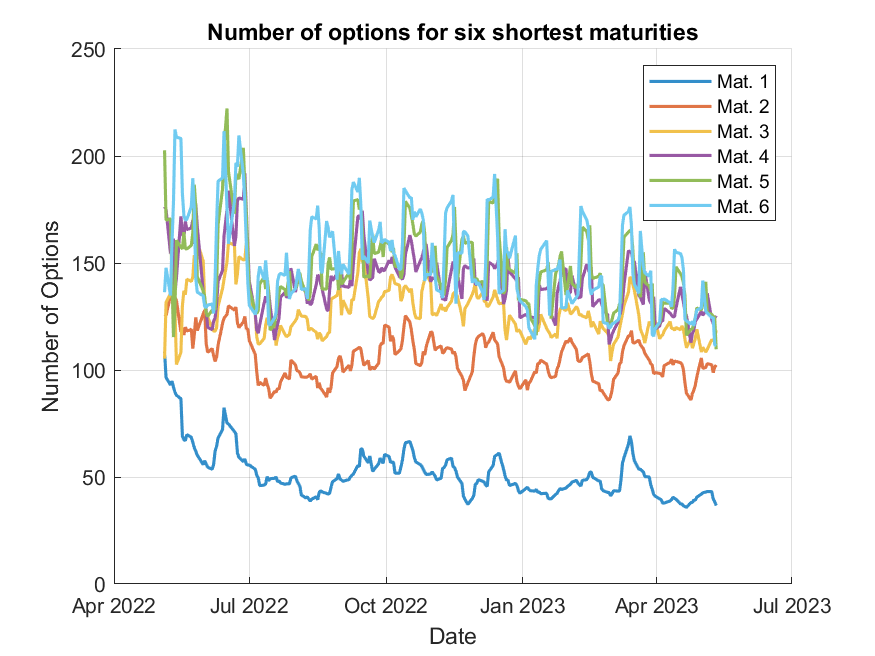}
		\caption{ Number of options in each daily surface for each of the $6$ maturities. The data covers the period May 6, 2022 to May 11, 2023. Data source: CBOE.}
		\label{fig:data_section1}
	\end{center}
\end{figure}

\section{Benchmarks}\label{benchmarks}

We compare the performance of \textit{Edgeworth++} to that of two models (a \textit{Rough Heston++} model and a \textit{2-factor Heston Merton} model) widely used both in the academic literature and in the finance industry. Both models rely on the affine structure of the price dynamics (c.f. \citealp{A-JLP:19}, for discussion in the rough case) to obtain implementable pricing formulae via Fourier inversion.

\subsection{Rough Heston++}
The $\mathbb{Q}$-dynamics of the logarithmic price $X_t$ are given by
    \begin{align}
        dX_t &= \mu_t dt + \sqrt{\sigma_t^2} (\rho d W_t +\sqrt{1-\rho^2} d W_t^\perp), \label{eq:rh++1}  \\
        \sigma^2_t &=\xi_0(t) +\frac{\nu}{\Gamma(H+\frac{1}{2})}\int_0^t \frac{\sigma_s}{(t-s)^{1/2-H}}d W_s, \label{eq:rh++2} 
    \end{align}
    where $\mu_t=-\frac{1}{2}\sigma_t^2$, $\nu>0$, $\rho\in[-1,1]$ and $H\in(0,\frac{1}{2}]$. The parameter $H$ is known as the Hurst exponent and drives the regularity of the trajectories of the Riemann-Liouville fractional Brownian motion. If $H=\frac{1}{2}$ and $\xi_0(t)=\sigma^2_0$, we retrieve the classical \cite{H:93} model with no volatility drift. When $H<\frac{1}{2}$, we obtain the genuinely \enquote{rough} model in, e.g., \citet{ER:19} and \citet{EGR:19}. 
    
    The term $\xi_0(t):=\E^{\mathbb{Q}}[\sigma^2_t|\mathcal{F}_0]=\E^{\mathbb{Q}}[\sigma^2_t]$ is known as the initial forward variance and can be viewed/modeled as a displacement. Consistent with \textit{Edgeworth++}, we specify the displacement as
   \begin{equation}
  \xi_0(t)=\sum_{k=0}^{n-1}\tilde{a}_{k}1_{[\tau_k,\tau_{k+1})}(t),\qquad \forall t\ge 0,
\end{equation}
where $\tau_0 = 0$ and $\tilde{a}_k \in \R$ for all $k=1,\dots, n-1$. Piecewise initial forward variance in rough models was also proposed, e.g., in 
\citet{HMT:21} and \citet{R:22}. Thus, \textit{Rough Heston++} has the same deterministic displacement as \textit{Edgeworth++}. The difference between the two models stems from the different characterization of the dynamics of the return and volatility processes. The absence of discontinuities in returns and the affinity of the variance process will be central to the performance of \textit{Rough Heston++}.

The parameter vector is $\Theta_t = (\sigma_{t},\rho_t,\nu_t,H_t,\tilde{a}_{t,1},\ldots,\tilde{a}_{t,n-1})$, where - as earlier - $n$ is the number of tenors on the surface. Thus, the \textit{Rough Heston++} model is very parsimonious with only $3+n$ parameters. 

The characteristic function of \textit{Rough Heston++} can be expressed in term of fractional Riccati ordinary differential equations (ODEs) for which there are no analytical solutions (\citealp{ER:19}). We do not impose Feller-type conditions (discussed, e.g., in \citealp{BP:24}). Details on implementation are discussed in Subsection \ref{sec:speed}.  

\subsection{2F Heston Merton}\label{HM}

The $\mathbb{Q}$-dynamics of the logarithmic price $X_t$ are given by the c\`adl\`ag process
    \begin{align*}
        d X_t &= \mu_t dt +\sqrt{\sigma^2_{1,t}}d W_{1,t}+\sqrt{\sigma^2_{2,t}}dW_{2,t}+z_{x} d N_t,\\
        d \sigma^2_{1,t}&=\kappa_{1}(\vartheta_1-\sigma^2_{1,t})dt+\zeta_1\sqrt{\sigma^2_{1,t}}dB_{1,t}+z_{v} d N_t,\\
         d \sigma^2_{2,t}&=\kappa_{2}(\vartheta_2-\sigma^2_{2,t})dt+\zeta_2\sqrt{\sigma^2_{2,t}}dB_{2,t},
    \end{align*}
    where $(W_{1},W_{2},B_{1},B_{2})$ is a four-dimensional Brownian motion with $W_{1}\perp W_2$, $W_{1}\perp B_2$, $[W_{1,t},B_{1,t}]_t=\rho_1 dt$ and $[W_{2,t},B_{2,t}]_t=\rho_2 dt$. The intensity of the Poisson process $dN_t$ is self-exciting and equal to $c_t=c_0+c_1 \sigma^2_{1,t}+c_{2}\sigma^2_{2,t}$. The distribution of the jump sizes $(z_{x},z_{v})$ is given by $z_{v}\sim \text{Exp}(m_v)$ with  $z_{x}|z_v \sim \mathcal{N}(\mu_x+\rho_0 z_v,\sigma_x)$ (as in \citealp{DPS:00}, \textit{inter alia}). For $i=0,1,2$ we impose $\rho_i\in[-1,1]$. Finally the $\mathbb{Q}$-drift $\mu_t$ is given, because of the jump compensation and the convexity adjustment, by
    \begin{equation*}
        \mu_t=-\frac{1}{2}\sigma^2_{1,t}-\frac{1}{2}\sigma^2_{2,t}-\frac{e^{\mu_x+\frac 12\sigma^2_x}}{1-m_v\rho_0} c_t.\footnote{As with previous model specifications, interest rates - and dividends - are set to zero.}
    \end{equation*} 

This model specification is a self-exciting version of widely-adopted 2-factor specifications in the affine family (\citealp{DPS:00}).\footnote{Jump self-excitation makes the model very close to the one adopted by \citet{bates2019crashes}. The main difference between the model in \citet{bates2019crashes} and the one presented here has to do with the pure jump nature of the volatility factors beyond the first factor.} As is typical in this family of models, we do not assume presence of a displacement (at least initially).

The parameter vector is $$\Theta_t = (\sigma^2_{1,t},\sigma^2_{2,t},\kappa_{t,1},\kappa_{t,2},\vartheta_{t,1},\vartheta_{t,2},\zeta_{t,1},\zeta_{t,2},\rho_{t,0},\rho_{t,1},\rho_{t,2},\mu_{t,x},\sigma_{t,x},m_{t,v},c_{t,0},c_{t,1},c_{t,2}),$$
and contains $17$ parameters. 

As is well-known, the characteristic function of models in the affine family, including the adopted \textit{2F Heston Merton} model, can be expressed in term of Riccati ODEs for which semi-analytical solutions generally exists whose tractability is, however, limited to simpler models. Below, we report results with and without a Feller condition. Details on implementation are, again, discussed in Subsection \ref{sec:speed}. 

\subsection{Estimation}

The parameters of each model are obtained by minimizing an objective function measuring the squared distance between the market implied volatilities and their model-implied counterparts:
	\begin{equation}\label{def:rmse}
		\text{RMSE}(\Theta_t)=100\times\sqrt{\frac{1}{|\mathcal{T}|}\sum_{\tau\in\mathcal{T}}\frac{1}{|\mathcal{K}_\tau|}\sum_{K\in \mathcal{K}_{\tau}} \big(\sigma_{BS}^{\text{mdl}}(\tau,K,\Theta_t)-\sigma_{BS}^{\text{mkt}}(\tau,K)\big)^2},
	\end{equation}
	where $\Theta_t\in\R^d$ is the vector of model parameters, $\mathcal{T}$ is the set of maturities (tenors), $\mathcal{K}_{\tau}$ is the set of strikes for each tenor $\tau\in\mathcal{T}$, $\sigma_{BS}^{\text{mdl}}(\tau, K,\Theta_t)$ is the model implied  volatility for a call/put option with maturity $\tau>0$, strike $K$, and parameter vector $\Theta_t$, and $\sigma_{BS}^{\text{mkt}}(\tau,K)$ is the corresponding market implied volatility. The objective function is also used as a metric to evaluate model performance.
	
We emphasize that the vector $\Theta_t$ has a subscript $t$ signifying that estimation of the model parameters/state variables is conducted for each day in the sample (at 10:30, as discussed in Section \ref{sec:data}). Importantly, because \textit{Edgeworth++} is the only model defined in terms of \textit{processes} (i.e., without fixed parameters, even in the jump component), the procedure is internally consistent in its case. We adopt it with competing model specifications - even for parameters that should in theory be fixed and are, instead, re-estimated during each day in the sample - to give the competing models the highest chance to succeed.  

\section{Performance}\label{sec:numerical_results}

We evaluate model performance along three dimensions: computational speed, the pricing of 0DTE options and the pricing of \textit{ultra-short-term} implied-volatility surfaces. We begin by comparing the three models based on pricing speed.

\subsection{Pricing speed}\label{sec:speed}

To assess computational efficiency, we conduct a case study designed to mirror our empirical pricing exercise. Specifically, we price options with six maturities. This setup corresponds to pricing, e.g., on a typical Monday at 10:30 with expirations every day at 16:00 until the next Monday. For each tenor, we select three contracts: given the implied forward, we consider an ATM put option, the closest OTM put at log-moneyness $m=-0.15$ and the closest OTM call at $m=0.15$, where $m$ is defined in Eq. \eqref{eq:m}. We use the parameter vector $\Theta_t$ estimated on May 6, 2022.

We consider two pricing exercises. First, we price a single cross section (i.e., three option contracts) corresponding to the shortest maturity (0DTE options), which replicates the empirical exercise carried out in Section~\ref{sec:0dte_numerics}. Second, we price the entire surface (comprising six tenors, for a total of 18 option contracts), which mirrors the empirical exercise in Section~\ref{sec:6shortest_numerics}. To mitigate the impact of random variation, we repeat the pricing exercise $10^3$ times. 

As in our estimation, option prices are computed using a Fast Fourier Transform method based on the characteristic function of each model. Specifically, we employ the SINC method developed by \cite{BBRR:22}.\footnote{As with any Fourier-based pricing method, the SINC approach requires the choice of hyperparameters. In our implementation, we set the number of nodes of the Fourier approximation equal to $N_f=10^4$.} 
Because the same pricing algorithm and parameters are used across all specifications, differences in computational time reflect differences in the complexity of the underlying characteristic function, rather than variation in numerical implementation.

\begin{table}[H]
	\centering
	\begin{tabular}{|c|c|c|c|}
		\hline
		Tenors & Edgeworth++ & Rough Heston++ & 2F Heston Merton \\
		\hline
		0 days (0DTE) 
		& \textbf{0.0046} $\pm$ 4.4 $\times 10^{-5}$
		& 0.1281 $\pm$ 0.0004 
		& 0.3080 $\pm$ 0.0075  \\
		up to 1 week 
		& \textbf{0.1452} $\pm$ 0.0023 
		& 1.8149 $\pm$ 0.0364 
		& 13.0236 $\pm$ 0.2348  \\
		\hline
	\end{tabular}
	\caption{Computational time (in seconds) on a local machine to obtain option prices given model parameters. For each maturity, we consider three contracts: one ATM put, one OTM call, and one OTM put. The first row reports computational times for the three options at the shortest maturity only. The second row reports computational times for the entire surface of 18 options. The pricing exercise is repeated $10^3$ times, and average times are reported. The value after the $\pm$ symbol corresponds to $1.96$ times the standard deviation of the average time. The fastest model is shown in bold.}
	\label{tab:panel}
\end{table}

Table~\ref{tab:panel} highlights substantial differences in performance across models. For the shortest maturity (0DTE options), \textit{Edgeworth++} is by far the fastest specification, requiring only 0.0046 seconds on average, compared to 0.1281 seconds for \textit{Rough Heston++} and 0.3080 seconds for \textit{2F Heston Merton}. These figures correspond to 28-fold and 67-fold speed improvements relative to \textit{Rough Heston++} and \textit{2F Heston Merton}, respectively.

The differences in performance are readily traced to the structure of the corresponding characteristic function. In \textit{Edgeworth++}, the characteristic function (expansion) is in closed form (c.f. Corollary~\ref{cor:expansion_piecewise}). In contrast, \textit{2F Heston Merton} requires numerical integration of a system of ODEs due to the joint presence of jumps and multiple volatility factors. Similarly, the affine fractional nature of \textit{Rough Heston++} calls for numerical integration. In our implementation, we rely on the rational approximation proposed by \cite{GR:19}, which has been shown to be both accurate and computationally efficient.

When pricing the entire surface, computational costs increase for all models, but at markedly different rates. The \textit{2F Heston Merton} model becomes particularly time consuming, with average computational time rising from 0.3080 seconds for a single maturity to over 13 seconds for the full surface. This outcome reflects the need to solve a system of ODEs over each maturity, so that longer horizons require more computational steps when keeping the system discretization constant. Comparing \textit{Rough Heston++} to \textit{Edgeworth++}, we find that computational time increases by a factor of about 14 for the former (from 0.1281 to 1.8149 seconds) and by a factor of about 32 for the latter (from 0.0046 to 0.1452 seconds). This difference reflects how the shift extension enters the characteristic function. In \textit{Rough Heston++}, the shift is defined on the variance process and factorizes conveniently (see Eq.~(\ref{eq:rh++2})), whereas in \textit{Edgeworth++} it is defined on volatility and enters the characteristic function in a nontrivial manner, thereby increasing computational complexity. Having made this point, \textit{Edgeworth++} remains roughly 12 times faster than \textit{Rough Heston++} when pricing the full surface.

Overall, the computational speed of the \textit{Edgeworth++} model enables its deployment not only in large-scale estimation exercises, as in Sections~\ref{sec:0dte_numerics} and~\ref{sec:6shortest_numerics} below, but also in real-time option pricing applications. This feature is important for market makers and other market participants who require fast and reliable pricing in high-frequency trading environments. As mentioned in the Introduction, because of computational considerations, in this article we do not compare \textit{Edgeworth++} to interesting model specifications (e.g., the quadratic rough Heston model of \citealp{bourgey2026quadratic}) without a known characteristic function permitting Fourier pricing.

\subsection{0DTEs}\label{sec:0dte_numerics}
Next, we assess pricing performance on the shortest maturity of $5.5$ hours (corresponding to 0DTE options). When pricing a single tenor, the displacement is redundant. Thus, we set $\xi_0(t)=\sigma_0^2\in\mathbb{R}_{+}$ in the \textit{Rough Heston++} model and $\phi(t)=0$ in the \textit{Edgeworth++} model. In other words, in the latter case, we use the baseline specification (i.e., \textit{Edgeworth}, without \textit{++}) in \cite{BFR:24} (c.f. Corollary \ref{cor:expansion}). 

Evaluating a single tenor is a helpful starting point to highlight different features of the models. To this end, Table \ref{tab:one_maturity_comparison} reports summary statistics for the daily RMSEs (c.f. Eq. (\ref{def:rmse})). Fig. \ref{fig:0dte_rH_combined} displays the distribution of the daily RMSEs.

\textit{Edgeworth} performs the best in terms of RMSE distribution. We also find that \textit{Edgeworth}'s prices are within the quoted bid-ask spread $82\%$ of the time. While the \textit{2F Heston Merton} model performs similarly, it does not outperform \textit{Edgeworth} along either metric (RMSE distribution or percentage of prices within the bid-ask spread) in spite of its 17 parameters versus \textit{Edgeworth}'s 8. We also note that the satisfactory performance of \textit{2F Heston Merton} is dependent on the removal of the Feller condition. Should the Feller condition not be removed, model performance would deteriorate drastically. 

These results are suggestive of three observations. First, affine structures are generally not flexible enough to capture the implied-volatility smiles of 0DTEs options. \citet{BFR:24} reached the same conclusion using a benchmark \textit{1F Heston Merton} model.\footnote{\citet{BFR:24} refer to it as the Bates model in that \citet{bates1996jumps} popularized the addition of jumps (in Merton's tradition) to an affine variance specification (in Heston's tradition).} Adding a second volatility factor in a \textit{2F Heston Merton} model does not modify their conclusion (we will return to this issue). Second, the introduction of jumps in volatility (in \textit{2F Heston Merton}) has a muted impact on 0DTE pricing. Again, \citet{BFR:24} reached this same conclusion after contrasting an \textit{Edgeworth} model with jumps in prices to an \textit{Edgeworth} model with jumps in both prices and volatility. Third, affine models have satisfactory pricing performance only when the Feller condition is relaxed.

\begin{table}[t!]
	\centering
	\begin{tabular}{|l|c|c|c|c|c|}
		\hline
		& \multicolumn{3}{c|}{RMSE} &  Number of & Fraction \\
		Model & Mean & 10th perc. & 90th perc. &  Parameters & in Bid/Ask\\
		\hline
		Edgeworth & \textbf{0.4255} & 0.2829 & \textbf{0.6026} & 8 & \textbf{0.820} \\
		Rough Heston & 1.6813 & 0.9099 & 2.5123 & \textbf{4} & 0.315 \\
		2F Heston Merton & 0.4700 & \textbf{0.2789} & 0.7387 & 17 & 0.783 \\
		2F Heston Merton (Feller) & 0.6277 & 0.3038 & 1.2289 & 17 & 0.676 \\
				\hline
	\end{tabular}
	\caption{Model comparisons when pricing 0DTEs from May 6, 2022 to May 11, 2023. The RMSE is computed as in Eq. \eqref{def:rmse} with $|\mathcal{T}|=1$. Data source: CBOE.}
	\label{tab:one_maturity_comparison}
\end{table}

\begin{figure}[t!]
   \centering
   \includegraphics[width=0.60\linewidth]{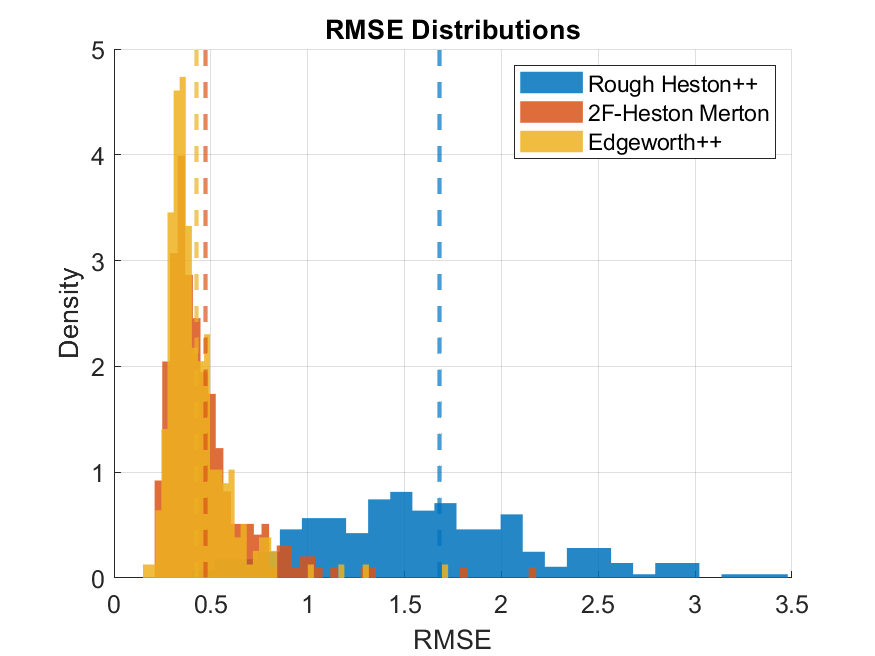}
   \caption{Histograms of the daily RMSEs of estimated models in our samples on 0DTEs, from May 6, 2022 to May 11, 2023. Data source: CBOE.}
    \label{fig:0dte_rH_combined}
\end{figure}

The \textit{Rough Heston} model finds it difficult to fit the short end of the implied-volatility smile. The jump component in prices appears to be particularly important to capture the OTM profile and cannot be substituted by volatility roughness. Generally speaking, the rough model is very parsimonious (only 4 parameters) but it is also somewhat constrained, as a result. 

In order to formalize matters, in the spirit of \citet{forde2009small}, we provide a short-tenor/large deviation expansion of \textit{Edgeworth}'s implied-volatility surface. 

\begin{thm}\label{secondo}
Assume the model in Corollary \ref{cor:expansion}. The first four cumulants of the continuous price process $X^c_{\tau} - X^c_0 = \mu_0 \tau + \sigma_0 \sqrt{\tau}Z^c_{\tau}$ are:
\begin{align*}
\kappa^{\star}_1(\tau) &= \mu_0\tau, \\
\kappa^{\star}_2(\tau) &= \sigma^2_0\tau +O(\tau^2), \\
\kappa^{\star}_3(\tau) &= \sigma^3_0\theta_3\tau^2 +o(\tau^2), \\
\kappa^{\star}_4(\tau) &= \sigma^4_0\theta_4\tau^3 +o(\tau^3),
\end{align*}
with $\theta_3 = 3\frac{\tilde\beta_0\rho_0}{\sigma_0}$ and $\theta_4 = 4\left(\frac{\eta_0}{\sigma_0}+\frac{\tilde\beta_0^2}{\sigma_0^2}(1+2\rho_0^2)\right)$. Also, the short-tenor implied-volatility surface evaluated at log-moneyness $x$ (around $x=0$) is
\begin{eqnarray}\label{IV}
I(x) = \sigma_0 \left[1+\frac{\theta_3}{6\sigma_0}x+\left(\frac{\theta_4}{24\sigma^2_0}-\frac{\theta_3^2}{12\sigma^2_0}\right)x^2+O(x^3)
\right],
\end{eqnarray}
which implies
\begin{align*}
	I(0)=\sigma_0,\qquad I'(0)=\frac{\theta_3}{6},\qquad I''(0)= \frac{\theta_4}{12\sigma_0}-\frac{\theta_3^2}{6\sigma_0}.
\end{align*}
\end{thm}

\begin{proof}
    See Appendix \ref{proof_secondo}.
\end{proof}

As expected, negative skew is driven by leverage ($\rho$). Convexity is, instead, largely driven by the volatility of volatility ($\widetilde{\beta}$) and the volatility of the volatility of volatility associated with the price Brownian motion $W$ ($\eta$). The latter is 0 in Heston-type models, a fact that makes \textit{Edgeworth} more flexible that typical affine specifications \textit{even} when these specifications with fixed parameters are re-estimated (as we do, in order to enhance their performance) for every day in the sample. 

Consider the affine model in Subsection \ref{HM} with only one variance factor ($\sigma_{1,t}^2$). By a simple application of It\^o's lemma, this is a model with $\widetilde{\beta}_t = \frac{1}{2}\sigma,$ $\rho_t = \rho$ and $\eta_t = 0$ $\forall t.$ Plugging these values into Eq. (\ref{IV}) yields the result in Theorem 3.2 of \citet{forde2009small}.

\subsection{\textit{Ultra-short-term} surfaces}\label{sec:6shortest_numerics}
We now turn to the \textit{joint} pricing of \textit{ultra-short-term} implied-volatility smiles. We recall that both the \textit{Rough Heston++} model and the \textit{Edgeworth++} model have a piecewise constant displacement aimed at capturing the level of the ATM implied-volatility term structure (essentially, the forward variance). The \textit{2F Heston Merton} model has no deterministic displacement, but it features two independent stochastic volatility factors, one with exponentially-distributed jumps. This latter model has the largest number of parameters across competitors.

\begin{table}[t!]
	\centering
	\begin{tabular}{|l|c|c|c|c|c|}
		\hline
		& \multicolumn{3}{c|}{RMSE} &  Number of & Fraction \\
		Model & Mean & 10th perc. & 90th perc. &  Parameters & in Bid/Ask\\
		\hline
		Edgeworth++ & \textbf{1.0195} & \textbf{0.7071} & \textbf{1.4234} & 13 & \textbf{0.383} \\
		Rough Heston++ & 2.7668 & 2.1327 & 3.4438 & \textbf{9} & 0.097 \\
		2F Heston Merton & 1.8826 & 1.1030 & 2.8391 & 17 & 0.301 \\
		2F Heston Merton (Feller) & 2.5959 & 1.4897 & 3.7923 & 17 & 0.269 \\
		\hline
	\end{tabular}
	\caption{Model comparisons when pricing ultra-short tenors between May 2022 and May 2023. The RMSE is computed as in Eq. \eqref{def:rmse} with $|\mathcal{T}|=6$. We employed the SINC method with $10^4$ nodes for the fast Fourier approximation. Data source: CBOE.}
	\label{tab:multi_maturity_comparison}
\end{table}

Table \ref{tab:multi_maturity_comparison} reports summary statistics for the daily RMSEs. Fig. \ref{fig:rmse_distributions} visualizes the daily RMSE distributions. In terms of average RMSE, \textit{Edgeworth++} achieves the best performance with an average error of $1$ volatility points. \textit{Edgeworth++}'s prices fall within the bid-ask spread for 36.7\% of the options. The second-best model is \textit{2F Heston Merton}, whose average RMSE is $1.88$, provided the Feller condition is relaxed. With the Feller condition imposed, the average RMSE grows to $2.60$, which is almost as high as the \textit{Rough Heston++}'s value of $2.77$. Once more, \textit{Rough Heston++} struggles to capture the full implied-volatility smiles across levels of moneyness, something which reflects a structural limitation of this (otherwise very parsimonious) model specification when price discontinuities are absent. These findings complement the recent results of \citet{GE:22} and \citet{A-JLi:25} who document underperformance of rough models on implied-volatility and skew surfaces with various medium-to-long maturities.

\begin{figure}[h!]
	\centering
	\includegraphics[width=0.60\linewidth]{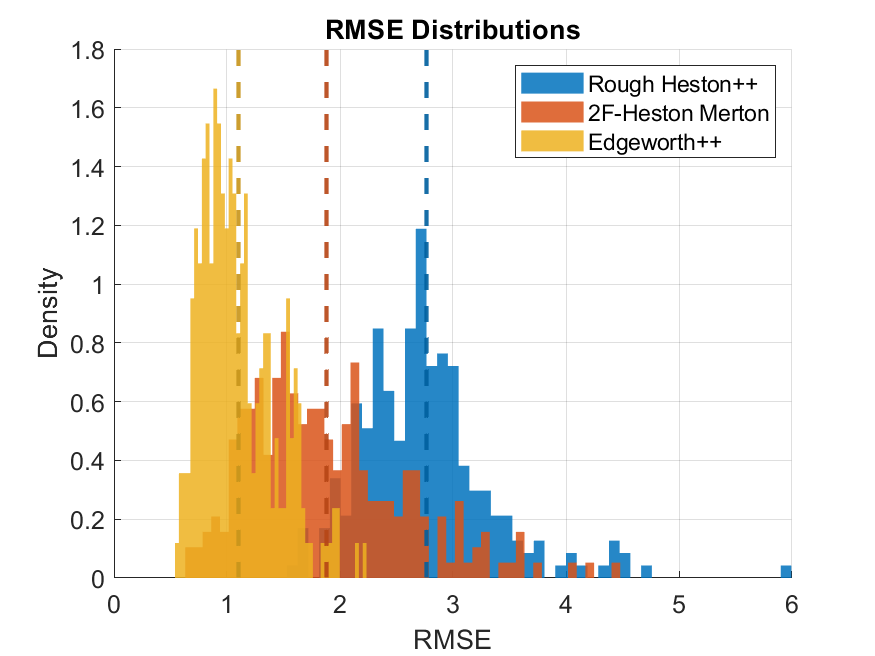}
	\caption{Histograms of the daily RMSEs of estimated models on \textit{ultra-short-term} implied-volatility surfaces, from May 6, 2022 to May 11, 2023. Data source: CBOE.}
	\label{fig:rmse_distributions}
\end{figure}

\begin{table}[h!]
	\centering
	\begin{tabular}{lcccccc}
		\toprule
		Model & Tenor 1 & Tenor 2 & Tenor 3 & Tenor 4 & Tenor 5 & Tenor 6 \\
		\midrule
		\multicolumn{7}{c}{\textbf{DOTMP}} \\
		\midrule
		Edgeworth++ & 2.9089 & \textbf{1.2281} & \textbf{0.8191} & \textbf{0.6284} & \textbf{0.6533} & \textbf{0.8498} \\
		Rough Heston++ & \textbf{2.8690} & 3.4466 & 3.5848 & 3.4324 & 3.2582 & 3.2351 \\
		2F Heston Merton & 4.6261 & 2.1708 & 1.4591 & 1.0533 & 0.8173 & 0.8809 \\
		\midrule
		\multicolumn{7}{c}{\textbf{OTMP}} \\
		\midrule
		Edgeworth++ & \textbf{1.0421} & \textbf{0.5297} & \textbf{0.4009} & \textbf{0.3138} & \textbf{0.2970} & \textbf{0.3186} \\
		Rough Heston++ & 1.0530 & 1.8936 & 1.6878 & 1.5603 & 1.3046 & 0.9132 \\
		2F Heston Merton & 5.2459 & 2.6680 & 1.6080 & 1.1169 & 0.7978 & 0.9083 \\
		\midrule
		\multicolumn{7}{c}{\textbf{ATM}} \\
		\midrule
		Edgeworth++ & 1.5011 & \textbf{0.5197} & \textbf{0.4266} & \textbf{0.3461} & \textbf{0.3239} & \textbf{0.3749} \\
		Rough Heston++ & \textbf{1.4944} & 3.7008 & 3.2228 & 3.0336 & 2.5590 & 1.9050 \\
		2F Heston Merton & 4.6233 & 2.7865 & 1.6019 & 1.0806 & 0.7799 & 0.9316 \\
		\midrule
		\multicolumn{7}{c}{\textbf{OTMC}} \\
		\midrule
		Edgeworth++ & 2.0809 & \textbf{0.5418} & \textbf{0.3975} & \textbf{0.3144} & \textbf{0.3207} & \textbf{0.4217} \\
		Rough Heston++ & \textbf{1.7350} & 3.0281 & 2.5156 & 2.2320 & 1.7920 & 1.2677 \\
		2F Heston Merton & 4.3402 & 2.7832 & 1.5444 & 0.9790 & 0.6907 & 0.7848 \\
		\midrule
		\multicolumn{7}{c}{\textbf{DOTMC}} \\
		\midrule
		Edgeworth++ & 2.3304 & \textbf{0.8467} & \textbf{0.6339} & \textbf{0.5512} & \textbf{0.5868} & 0.7567 \\
		Rough Heston++ & \textbf{1.8627} & 1.8161 & 1.8532 & 1.9817 & 2.1750 & 2.4586 \\
		2F Heston Merton & 3.9620 & 2.4303 & 1.3487 & 0.8989 & 0.6721 & \textbf{0.7141} \\
		\midrule
		\bottomrule
	\end{tabular}
	\caption{The table reports the mean RMSEs for options with different log-moneyness levels. Log-moneyness $m$ is defined as $m = \frac{\log(K)-\log(F)}{\sigma_{BS}\sqrt{\tau}}$, c.f.  Eq. \eqref{eq:m}. OTMP and DOTMP represent out-of-the-money ($-1 < m < -0.35$) and deep out-of-the-money ($m < -1$) put options, respectively. OTMC and DOTMC refer to out-of-the-money ($0.35 < m < 1$) and deep out-of-the-money ($m > 1$) call options, respectively. The ATM panel reports averages on at-the-money options, defined by $m\in (-0.35,0.35)$. Data source: CBOE.}
	\label{table:dissection}
\end{table}

Table \ref{table:dissection} partitions model's performance across 6 maturities and across strikes. We divide the strikes into 5 subsets, defined in the table caption. Based on average RMSEs, \textit{Edgeworth++} performs best (and generally significantly so) across both maturities and strikes. The exception is the shortest tenor, for which \textit{Edgeworth++} is almost equivalent to \textit{Rough Heston++} ATM and for OTM puts, and is slightly inferior to \textit{Rough Heston++} for OTM calls and deep OTM calls. The satisfactory short-term performance of \textit{Rough Heston++}, when the entire \textit{ultra-short-term} implied-volatility term structure is priced, should not be viewed as being in contradiction with the less positive results reported in the case of 0DTEs (c.f. Subsection \ref{sec:0dte_numerics}). Its flip side is, in fact, subpar performance (in absolute terms and relative to alternative model specifications) as the tenor lengthens. In essence, \textit{joint} pricing through the displacement (but, again, without jumps) results in a competitive RMSE value over the shortest tenor and more sizable RMSE values over longer maturities.   

Fig. \ref{fig:bidask} illustrates the percentage of options priced within the bid-ask spread by each model across maturities and strikes. The shape of the figure reflects the fact that the spreads are lower ATM. Based on this metric, again, \textit{Edgeworth++} outperforms the benchmarks across all strikes and maturities.

\begin{figure}[h!]
	\centering\includegraphics[width=0.6\textwidth]{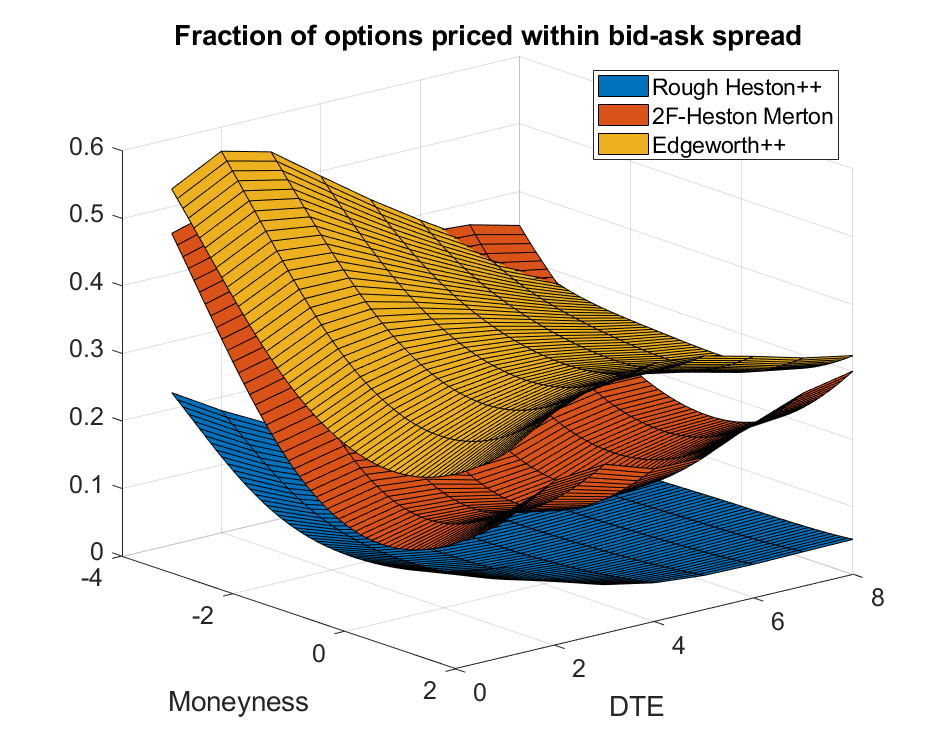}
	\caption{We report the percentage of prices within the bid-ask spreads for each model across strikes and tenors. Data source: CBOE.}
	\label{fig:bidask}
\end{figure}

\subsection{Discussion}\label{discussion}

\textit{Edgeworth++}'s key modeling features are a nonparametric \textit{stochastic} volatility component whose dynamics (through, e.g., the volatility of volatility and leverage) are central to ATM pricing and a \textit{deterministic} volatility component, i.e., the displacement. The first feature plays a role when fitting the ATM - and near ATM - implied-volatility smile for \textit{each} tenor. When combined with jumps (which \textit{Edgeworth++} includes), the end result is satisfactory fit over the entire implied-volatility smile (i.e., including OTM and deep OTM) for \textit{each} tenor. The second feature guarantees that the oscillations of the ATM implied-volatility term structure are captured accurately. Said differently, the first (along with jumps) captures the shape of the implied volatilities across broad log-moneyness levels for \textit{each} tenor. The second captures changes in the market expected forward variance \textit{across} tenors. 

Compared to \textit{Rough Heston++}, which also includes a displacement, \textit{Edgeworth++} benefits from the non-affine specification of its stochastic volatility component and from the inclusion of jumps in the price process. The volatility dynamics of \textit{Rough Heston++} lead to marked V-shapes around the ATM range which, in some cases, have difficulties adapting to the ATM implied volatilities. The absence of jumps in the model, in turn, may pull these V-shapes towards extreme OTM implied volatilities in an attempt to better fit the \textit{entire} implied-volatility smile, thereby affecting the level of the (model-implied) implied volatilities near ATM. In essence, while the variability of the ATM implied-volatility term structures is captured rather effectively by \textit{Rough Heston++}  (because of the displacement), the shape and the level of the implied volatilities over individual tenors may suffer from the model implicitly trading off ATM vs. OTM fit.

To illustrate, Fig. \ref{fig:calib_example2} reports the implied-volatility smiles for each of the 6 tenors on the same day used for Fig. \ref{fig:iv_term_structure}, namely August 18, 2022. Panel A of Fig. \ref{fig:iv_term_structure2} documents the fit of the ATM implied-volatility term structure associated with the three model specifications on that day (c.f. Fig. \ref{fig:iv_term_structure}). We note that the V-shapes delivered by \textit{Rough Heston++} do not adapt to the market implied volatilities (c.f. Fig. \ref{fig:calib_example2}). Even though the displacement helps \textit{Rough Heston++} capture the empirical variability of the term structure of ATM implied volatilities, the pull of the V-shapes towards OTM implied volatilities leads to downward biases ATM (c.f. Panel A of Fig. \ref{fig:iv_term_structure2}).  

\begin{figure}[t!]
	\centering
	\begin{subfigure}{0.30\linewidth} 
		\centering
		\includegraphics[width=\linewidth]{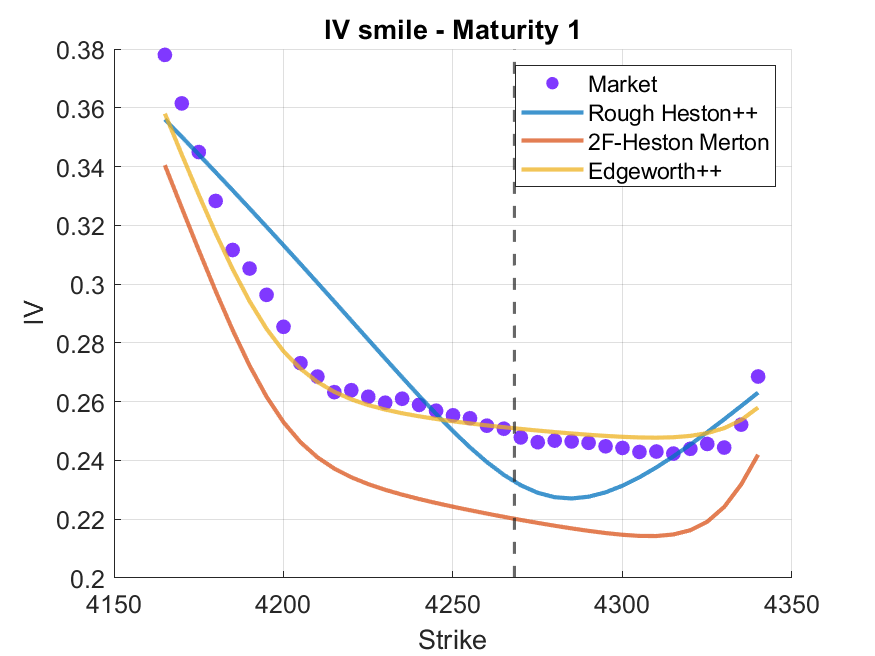}
	\end{subfigure}
	\begin{subfigure}{0.30\linewidth}
		\centering
		\includegraphics[width=\linewidth]{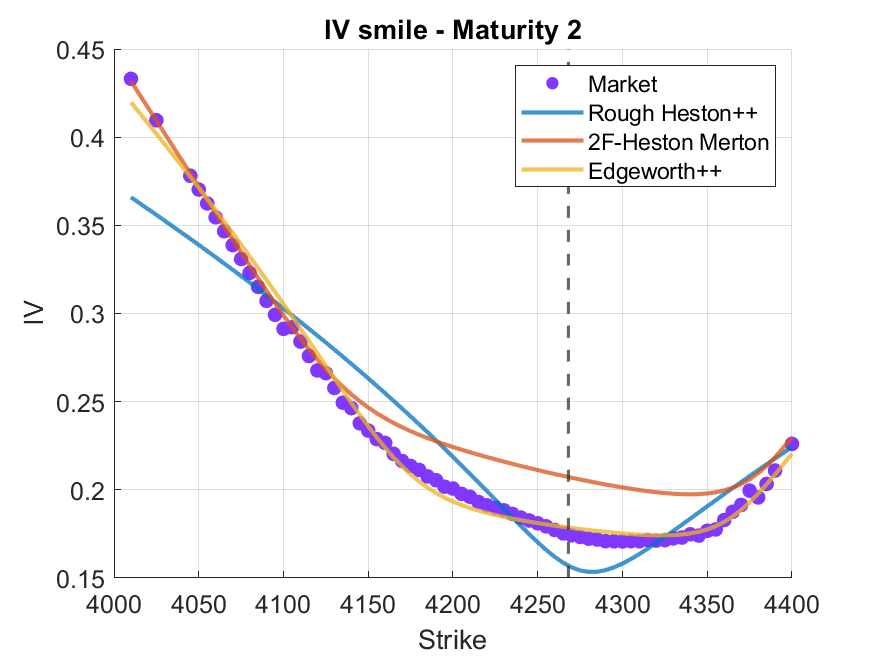}
	\end{subfigure}
	\begin{subfigure}{0.30\linewidth} 
		\centering
		\includegraphics[width=\linewidth]{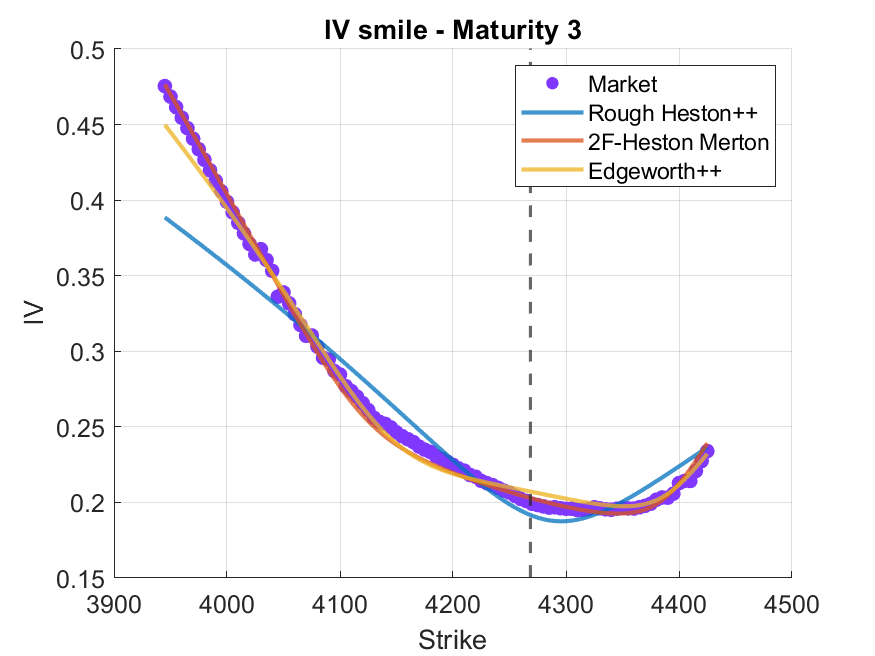}
	\end{subfigure}
	
	\vspace{0.3cm}
	
	\begin{subfigure}{0.30\linewidth}
		\centering
		\includegraphics[width=\linewidth]{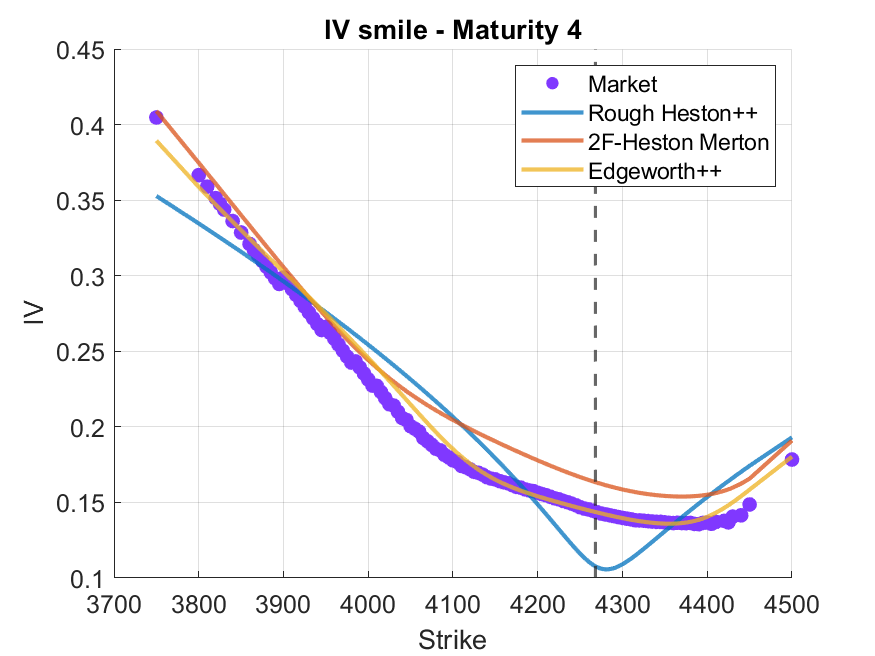}
	\end{subfigure}
	\begin{subfigure}{0.30\linewidth}
		\centering
		\includegraphics[width=\linewidth]{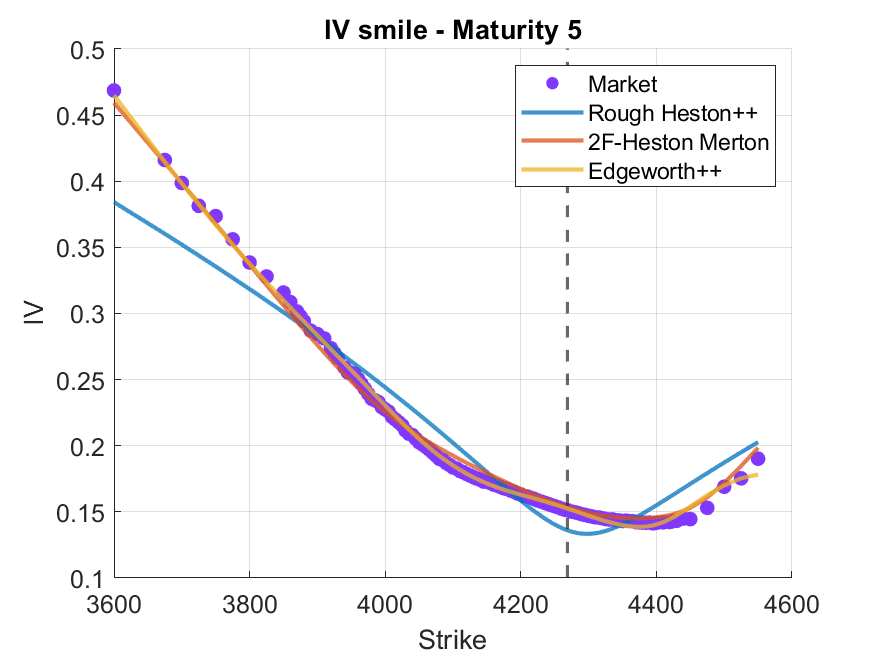}
	\end{subfigure}
	\begin{subfigure}{0.30\linewidth} 
		\centering
		\includegraphics[width=\linewidth]{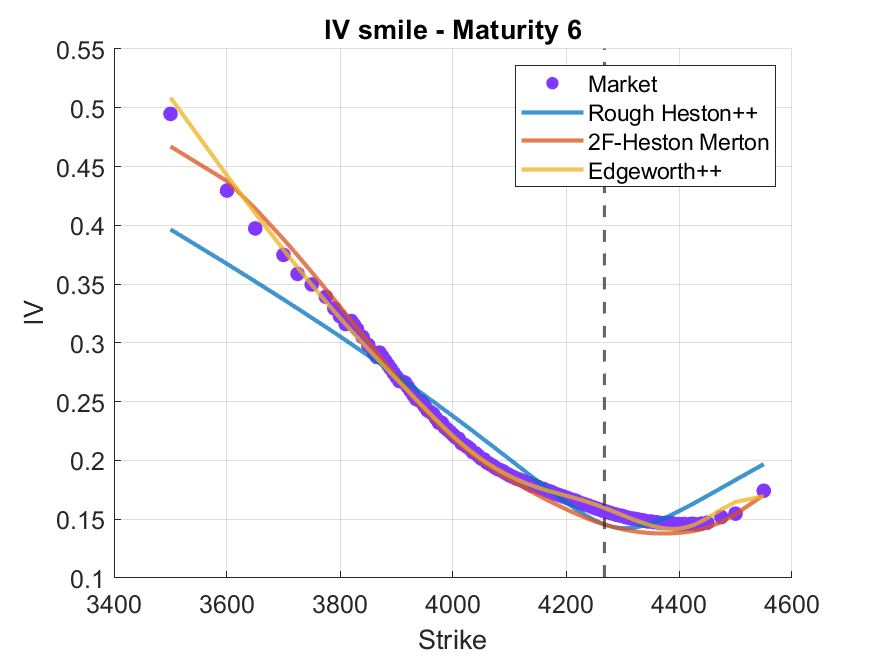}
	\end{subfigure}
	
	\caption{Estimated implied-volatility smiles on August 18, 2022. The grey dashed lines denote the ATM level for the corresponding maturity. Data source: CBOE.}
	\label{fig:calib_example2}
\end{figure}

Relative to the \textit{2F Heston Merton} model, the key to improved pricing accuracy on the full ATM implied-volatility term structure in \textit{Edgeworth++} is the displacement. Simply put, two stochastic volatility factors cannot yield the variability of the ATM implied-volatility term structure found in the data. As shown in Panel A of Fig. \ref{fig:iv_term_structure2}, the \textit{2F Heston Merton} implied term structure is excessively smooth and is, in essence, an averaged version of the term structure in the data. 

\begin{figure}[t!]
    \centering
    \begin{subfigure}{0.48\linewidth}
        \centering Panel A\\
        \includegraphics[width=\linewidth]{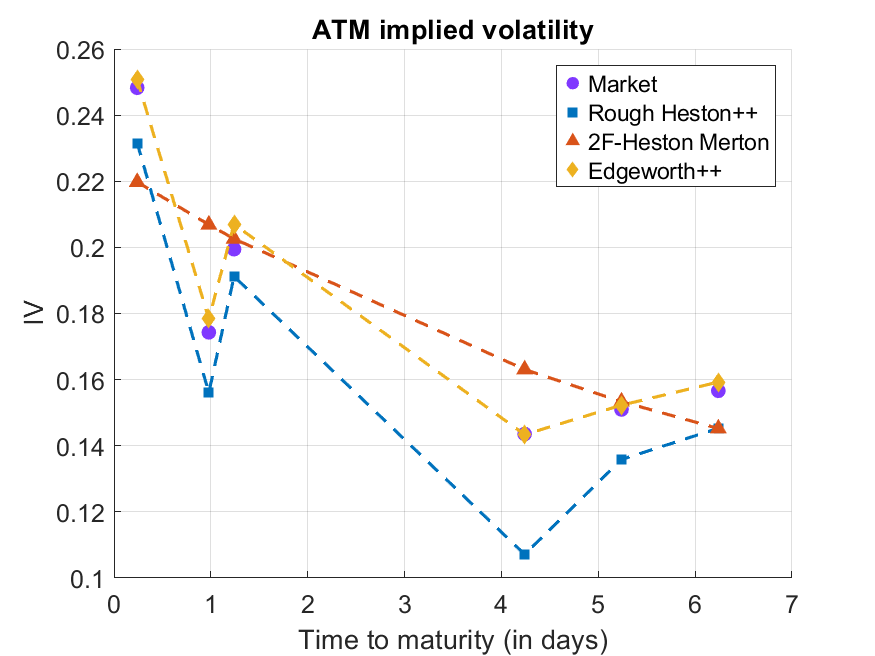}
        \caption{August 18, 2022}
        \label{fig:iv_term_structure_calib3}
    \end{subfigure}
    \hfill
    \begin{subfigure}{0.48\linewidth}
        \centering Panel B\\
        \includegraphics[width=\linewidth]{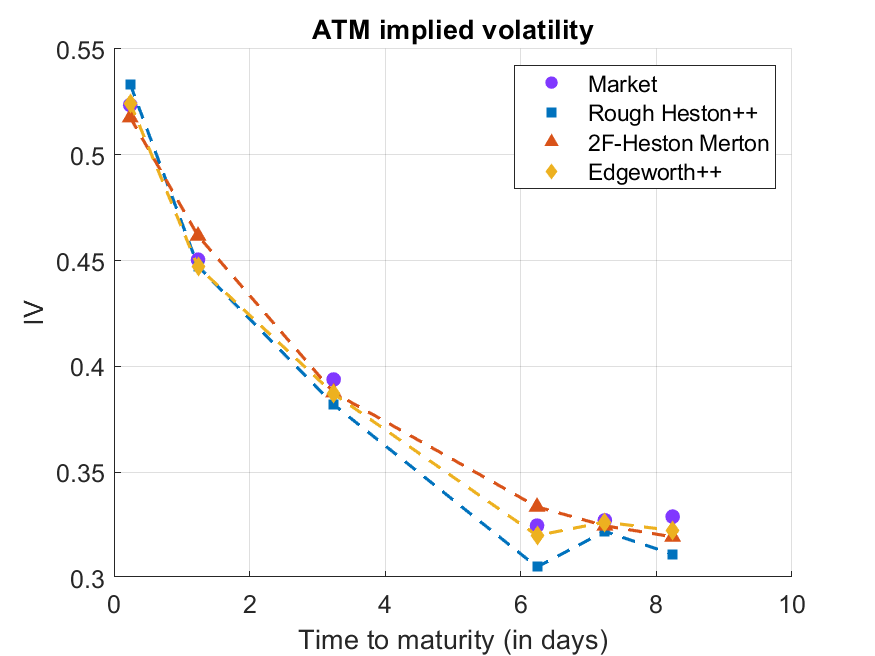}
       \caption{May 10, 2022}
        \label{fig:iv_term_structure_calib2}
    \end{subfigure}
    \caption{ATM implied-volatility term structure on 2 typical day, Thursday August 18, 2022 and Tuesday May 10, 2022. The colored dashed lines correspond to the fit of alternative models. Data source: CBOE.}
    \label{fig:iv_term_structure2}
\end{figure}

While August 18, 2022 is a typical day in our sample, less variable term structures may arise. In Panel B of Fig. \ref{fig:iv_term_structure2} we report the \textit{ultra-short-term} ATM implied-volatility term structure on May 10, 2022.  The purple dots (the market implied volatilities) are now more regularly distributed on an almost monotonically decreasing curve. All models perform satisfactorily, with \textit{Edgeworth++} fitting almost perfectly.

Fig. \ref{fig:calib_example3} provides evidence regarding the implied-volatility smiles for each tenor on May 10, 2022. As earlier on a different day, the figure is explicit about the structure imposed by the V-shapes featured by \textit{Rough Heston++}. It is also suggestive of the limited adaptability around ATM that is delivered by the \textit{2F Heston Merton}'s affine dynamics. 

Similarly to May 10, 2022, we are not excluding that the use of business-time sampling in the computation of the implied volatilities may lead to better-behaved term structures.  

\begin{figure}[t!]
	\centering
	\begin{subfigure}{0.30\linewidth} 
		\centering
		\includegraphics[width=\linewidth]{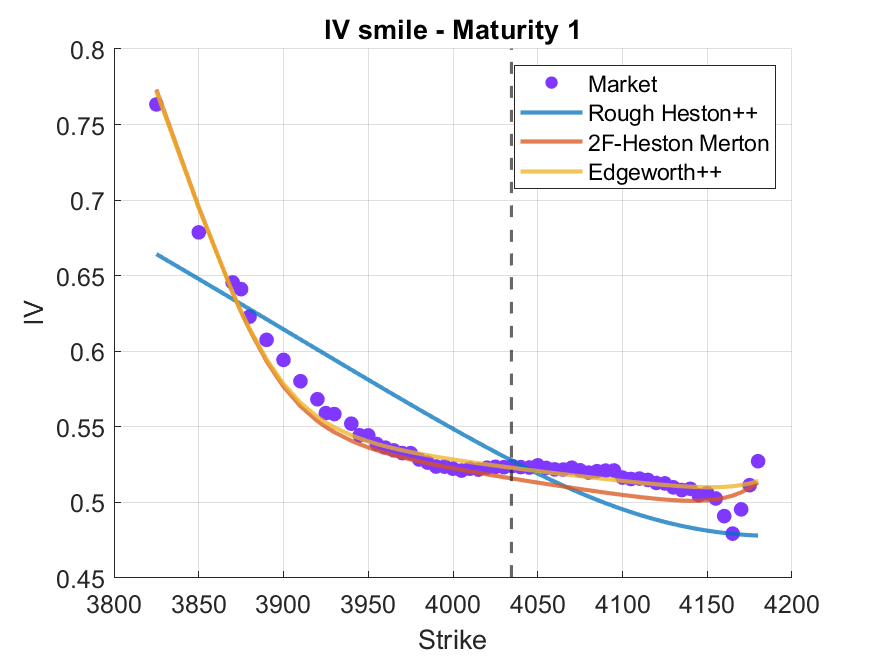}
	\end{subfigure}
	\begin{subfigure}{0.30\linewidth}
		\centering
		\includegraphics[width=\linewidth]{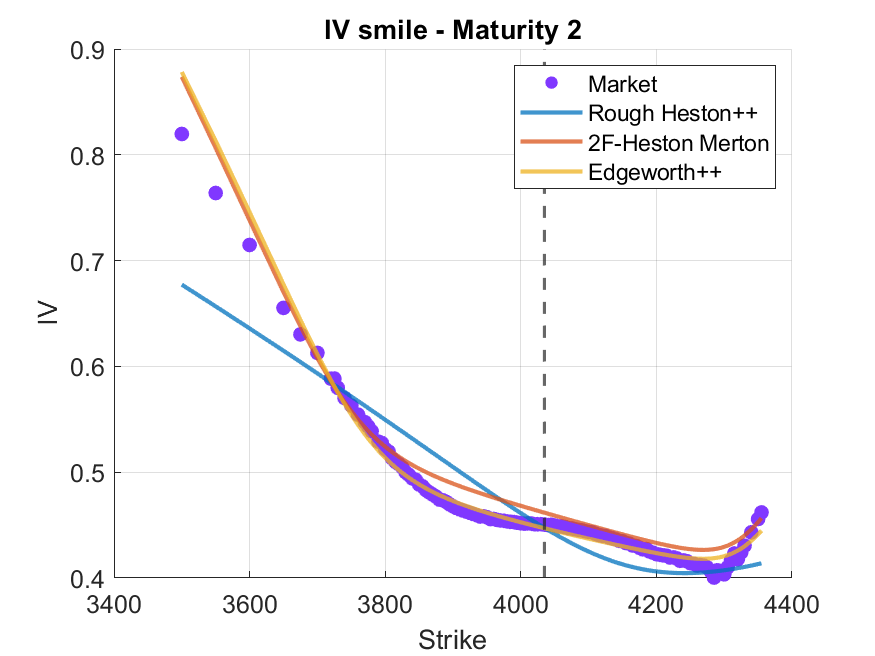}
	\end{subfigure}
	\begin{subfigure}{0.30\linewidth} 
		\centering
		\includegraphics[width=\linewidth]{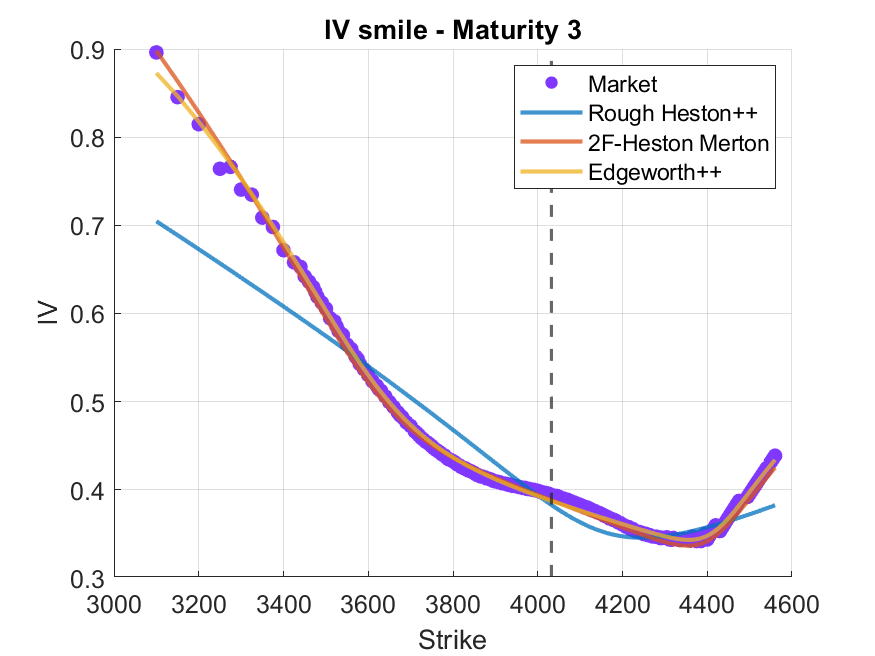}
	\end{subfigure}
	
	\vspace{0.3cm}
	
	\begin{subfigure}{0.30\linewidth}
		\centering
		\includegraphics[width=\linewidth]{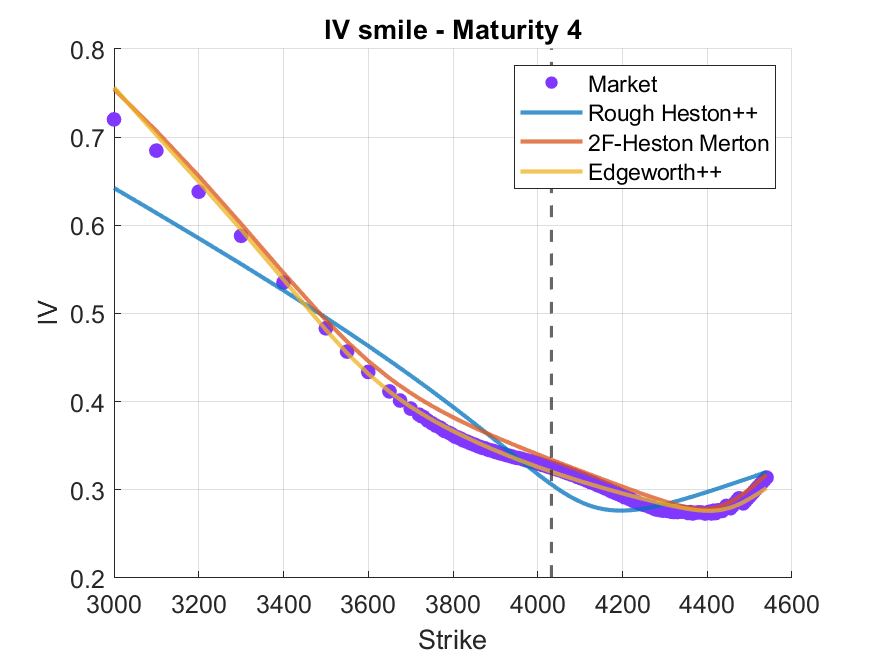}
	\end{subfigure}
	\begin{subfigure}{0.30\linewidth}
		\centering
		\includegraphics[width=\linewidth]{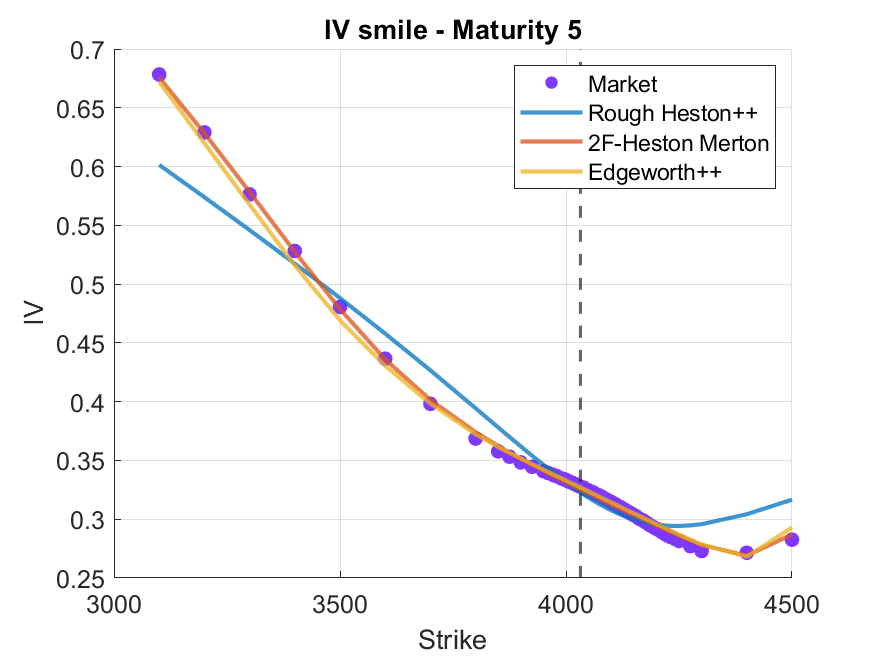}
	\end{subfigure}
	\begin{subfigure}{0.30\linewidth} 
		\centering
		\includegraphics[width=\linewidth]{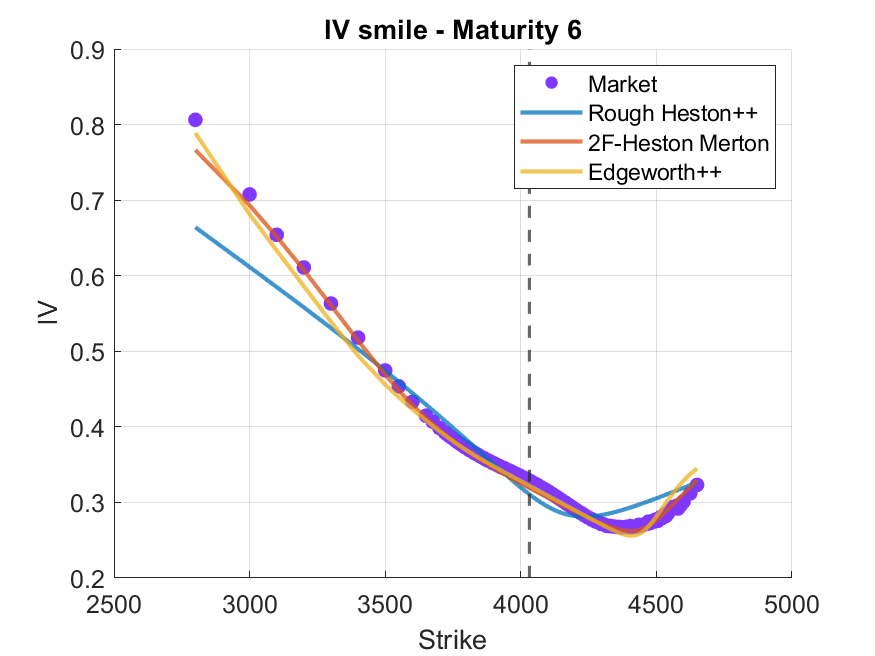}
	\end{subfigure}
	
	\caption{Estimated implied-volatility smiles on May 10, 2022. The grey dashed lines denote the ATM level for the corresponding maturity. Data source: CBOE.}
	\label{fig:calib_example3}
\end{figure}

\section{Improving the benchmarks}\label{impbench}

We have abided by traditional formulations of the benchmark models. Specifically, \textit{Rough Heston++} does not include price discontinuities (consistent with \citealp{BFG:16}, \citealp{Gatheral2018}, and much of the literature) and \textit{2F Heston Merton} does not include a displacement (consistent with \citealp{christoffersen2009shape}, and, again, much of the literature). 

Regarding the former, a central motivation for rough volatility lied in its ability to capture the presumed power law decay of the ATM short-term implied-volatility skew without price discontinuities. Limited interest has historically been placed in OTM dynamics. We have, however, suggested that the OTM range requires discontinuous changes. Absent discontinuous changes, attempting to fit OTM profiles may lead the parsimony of \textit{Rough Heston++} to be costly ATM. We have also suggested that, even when the exclusive focus is on the short-term ATM skew, affinity constrains \textit{Rough Heston++}  more than processes which suitably free up the volatility dynamics.  

As for \textit{2F Heston Merton}, there is an understanding that multiple volatility factors can lead to rich ATM implied-volatility term structures by capturing risk compensations over alternative time scales. We are not excluding that the addition of volatility factors (with the resulting proliferation of parameters) may lead to suitably oscillating \textit{ultra-short-term} implied-volatility term structures similar to those found in the data (and captured with a displacement). A classical 2-factor model has, however, been shown to be overly constrained. The affine features of the model have also led to insufficient ATM fit for individual tenors in spite of the very large number of parameters. 

Based on these considerations, in order to give the benchmarks their best chance to succeed, we enhance them with the dimensions that they are more clearly missing. We add (Gaussian) price discontinuities to \textit{Rough Heston++} (which leads to \textit{Rough Heston Merton++}) - as in the recent work of \citet{bondi2024rough} - and add a displacement to the \textit{2F Heston Merton} (which leads to \textit{2F Heston Merton++}) - as in the work of \citet{PRS:14,PPR:18}.\footnote{We could also add volatility factors. The number of parameters, however, would make estimation intractable.} The latter model is implemented without a Feller condition (again, to optimize performance). We also consider the \textit{1F Heston Merton} case, which will be shown to be inferior to its 2-factor counterpart, but whose parsimony makes it interesting.

Table \ref{tab:one_maturity_compa} is about the 0DTE case. Table \ref{tab:multi_maturity_compa} provides information on the full term structure. In order to help the reader, in both tables we report figures associated with the classical formulations of the same models and, therefore, repeat what was reported in Table \ref{tab:one_maturity_comparison} and Table \ref{tab:multi_maturity_comparison}.

As expected, adding price discontinuities to \textit{Rough Heston} helps price 0DTEs (Table \ref{tab:one_maturity_compa}). The previous mispricing of the \textit{Rough Heston} model is now drastically reduced. On the one hand, jumps help OTM. On the other hand, a parsimonious model (like \textit{Rough Heston}) is now asked to do less work OTM, thereby resulting in better fit ATM and near ATM. The mean RMSE is now 0.4813, a sizable improvement as compared to the no jump case (1.6813). Having made these points, \textit{Edgeworth} continues to perform better than \textit{Rough Heston Merton} due to the ability of its nonparametric volatility dynamics to replicate effectively the 0DTE implied-volatility skew and convexity (c.f. Theorem \ref{secondo}).  

\begin{table}[H]
	\centering
	\begin{tabular}{|l|c|c|c|c|c|}
		\hline
		& \multicolumn{3}{c|}{RMSE} &  Number of & Fraction \\
		Model & Mean & 10th perc. & 90th perc. &  Parameters & in Bid/Ask\\
		\hline
		Edgeworth & \textbf{0.4255} & 0.2829 & \textbf{0.6026} & 8 & \textbf{0.820} \\
		Rough Heston & 1.6813 & 0.9099 & 2.5123 & \textbf{4} & 0.315 \\
		2F Heston Merton & 0.4700 & \textbf{0.2789} & 0.7387 & 17 & 0.783 \\
		2F Heston Merton (Feller) & 0.6277 & 0.3038 & 1.2289 & 17 & 0.676 \\
		Rough Heston Merton & 0.4813 & 0.2892 & 0.7007 & 7 & 0.752 \\
		Heston Merton & 0.4825 & 0.2906 & 0.7122 & 8 & 0.760 \\
		\hline
	\end{tabular}
	\caption{Model comparisons when pricing 0DTEs from May 6, 2022 to May 11, 2023. The RMSE is computed as in Eq. \eqref{def:rmse} with $|\mathcal{T}|=1$. Data source: CBOE.}
	\label{tab:one_maturity_compa}
\end{table}
Similarly, adding a displacement to \textit{2F Heston Merton} improves its pricing performance on the term structure of ATM implied volatilities. While the RMSE average gain induced by the displacement is a noteworthy 30\%, \textit{Edgeworth++} - again - continues to fare better (with a considerably lower number of parameters).  Once more, the flexibility of the nonparametric volatility dynamics in \textit{Edgeworth++} outperforms the more constrained dynamics in the reported diffusive affine models, irrespective of the number of factors (i.e., whether 1 or 2).

One subtle issue with typical affine stochastic volatility models is that the volatility of volatility is constant and, as a consequence, the volatility of the volatility of volatility is zero. This restriction is inconsistent with recent work on the pricing of volatility-of-volatility risk (e.g., \citealp{huang2019volatility}, and \citealp{chen2022volatility}). More specifically in our case, it is at odds with the role that the volatility of volatility plays in the pricing of short-tenor options when the volatility dynamics are unrestricted, like in \textit{Edgeworth++} (c.f. the $\eta$ term in Theorem \ref{th:expansion}).      

\begin{table}[H]
	\centering
	\begin{tabular}{|l|c|c|c|c|c|}
		\hline
		& \multicolumn{3}{c|}{RMSE} &  Number of & Fraction \\
		Model & Mean & 10th perc. & 90th perc. &  Parameters & in Bid/Ask\\
		\hline
		Edgeworth++ & \textbf{1.0195} & \textbf{0.7071} & \textbf{1.4234} & 13 & \textbf{0.383} \\
		Rough Heston++ & 2.7668 & 2.1327 & 3.4438 & \textbf{9} & 0.097 \\
		2F Heston Merton & 1.8826 & 1.1030 & 2.8391 & 17 & 0.301 \\
		2F Heston Merton (Feller) & 2.5959 & 1.4897 & 3.7923 & 17 & 0.269 \\
		Rough Heston Merton++ & 1.1835 & 0.7736 & 1.7887 & 12 & 0.328 \\
		Heston Merton++ & 1.8843 & 0.9596 & 3.2450 & 13 & 0.266 \\
		2F Heston Merton++ & 1.2702 & 0.7080 & 2.0077 & 22 & 0.318 \\
		\hline
	\end{tabular}
	\caption{Model comparison in pricing ultra-short tenors between May 2022 and May 2023. The RMSE is computed as in \eqref{def:rmse} with $|\mathcal{T}|=6$. We employed the SINC method with $N_{f}=10^4$ nodes for the fast Fourier approximation. Data source: CBOE}
	\label{tab:multi_maturity_compa}
\end{table}

\section{Spot volatility estimation}\label{spot}
In a final comparison, we evaluate the ability of competing pricing models to estimate spot volatility. We recall that spot volatility is the same under the $\mathbb{P}$ and the $\mathbb{Q}$ measure.

Panel A of Fig. \ref{fig:spotvol_surf} compares spot volatility estimates obtained from \textit{Edgeworth++} fitted on the entire surface and only on 0DTEs (at 10:30, as earlier). The comparison is revealing since \citet{BFR:24} have shown that the spot volatility estimates obtained from 0DTEs are essentially unbiased and efficient. \citet{BFR:24} have also shown that, when the tenor lengthens, the option-implied estimates deteriorate. Panel A, instead, documents that, thanks to the introduction of a displacement, \textit{Edgeworth++} yields spot volatility estimates which are similar to those obtained using the shortest tenor only and are considerably more precise than those based on time-series data. The time-series estimates are derived from fitting an EGARCH model with intraday effect on the 1-minute time series of SPY, the most liquid S\&P 500 ETF, as in \citet{BFR:24}.

Panel B and C of of Fig. \ref{fig:spotvol_surf} focus on the spot volatility estimates obtained with the two benchmark models discussed in Section \ref{sec:numerical_results} and their enhanced versions in Section \ref{impbench}. \textit{2F Heston Merton} delivers reasonable estimates, albeit with a negative bias relative to the reference 0DTE-implied values (Panel B). The displacement (in \textit{2F Heston Merton++}) essentially removes the bias (Panel C). Similarly, the \textit{Rough Heston++} estimates have an obvious positive bias in the absence of price discontinuities (Panel B). Adding jumps (in \textit{Rough Heston Merton++}) removes the bias but preserves somewhat the jaggedness of the resulting estimates (Panel C). 

\begin{figure}[h!]
	\centering
	\begin{tabular}{cc}
		Panel A & Panel B \\
		\includegraphics[width=0.4\linewidth]{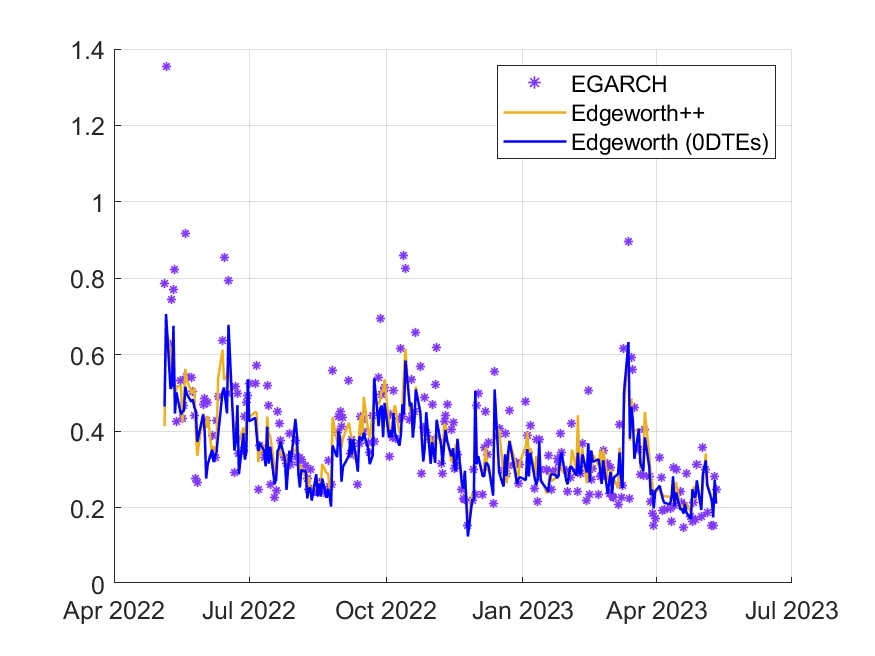} &
		\includegraphics[width=0.4\linewidth]{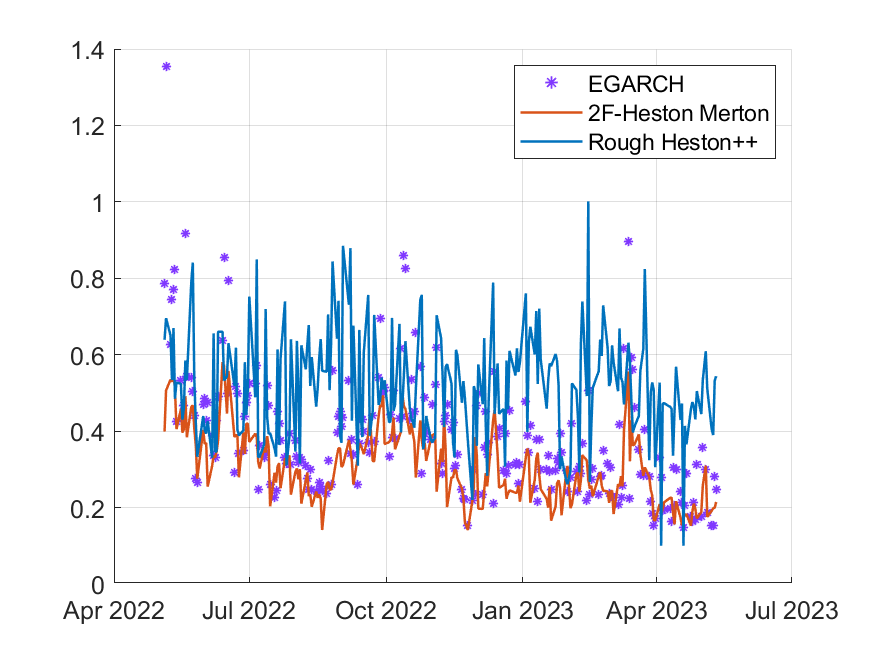} \\[0.5em]
		
		\multicolumn{2}{c}{Panel C} \\
		\multicolumn{2}{c}{
			\includegraphics[width=0.4\linewidth]{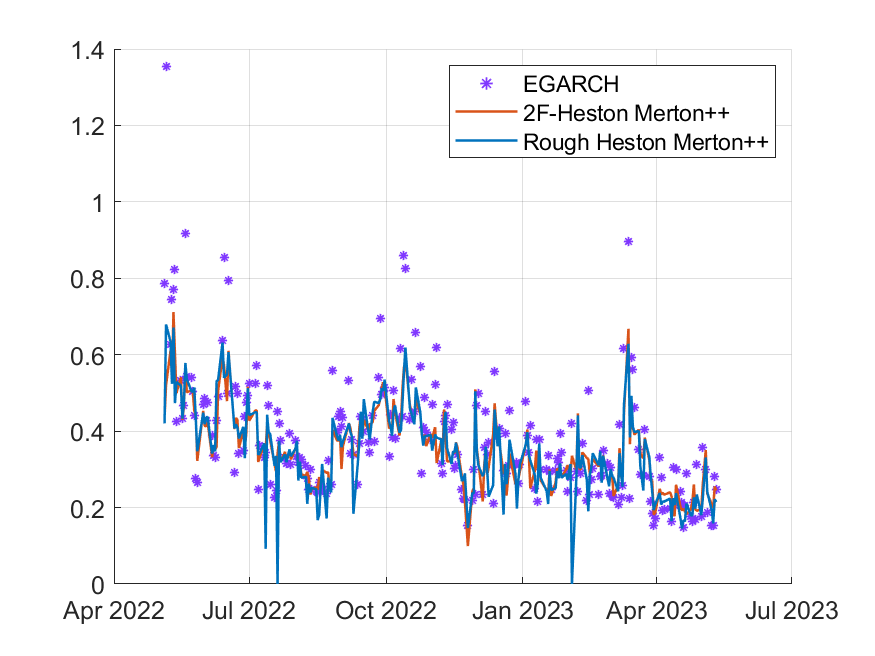}
		}
	\end{tabular}
	\caption{Comparison of spot volatility estimates obtained at 10:30 on every day from May 6, 2022 to May 11, 2023. Panel A reports time-series estimates obtained with an EGARCH model, estimates from the \textit{Edgeworth} model estimated on 0DTEs only, and estimates from the \textit{Edgeworth++} model estimated on surfaces with 6 maturities. Panel B reports time-series estimates obtained with an EGARCH model, estimates from the \textit{Rough Heston++} model and estimates from the \textit{2F Heston Merton} model. Panel C reports time-series estimates obtained with an EGARCH model, estimates from the \textit{Rough Heston Merton++} model and estimates from the \textit{2F Heston Merton++} model. Data source: CBOE.}
	\label{fig:spotvol_surf}
\end{figure}

Fig. \ref{fig:hist_spot_vol}, Panel A, compares the distributions of the absolute estimation error (relative to the 0DTE estimates) of spot volatility estimates obtained from \textit{2F Heston Merton} and \textit{Edgeworth++}. Fig. \ref{fig:hist_spot_vol}, Panel B, does the same after adding a displacement to \textit{2F Heston Merton}. As suggested above, the displacement removes the negative bias of the spot volatility estimates associated with \textit{2F Heston Merton}. In both cases, however, \textit{Edgeworth++} appears to be more efficient. 

On the whole, we interpret these results as evidence of the positive role that the displacement may play for inference, in addition to pricing.

\begin{figure}[h!]
	\centering
	\begin{tabular}{cc}
Panel A & Panel B \\
\includegraphics[width=0.5\linewidth]{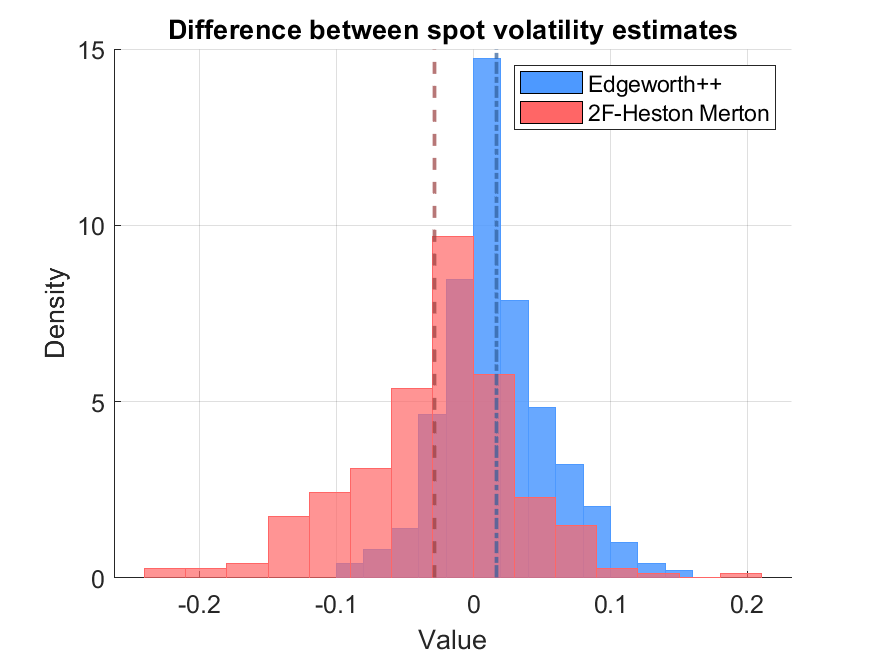} & 
\includegraphics[width=0.5\linewidth]{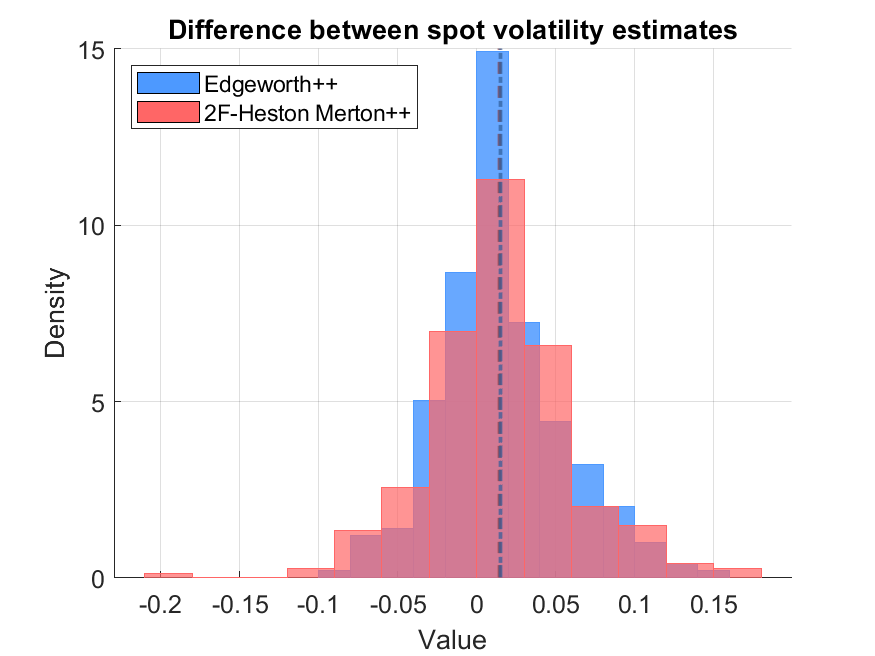}
\end{tabular}
	\caption{Panel A: Histogram of the absolute error between the spot volatility estimates from 0DTEs and the spot volatility estimates from \textit{Edgeworth++} and \textit{2F Heston Merton}, respectively. Panel B: Histogram of the absolute error between the spot volatility estimates from 0DTEs and the spot volatility estimates from \textit{Edgeworth++} and \textit{2F Heston Merton++}, respectively. Data source: CBOE.}
	\label{fig:hist_spot_vol}
\end{figure}

\section{Concluding remarks}\label{sec:conclusions} 
Pricing implied-volatility smiles over \textit{ultra-short-term} tenors, while capturing \textit{oscillating} ATM implied-volatility term structures, is subtle. The introduction of price discontinuities and a displacement may put alternative models on an equal footing, but the specificities of each model's volatility dynamics are bound to yield pricing differences across short tenors.

Affine volatility models (whether \textit{diffusive} or \textit{rough}) have, for good reasons, been extraordinarily successful. We, however, find that, when pricing \textit{ultra-short-term} tenors, they may be excessively constrained. Affine univariate \textit{rough} models are very parsimonious but the cost of parsimony is implied volatilities which may not easily adapt to the full smiles observed in the data for broad levels of moneyness. Affine multivariate diffusive models may better capture the salient features of these smiles, but at the cost of substantial proliferation of parameters. A typical constraint in affine models is a deterministic volatility of volatility.

We have shown that \textit{Edgeworth++}, a specification with a nonparametric stochastic volatility factor and a deterministic displacement factor, has both computational and pricing benefits - relative to natural competitors - without requiring an excessive number of parameters. Leaving the displacement aside, we have also discussed how the \textit{Edgeworth++} \enquote{parameters} are, in fact, time-specific realizations of essentially unrestricted processes. The volatility parameters in the competitor models (i.e., affine parametric specifications, whether univariate or multivariate, rough or diffusive) are, instead, re-estimated to optimize model performance at each evaluation time, thereby ignoring that genuine parameters should be held stable over time. In this sense, \textit{Edgeworth++} is economically motivated and internally coherent.

\pagebreak

\bibliographystyle{chicago}

\newpage
\appendix

\section{Proofs}\label{sec:appendix}
\subsection{Preliminaries}
\singlespacing
We begin with Definition 2 in \citet{BR:24}. The definition presents the notion of differentiability of the price process used in the article.
\begin{definition}\label{itodiff_k}
 An adapted, real stochastic $W$-It\^o process $\chi$ is $k$-times $W$-differentiable (we write $\chi\in\mathcal D_W^{(k)}$) if it admits the representation
\begin{align*}
d\chi_t &= a^{\chi}_t dt + b^{\chi}_t dW_t, \\
d\mathcal S_t^{j} &= \mathcal{A}^{\chi(j)}_t dt + \mathcal{B}^{\chi(j)}_t dW_t,
\end{align*}
where the $2^{j}$-dimensional vector process $\mathcal S_t^{j}$ is defined as $\mathcal S_t^{j} = (\mathcal{A}^{\chi(j-1)}_t ,\mathcal{B}^{\chi(j-1)}_t )^\top,$  \sloppy\mbox{$j=1,\ldots, k$} with $(\mathcal{A}_t^{\chi(0)} ,\mathcal{B}_t^{\chi(0)})^\top = (a^{\chi}_t ,b^{\chi}_t )^\top$ being adapted and c\`adl\`ag. 
\end{definition} 

We now turn to a lemma which will be used extensively in the characteristic function expansion with displacement in the proof of Theorem \ref{th:expansion}.

\begin{lemma}\label{lemma1:extended}
    Let $\chi\in \mathcal D_W^{(1)},$  in the sense of Definition \ref{itodiff_k}, and let $a:\R\to\R$ be a deterministic function. Then, for any $t>0$ and $u\in\R$, the extended $W$-transform of $\chi$ via $a$ satisfies 
    \begin{align*}
        \E[e^{iu\int_0^t a_s d W_s}\chi_t]=e^{-\frac{u^2}{2}\int_0^t a_s^2 ds}\left(\chi_0 + \int_0^t e^{\frac{u^2}{2}\int_0^s a_u^2 du}\E[e^{iu\int_0^{s}a_u dW_u}(d\chi_s + iu a_s d[\chi,W]_s)]\right).
    \end{align*}
\end{lemma}
\begin{proof}
    
By It\^o's Lemma, notice that 
\begin{equation*}
    d(e^{iu\int_0^s a_u d W_u})=iu a_s e^{iu \int_0^s a_u d W_u} d W_s - \frac{1}{2}u^2 a_{s}^2 e^{iu \int_0^s a_u dW_u}d s.
\end{equation*}
By a second application of It\^o's Lemma, we have
\begin{align*}
    e^{iu\int_0^t a_s d W_s}\chi_t &= \chi_0 + \int_0^t e^{iu\int_0^s a_u d W_u}d \chi_s +\int_0^t \chi_s d(e^{iu\int_0^s a_u d W_u})+\int_0^t d[\chi,e^{iu\int_0^\cdot a_u d W_u}]_s\\
    &= \chi_0 + \int_0^t e^{iu\int_0^s a_u d W_u}d \chi_s+iu\int_0^t \chi_s a_s e^{iu\int_0^s a_u dW_u} dW_s-\frac{1}{2}u^2\int_0^t a_s^2 e^{iu\int_0^s a_u dW_u} ds\\
    &\qquad+\int_0^t \left(iu a_s e^{iu \int_0^s a_u d W_u} \right) d[\chi,W]_s .   
\end{align*}
Taking expectations, we obtain
\begin{equation*}
    \E[e^{iu\int_0^t a_s d W_s}\chi_t]= \chi_0-\frac{1}{2}u^2\int_0^t a_s^2\E[e^{iu\int_0^s a_u dW_u}]ds+\int_0^t \E[e^{iu\int_0^s a_u dW_u}(d\chi_s +iu a_s d[\chi,W]_s)].
\end{equation*}
Solving the ODE as a function of $t>0,$ for a fixed $u\in\R$ with initial value $\chi_0$, yields the result.

\end{proof}

\subsection{Proof of Theorem \ref{th:expansion}}\label{appendix:expansion_proof}
We are interested in the characteristic function of 
\begin{equation*}
    Z^c_{\tau}=\frac{X^c_\tau-X^c_0-\mu_0\tau}{\sigma_0\sqrt{\tau}},
\end{equation*}
for $\tau>0$, with $X^c$ satisfying Eq. \eqref{dyn}. Write
\begin{align}\nonumber
    \E[e^{iu Z^c_\tau}]&=\E[e^{iu \frac{W_\tau}{\sqrt{\tau}}}e^{iu \sqrt{\tau}Y_\tau}]\\ \nonumber
    &=\E[e^{iu \frac{W_\tau}{\sqrt{\tau}}}e^{iu \frac{1}{\sigma_0\sqrt{\tau}}(\int_0^\tau( \mu_s-\mu_0 )ds+\int_0^\tau (\sigma_s-\sigma_0)dW_s)}]\\ \nonumber
    &=\E[e^{iu \frac{W_\tau}{\sqrt{\tau}}}e^{iu \frac{1}{\sigma_0\sqrt{\tau}}(\int_0^\tau( \mu_s-\mu_0 )ds+\int_0^\tau \phi(s)dW_s+\int_0^\tau(\int_0^s\alpha_u du +\beta_u dW_u +\beta_u^{\prime}dW_u^\prime)dW_s)}]\\ \nonumber
    &=\E[e^{iu \frac{W_\tau}{\sqrt{\tau}}}e^{iu \frac{1}{\sigma_0\sqrt{\tau}}\int_0^\tau \phi(s)dW_s}e^{iu\frac{1}{\sigma_0\sqrt{\tau}}(\int_0^\tau( \mu_s-\mu_0 )ds+\int_0^\tau(\int_0^s\alpha_u du +\beta_u dW_u +\beta_u^{\prime}dW_u^\prime)dW_s)}]\\ \nonumber
    &=\E[e^{iu \frac{\int_0^\tau \tilde{\phi}(s)dW_s}{\sqrt{\tau}}} e^{iu\sqrt{\tau}Y_{\tau}^{\prime}}]\\ \nonumber
    &=\E[e^{iu \frac{\int_0^\tau \tilde{\phi}(s)dW_s}{\sqrt{\tau}}}(1+\sum_{k\ge1}\frac{(iu)^k}{k\!}(Y_{\tau}^{\prime})^{k}\tau^{k/2})]\\ \nonumber
    &=\E[e^{iu \frac{\int_0^\tau \tilde{\phi}(s)dW_s}{\sqrt{\tau}}}(1+iuY_{\tau}^{\prime}\sqrt{\tau}-\frac{u^{2}}{2}Y_{\tau}^{\prime}\tau + \dots)]\\ \nonumber
    &=e^{-\frac{u^2}{2}\frac{\int_0^{\tau}\tilde{\phi}^{2}(s)ds}{\tau}}+\underbrace{\E[e^{iu\frac{\int_0^\tau\tilde{\phi}_s dW_s}{\sqrt{\tau}}}iu Y_\tau^{\prime}\sqrt{\tau}]}_{A_\tau}+\underbrace{\E[e^{iu\frac{\int_0^\tau\tilde{\phi}(s) dW_s}{\sqrt{\tau}}}\frac{1}{2}(iu)^2 (Y_\tau^{\prime})^2\tau]}_{B_{\tau}}+\\ \nonumber
    &\quad+\underbrace{\E[e^{iu\frac{\int_0^\tau\tilde{\phi}_s dW_s}{\sqrt{\tau}}}\sum_{k\ge3}\frac{(iu)^k}{k\!}(Y_{\tau}^{\prime})^{k}\tau^{k/2}]}_{C_\tau},
\end{align}
where we introduced the notation $\tilde{\phi}(s)=1+\phi(s)/\sigma_0$ and 
\begin{equation*}
    Y_{\tau}^{\prime}=\frac{1}{\sigma_0\tau}\left(\int_0^\tau( \mu_s-\mu_0 )ds+\int_0^\tau\left(\int_0^s\alpha_u du +\beta_u dW_u +\beta_u^{\prime}dW_u^\prime\right)dW_s\right).
\end{equation*}
Let us start with the first term, i.e.,
\begin{align*}
    A_{\tau}=\E[e^{iu\frac{\int_0^\tau\tilde{\phi}(s) dW_s}{\sqrt{\tau}}}iu Y_\tau^{\prime}\sqrt{\tau}]&=\frac{iu}{\sigma_0\sqrt{\tau}}\E\Big[e^{iu\frac{\int_0^\tau\tilde{\phi}(s) dW_s}{\sqrt{\tau}}}\left(\int_0^\tau(\mu_s-\mu_0)ds\right)\Big]\\
    &\quad + \frac{iu}{\sigma_0\sqrt{\tau}}\E\Big[e^{iu\frac{\int_0^\tau\tilde{\phi}(s) dW_s}{\sqrt{\tau}}}\int_0^\tau\left(\int_0^s\alpha_u du +\beta_u dW_u +\beta_u^{\prime}dW_u^\prime\right)dW_s\Big]\\
    &=A_{\tau,0}+\frac{iu}{\sigma_0\sqrt{\tau}}\E\Big[e^{iu\frac{\int_0^\tau\tilde{\phi}(s) dW_s}{\sqrt{\tau}}}\int_0^\tau\beta_0 W_s +\alpha_0 s +\\
    &\quad+\left(\int_0^s(\alpha_u-\alpha_0) du +(\beta_u-\beta_0)dW_u +\beta_u^{\prime}dW_u^\prime\right)dW_s\Big] \\
    &= \sum_{j=0}^5 A_{\tau,j}.
\end{align*}
By Lemma \ref{lemma1:extended} with $a = \tilde{\phi}/\sqrt{\tau}$ and $\chi_t=\int_0^t (\mu_s-\mu_0)ds,$ we have
\begin{align*}
    A_{\tau,0}&=\frac{iu}{\sigma_0 \sqrt{\tau}}e^{-\frac{u^2}{2\tau}\int_0^{\tau}\tilde{\phi}^2(s)ds}\int_0^\tau e^{\frac{u^2}{2\tau}\int_0^s \tilde{\phi}^2(u)du}\E[e^{iu\frac{\int_0^s\tilde{\phi}(u)dW_u}{\sqrt{\tau}}}(\mu_s-\mu_0)ds]\\
    &=\frac{iu}{\sigma_0 \sqrt{\tau}}e^{-\frac{u^2}{2\tau}\int_0^{\tau}\tilde{\phi}^2(s)ds}\int_0^\tau e^{\frac{u^2}{2\tau}\int_0^s \tilde{\phi}^2(u)d u}\E[e^{iu\frac{\int_0^s\tilde{\phi}(u)dW_u}{\sqrt{\tau}}}(\delta_0W_s+\\
    &\quad+\int_0^s \gamma_{s_1}ds_1+(\delta_{s_1}-\delta_{0})dW_{s_1}+ \delta_{s_1}^{\prime}dW_{s_1}^{\ast})ds]\\
&=\overline{A}_{\tau,0}^{0}+\overline{A}_{\tau,0}^{1}+\overline{A}_{\tau,0}^{2}+\overline{A}_{\tau,0}^{3}.
\end{align*}
We begin with the dominating term, i.e., $\overline{A}_{\tau,0}^{0}$:
\begin{align*}
    \overline{A}_{\tau,0}^{0}=&\frac{iu\delta_0}{\sigma_0 \sqrt{\tau}}e^{-\frac{u^2}{2\tau}\int_0^{\tau}\tilde{\phi}^2(s)ds}\int_0^\tau e^{\frac{u^2}{2\tau}\int_0^s \tilde{\phi}^2(u)d u} \E\left[e^{iu\frac{\int_0^s\tilde{\phi}(u)dW_u}{\sqrt{\tau}}}W_s\right]ds\\
    =&\frac{iu\delta_0}{\sigma_0 \sqrt{\tau}}e^{-\frac{u^2}{2\tau}\int_0^{\tau}\tilde{\phi}^2(s)ds}\int_0^\tau \int_0^s e^{\frac{u^{2}}{2\tau}\int_0^{s_1} \tilde{\phi}^2(u)du}\E\left[e^{iu\frac{\int_0^{s_1}\tilde{\phi}(u)dW_u}{\sqrt{\tau}}}(dW_{s_1}+\frac{iu}{\sqrt{\tau}}\tilde{\phi}(s_1) ds_1)\right]ds\\
    =&\frac{iu\delta_0}{\sigma_0 \sqrt{\tau}}e^{-\frac{u^2}{2\tau}\int_0^{\tau}\tilde{\phi}^2(s)ds}\int_0^\tau \int_0^s e^{\frac{u^{2}}{2\tau}\int_0^{s_1} \tilde{\phi}^2(u)du}\cdot \\
    & \cdot\left(\E\left[e^{iu\frac{\int_0^{s_1}\tilde{\phi}(u)dW_u}{\sqrt{\tau}}}dW_{s_1}\right] +\E[e^{iu\frac{\int_0^{s_1}\tilde{\phi}(u)dW_u}{\sqrt{\tau}}}\frac{iu}{\sqrt{\tau}}\tilde{\phi}(s_1)]\right)ds_1ds\\
     =&\frac{-u^2\delta_0}{\sigma_0 \tau}e^{-\frac{u^2}{2\tau}\int_0^{\tau}\tilde{\phi}^2(s)ds}\int_0^\tau \int_0^s e^{\frac{u^{2}}{2\tau}\int_0^{s_1} \tilde{\phi}^2(u)du}\E[e^{iu\frac{\int_0^{s_1}\tilde{\phi}(u)dW_u}{\sqrt{\tau}}}]\tilde{\phi}(s_1)ds_1ds\\
     =&\frac{-u^2\delta_0}{\sigma_0 \tau}e^{-\frac{u^2}{2\tau}\int_0^{\tau}\tilde{\phi}^2(s)ds}\int_0^\tau \int_0^s \tilde{\phi}(s_1)ds_1ds,
\end{align*}
where, in the last identity, we used the fact that
\begin{equation*}
    \frac{\int_0^{s_1}\tilde{\phi}(u)dW_u}{\sqrt{\tau}}\sim \mathcal{N}\left(0,\frac{1}{\tau}\int_0^{s_1}\tilde{\phi}^2(u)du\right)
\end{equation*}
before simplifying terms. Next, we turn to $\overline{A}_{\tau,0}^{1}$ and apply Lemma \ref{lemma1:extended} with $a= \tilde{\phi}/\sqrt{\tau}$ and $\chi_{t}=\int_0^t\gamma_s ds:$
\begin{align*}
    \overline{A}_{\tau,0}^{1}=&\frac{iu}{\sigma_0 \sqrt{\tau}}e^{-\frac{u^2}{2\tau}\int_0^{\tau}\tilde{\phi}^2(s)ds}\int_0^\tau e^{\frac{u^2}{2\tau}\int_0^s \tilde{\phi}^2(u)d u} \E\left[e^{iu\frac{\int_0^s\tilde{\phi}(u)dW_u}{\sqrt{\tau}}}\int_0^s \gamma_{s_1}ds_1\right]ds\\
    =&\frac{iu}{\sigma_0 \sqrt{\tau}}e^{-\frac{u^2}{2\tau}\int_0^{\tau}\tilde{\phi}^2(s)ds}\int_0^\tau \int_0^s e^{\frac{u^2}{2\tau}\int_0^{s_{1}}\tilde{\phi}^2(u)du}\E\left[e^{iu\frac{\int_0^{s_1}\tilde{\phi}(u)dW_u}{\sqrt{\tau}}}\gamma_{s_1}\right]ds_{1}ds\\
    =&\underbrace{\frac{iu}{\sigma_0 \sqrt{\tau}}e^{-\frac{u^2}{2\tau}\int_0^{\tau}\tilde{\phi}^2(s)ds}\int_0^\tau \int_0^s e^{\frac{u^2}{2\tau}\int_0^{s_{1}}\tilde{\phi}^2(u)du}\E\left[e^{iu\frac{\int_0^{s_1}\tilde{\phi}(u)dW_u}{\sqrt{\tau}}}(\gamma_{s_1}-\gamma_0)\right]ds_{1}ds}_{\overline{A}_{\tau,0}^{1,a}} \\
    & +\gamma_0\frac{iu}{2\sigma_0}\tau^{3/2} e^{-\frac{u^2}{2\tau}\int_0^{\tau}\tilde{\phi}^2(s)ds}.
\end{align*}
By upper-bounding $\overline{A}_{\tau,0}^{1,a},$ we may conclude that both $\overline{A}_{\tau,0}^{1}$ and $\overline{A}_{\tau,0}^{2}$ are of order $\tau^{3/2}$ and, hence, do not contribute to the expansion up to order 2.
Notice that $\overline{A}_{\tau,0}^{3}=0,$ due to the independence between $W$ and $W^{\prime}$. 

We turn to term $A_{\tau,1}$ and apply Lemma \ref{lemma1:extended} with $\chi_t = \int_0^t W_s dW_s$:
\begin{align*}
    A_{\tau,1}&=\frac{iu\beta_0}{\sigma_0\sqrt{\tau}}\E\Big[e^{iu\frac{\int_0^\tau\tilde{\phi}(s) dW_s}{\sqrt{\tau}}}\int_0^\tau W_s  dW_s\Big]\\
    &=\frac{iu\beta_0}{\sigma_0\sqrt{\tau}}e^{-\frac{u^2}{2\tau}\int_0^{\tau}\tilde{\phi}^2(s)ds}\int_0^\tau e^{\frac{u^2}{2\tau}\int_0^{s}\tilde{\phi}^2(s_1)ds_1}\frac{iu}{\sqrt{\tau}}\tilde{\phi}(s)\E\Big[e^{iu\frac{\int_0^s\tilde{\phi}(s_1) dW_{s_1}}{\sqrt{\tau}}}W_{s}\Big]ds\\
    &=\frac{-u^2\beta_0}{\sigma_0\tau}e^{-\frac{u^2}{2\tau}\int_0^{\tau}\tilde{\phi}^2(s)ds}\int_0^\tau e^{\frac{u^2}{2\tau}\int_0^{s}\tilde{\phi}^2(s_1)ds_1}\tilde{\phi}(s)\E\Big[e^{iu\frac{\int_0^s\tilde{\phi}(s_1) dW_{s_1}}{\sqrt{\tau}}}W_{s}\Big]ds\\
    &=\frac{-iu^3\beta_0}{\sigma_0\tau^{3/2}}e^{-\frac{u^2}{2\tau}\int_0^{\tau}\tilde{\phi}^2(s)ds}\int_0^\tau \tilde{\phi}(s)\int_0^s\tilde{\phi}(s_1)ds_1ds.
\end{align*}
Applying, now, Lemma \ref{lemma1:extended} with $\chi_t=\int_0^t s dW_s$ we have 
\begin{align*}
    A_{\tau,2}&=\frac{iu\alpha_0}{\sigma_0\sqrt{\tau}}\E\left[e^{iu\frac{\int_0^\tau \tilde{\phi}(s)dW_s}{\sqrt{\tau}}}\int_0^\tau s dW_s\right]\\
    &=\frac{iu\alpha_0}{\sigma_0\sqrt{\tau}}e^{-\frac{u^2}{2\tau}\int_0^{\tau}\tilde{\phi}(s)^2ds}\left(\int_0^\tau e^{\frac{u^2}{2\tau}\int_0^{s}\tilde{\phi}(u)^2du}\E[e^{iu\frac{\int_0^{s}\tilde{\phi}(u)dW_{u}}{\sqrt{\tau}}}(sdW_s+iu \frac{\tilde{\phi}(s)}{\sqrt{\tau}} sds)]\right)\\
    &=\frac{-u^2\alpha_0}{\sigma_0\tau}e^{-\frac{u^2}{2\tau}\int_0^{\tau}\tilde{\phi}(s)^2ds}\left(\int_0^\tau e^{\frac{u^2}{2\tau}\int_0^{s}\tilde{\phi}(u)^2du}\E[e^{iu\frac{\int_0^{s}\tilde{\phi}(u)dW_{u}}{\sqrt{\tau}}}]s\tilde{\phi}(s)ds\right)\\
     &=\frac{-u^2\alpha_0}{\sigma_0\tau}e^{-\frac{u^2}{2\tau}\int_0^{\tau}\tilde{\phi}(s)^2ds}\left(\int_0^\tau s\tilde{\phi}(s)ds\right).
\end{align*}
Next, we apply Lemma \ref{lemma1:extended} with $\chi_t=\int_0^t \int_0^s(\alpha_u-\alpha_0)du dW_s$. Recall, first, that
\begin{equation*}
    d[\chi,W]_t= \int_0^t(\alpha_u-\alpha_0)dudW_t d W_t = \int_0^t(\alpha_u-\alpha_0)dudt.
\end{equation*}
We, thus, have
\begin{align*}
    A_{\tau,3}= &\frac{iu}{\sigma_0\sqrt{\tau}}\E\left[e^{iu\frac{\int_0^\tau \tilde{\phi}(s)dW_s}{\sqrt{\tau}}}\int_0^\tau \left(\int_0^s(\alpha_u-\alpha_0)du\right) dW_s\right]\\
    = &\frac{iu}{\sigma_0\sqrt{\tau}}e^{-\frac{u^2}{2\tau}\int_0^{\tau}\tilde{\phi}(s)^2ds}\int_0^\tau e^{\frac{u^{2}}{2\tau}\int_0^s\tilde{\phi}(u)^2du}\times \\
    & \times\E\left[e^{iu\frac{\int_0^{s}\tilde{\phi}(u)dW_{u}}{\sqrt{\tau}}}\left(\int_0^s(\alpha_u-\alpha_0)dudW_s+iu \frac{\tilde{\phi}(s)}{\sqrt{\tau}} \int_0^s(\alpha_u-\alpha_0)duds\right)\right]\\
    =&\frac{iu}{\sigma_0\sqrt{\tau}}e^{-\frac{u^2}{2\tau}\int_0^{\tau}\tilde{\phi}(s)^2ds}\int_0^\tau e^{\frac{u^{2}}{2\tau}\int_0^s\tilde{\phi}(u)^2du}\times \\
    & \times\E\left[e^{iu\frac{\int_0^{s}\tilde{\phi}(u)dW_{u}}{\sqrt{\tau}}}\left(\int_0^s(\alpha_u-\alpha_0)dudW_s+iu \frac{\tilde{\phi}(s)}{\sqrt{\tau}} \int_0^s(\alpha_u-\alpha_0)duds\right)\right]\\
   =&-\frac{u^2}{\sigma_0\tau}e^{-[\frac{u^2}{2\tau}\int_0^{\tau}\tilde{\phi}(s)^2ds}\int_0^\tau e^{\frac{u^{2}}{2\tau}\int_0^s\tilde{\phi}(u)^2du}\E[e^{iu\frac{\int_0^{s}\tilde{\phi}(u)dW_{u}}{\sqrt{\tau}}}\int_0^s(\alpha_u-\alpha_0)du]\tilde{\phi}(s) ds.
\end{align*}
We may now use the same argument as for $\overline{A}_{\tau,0}^2$ and find a contribution to the third order of the expansion (i.e., in $\tau^{3/2}$). As for $A_{\tau,4}$, notice that the dominant term is the following
     \begin{align*}
        \overline{A}_{\tau,4}^{0}&=\frac{-iu^3\eta_0}{\sigma_0 \tau^{3/2}}e^{-\frac{u^2}{2\tau}\int_0^{\tau}\tilde{\phi}(u)^2du}\int_0^\tau \int_0^s \intphitildesqplus{s_1}\E [W_{s_1}e^{iu\frac{\int_0^{s_1}\tilde{\phi}(u)dW_u}{\sqrt{\tau}}}]d s_1 d s\\
         &=\frac{-iu^3\eta_0}{\sigma_0 \tau^{3/2}}\frac{iu}{\sqrt{\tau}}e^{-\frac{u^2}{2\tau}\int_0^{\tau}\tilde{\phi}(u)^2du}\int_0^\tau \int_0^s \left(\int_0^{s_1}\tilde{\phi}(s_2)d s_2\right)d s_1 d s\\
         &=\frac{u^4\eta_0}{\sigma_0 \tau^{2}}e^{-\frac{u^2}{2\tau}\int_0^{\tau}\tilde{\phi}(u)^2du}\int_0^\tau \int_0^s \left(\int_0^{s_1}\tilde{\phi}(s_2)d s_2\right)d s_1 d s.
     \end{align*}
This term contributes to the second order of the expansion (in $\tau$), while the terms $\overline{A}_{\tau,4}^{i}$ for $i=1,2$ are of order $\tau^{3/2}$ and $\overline{A}_{\tau,4}^{3}=0$. Recall, also, that the expression of $\overline{A}_{\tau,4}^{0}$ is consistent with the case $\phi(t)= 0$ as
     \begin{equation*}
\overline{A}_{\tau,4}^{0}=\frac{u^4\eta_0}{\sigma_0 \tau^{2}}e^{-\frac{u^2}{2}}\int_0^\tau \int_0^s \left(\int_0^{s_1}1d s_2\right)d s_1 d s=\frac{u^4\eta_0}{\sigma_0 \tau^{2}}e^{-\frac{u^2}{2}}\frac{\tau^3}{6}=\frac{u^4\eta_0}{6\sigma_0 }\tau e^{-\frac{u^2}{2}}.
     \end{equation*}
Finally, $A_{\tau,5}=0$, by the independence of $W$ and $W^\prime$.

In order to study the second order of the expansion, we add and subtract - as earlier - the initial values of $\alpha_0$, $\beta_0$ and $\beta_0^{\prime}$ to center the corresponding processes. We have
\begin{align*}
  B_{\tau}=-\frac{u^2}{2\sigma_0^2\tau}\E\left[\left(\int_0^\tau (\mu_s - \mu_0) ds +\int_0^\tau (\sigma_s -\sigma_0) d W_s\right)^2 e^{iu\frac{\int_0^\tau \tilde{\phi}(s)dW_s}{\sqrt{\tau}}}\right]  =\sum_{i=1}^{10} B_{\tau,i},
\end{align*}
where the $B_{\tau,i}$'s are numbered as in \cite{BR:24}.
Let $\chi_t=\int_0^t s d W_s$ and recall that
\begin{equation*}
    d \chi_t^2 = 2 \chi_t t dW_t + t^2 d t,\qquad d \chi_t t= t d\chi_t + \chi_t d t. 
\end{equation*}
We have
\begin{align*}
   B_{\tau,1}&=-\frac{\alpha_0^2 u^2}{2\sigma_0^2\tau}\E\left[\left(\int_0^\tau s d W_s\right)^2 e^{iu\intphitilde{s}}\right]\\
   &=-\frac{\alpha_0^2 u^2}{2\sigma_0^2\tau}\intphitildesqminus{s}\left(\int_0^\tau \intphitildesqplus{s}\E[e^{iu\frac{\int_0^s\tilde{\phi}(u)dW_u}{\sqrt{\tau}}}(d\chi_s^2 + iu\frac{\tilde{\phi}(s)}{\sqrt{\tau}} d[\chi^2,W]_s)]\right)\\
    &=-\frac{\alpha_0^2 u^2}{2\sigma_0^2\tau}\intphitildesqminus{s}\left(\int_0^\tau \intphitildesqplus{s}\E[e^{iu\frac{\int_0^s\tilde{\phi}(u)dW_u}{\sqrt{\tau}}}(2 \chi_s s dW_s + s^2 d s + 2iu\frac{\tilde{\phi}(s)}{\sqrt{\tau}}\chi_s s ds ]\right)\\
        &=-\frac{\alpha_0^2 u^2}{2\sigma_0^2\tau}\intphitildesqminus{s}\left(\int_0^\tau \intphitildesqplus{s}\E[e^{iu\frac{\int_0^s\tilde{\phi}(u)dW_u}{\sqrt{\tau}}}]s^2 d s\right)\\
        &\quad-\frac{i\alpha_0^2 u^3}{\sigma_0^2\tau^{3/2}}\intphitildesqminus{s}\left(\int_0^\tau \intphitildesqplus{s}\E[e^{iu\frac{\int_0^s\tilde{\phi}(u)dW_u}{\sqrt{\tau}}}\chi_s ]s\tilde{\phi}(s)  d s\right)\\
        &=-\frac{\alpha_0^2 u^2}{2\sigma_0^2\tau}\intphitildesqminus{s}\frac{\tau^3}{3}-\frac{i\alpha_0^2 u^3}{\sigma_0^2\tau^{3/2}}\intphitildesqminus{s}\left(\int_0^\tau \intphitildesqplus{s}\E[e^{iu\frac{\int_0^s\tilde{\phi}(u)dW_u}{\sqrt{\tau}}}\chi_s ]s\tilde{\phi}(s)  d s\right)\\
        &=-\frac{\alpha_0^2 u^2}{2\sigma_0^2\tau}\intphitildesqminus{s}\frac{\tau^3}{3}-\frac{i\alpha_0^2 u^3}{\sigma_0^2\tau^{3/2}}\intphitildesqminus{s}\frac{iu}{\sqrt{\tau}}\int_0^\tau\tilde{\phi}(s)s\int_0^s\tilde{\phi}(s_1)s_1 ds_1  ds\\
        &=\frac{\alpha_0^2}{\sigma_0^2}\intphitildesqminus{s}\left(-\frac{u^2}{6}\tau^2-\frac{iu^3}{\tau^{3/2}}\frac{iu}{\sqrt{\tau}}\int_0^\tau\tilde{\phi}(s)s\int_0^s\tilde{\phi}(s_1)s_1 ds_1 ds\right)\\
        &=\frac{\alpha_0^2}{\sigma_0^2}\intphitildesqminus{s}\left(\int_0^\tau \tilde{\phi}(s)s\left(\int_0^s\tilde{\phi}(s_1)s_1 ds_1\right) ds \cdot \frac{u^4}{\tau^2}-\frac{u^2}{6}\tau^2\right)
\end{align*}
and
\begin{align*}
    B_{\tau,2}=-\frac{\beta_0^2 u^2}{2\sigma_0^2\tau}\E\left[\left(\int_0^\tau W_s d W_s\right)^2 e^{iu\intphitilde{s}}\right]&=-\frac{\beta_0^2 u^2}{8\sigma_0^2\tau}\E\left[\left(W_\tau^2-\tau\right)^2 e^{iu\intphitilde{s}}\right]\\
    &=-\frac{\beta_0^2 u^2}{8\sigma_0^2\tau}\E\left[\left(W_\tau^4-2W_\tau^2\tau+\tau^2\right) e^{iu\intphitilde{s}}\right]\\
    &=\overline{B}_{\tau,2}^{0}+\overline{B}_{\tau,2}^{1}+\overline{B}_{\tau,2}^{2},
\end{align*}
where 
\begin{align*}
    \overline{B}_{\tau,2}^{2}=-\frac{\beta_0^2 u^2}{8\sigma_0^2\tau}\tau^2\E\left[ e^{iu\intphitilde{s}}\right]=-\frac{\beta_0^2 u^2}{8\sigma_0^2}\tau e^{-\frac{u^2}{2\tau}\int_0^\tau\tilde{\phi}(s)ds}.
\end{align*}
For any $n\ge2$, another application of It\^o's lemma, used below, gives
\begin{equation*}
    dW_t^n= \frac{1}{2}n(n-1)W_t^{n-2}dt+n W_t^{n-1}dW_t.
\end{equation*}
Recall
\begin{align*}
    \E\left[W_{s} e^{iu\frac{\int_0^{s}\tilde{\phi}(u)dW_s}{\sqrt{\tau}}}\right]&=e^{-\frac{u^2}{2\tau}\int_0^{s} \tilde{\phi}(u)^2 du}\left(\int_0^s \intphitildesqplus{s_1}\E[e^{iu\frac{\int_0^{s_1}\tilde{\phi}(u)dW_u}{\sqrt{\tau}}}]\frac{iu}{\sqrt{\tau}}\tilde{\phi}(s_1)d s_1\right)\\
    &=\frac{iu}{\sqrt{\tau}}e^{-\frac{u^2}{2\tau}\int_0^{s} \tilde{\phi}(u)^2 du}\left(\int_0^s\tilde{\phi}(s_1)d s_1\right).
\end{align*}
Write
\begin{align*}
    \overline{B}_{\tau,2}^{1}&=\frac{\beta_0^2 u^2}{4\sigma_0^2}\E\left[W_{\tau}^2 e^{iu\intphitilde{s}}\right]\\
    &=\frac{\beta_0^2 u^2}{4\sigma_0^2}\left(e^{-\frac{u^2}{2\tau}\int_0^\tau \tilde{\phi}(u)^2 du}\left(\int_0^\tau \intphitildesqplus{s}\E[e^{iu\frac{\int_0^s\tilde{\phi}(u)dW_u}{\sqrt{\tau}}}(2W_s dW_s +ds+2iu\frac{\tilde{\phi}(s)}{\sqrt{\tau}}W_s ds)]\right)\right)\\
    &=\frac{\beta_0^2 u^2}{4\sigma_0^2}\left(e^{-\frac{u^2}{2\tau}\int_0^\tau \tilde{\phi}(u)^2 du}\left(\int_0^\tau \intphitildesqplus{s}\E[e^{iu\frac{\int_0^s\tilde{\phi}(u)dW_u}{\sqrt{\tau}}}(1+2iu\frac{\tilde{\phi}(s)}{\sqrt{\tau}}W_s)]ds\right)\right)\\
     &=\frac{\beta_0^2 u^2}{4\sigma_0^2}\left(e^{-\frac{u^2}{2\tau}\int_0^\tau \tilde{\phi}(u)^2 du}\left(\tau+2\frac{iu}{\sqrt{\tau}}\int_0^\tau \intphitildesqplus{s}\E[e^{iu\frac{\int_0^s\tilde{\phi}(u)dW_u}{\sqrt{\tau}}}W_s)]\tilde{\phi}(s)ds\right)\right)\\
     &=\frac{\beta_0^2 u^2}{4\sigma_0^2}\left(e^{-\frac{u^2}{2\tau}\int_0^\tau \tilde{\phi}(u)^2 du}\left(\tau+2\frac{(iu)^2}{\tau}\int_0^\tau \int_0^s \intphitildesqplus{s_1}\E[e^{iu\frac{\int_0^{s_1}\tilde{\phi}(u)d W_u}{\sqrt{\tau}}}]\tilde{\phi}(s_1)d s_1\tilde{\phi}(s)ds\right)\right)\\
     &=\frac{\beta_0^2 u^2}{4\sigma_0^2}\left(e^{-\frac{u^2}{2\tau}\int_0^\tau \tilde{\phi}(u)^2 du}\left(\tau-2\frac{u^2}{\tau}\int_0^\tau \int_0^s \tilde{\phi}(s_1)d s_1\tilde{\phi}(s)ds\right)\right),
\end{align*}
and
\begin{align*}
    \overline{B}_{\tau,2}^{0}&=-\frac{\beta_0^2 u^2}{8\sigma_0^2 \tau}\E\left[W_{\tau}^4 e^{iu\intphitilde{s}}\right]\\
    &=-\frac{\beta_0^2 u^2}{8\sigma_0^2 \tau}\left(e^{-\frac{u^2}{2\tau}\int_0^\tau \tilde{\phi}(u)^2 du}\left(\int_0^\tau \intphitildesqplus{s}\E[e^{iu\frac{\int_0^s\tilde{\phi}(u)dW_u}{\sqrt{\tau}}}(4W_s^3 dW_s +6W_s^2 ds+4iu\frac{\tilde{\phi}(s)}{\sqrt{\tau}}W_s^3 ds)]\right)\right)\\
    &=-\frac{\beta_0^2 u^2}{8\sigma_0^2 \tau}\left(e^{-\frac{u^2}{2\tau}\int_0^\tau \tilde{\phi}(u)^2 du}\left(\int_0^\tau \intphitildesqplus{s}\E[e^{iu\frac{\int_0^s\tilde{\phi}(u)dW_u}{\sqrt{\tau}}}(6W_s^2+4iu\frac{\tilde{\phi}(s)}{\sqrt{\tau}}W_s^3)]ds\right)\right)\\
    &=-\frac{\beta_0^2 u^2}{8\sigma_0^2 \tau}(e^{-\frac{u^2}{2\tau}\int_0^\tau \tilde{\phi}(u)^2 du}(6\int_0^\tau \intphitildesqplus{s}\E[e^{iu\frac{\int_0^s\tilde{\phi}(u)dW_u}{\sqrt{\tau}}}W_s^2]ds+\\
    &\quad+\frac{4iu}{\sqrt{\tau}}\int_0^\tau \intphitildesqplus{s}\E[e^{iu\frac{\int_0^s\tilde{\phi}(u)dW_u}{\sqrt{\tau}}}W_s^3]\tilde{\phi}(s)ds
    ))\\
    &=-\frac{\beta_0^2 u^2}{8\sigma_0^2 \tau}(e^{-\frac{u^2}{2\tau}\int_0^\tau \tilde{\phi}(u)^2 du}(6\int_0^\tau \left(s-2\frac{u^2}{\tau}\int_0^{s} \int_0^{s_1} \tilde{\phi}(s_2)d s_2\tilde{\phi}(s_1)d s_1\right)ds+\\
    &\quad+\frac{4iu}{\sqrt{\tau}}\int_0^\tau \int_0^s \intphitildesqplus{s_1}\E[e^{iu\frac{\int_0^{s_1}\tilde{\phi}(u)dW_u}{\sqrt{\tau}}}(3W_{s_1}^2dW_{s_1}+3W_{s_1}d s_{1}+3iu\frac{\tilde{\phi}(s_1)}{\sqrt{\tau}}W_{s_1}^2d s_1)]\tilde{\phi}(s)ds
    ))\\
    &=-\frac{\beta_0^2 u^2}{8\sigma_0^2 \tau}(e^{-\frac{u^2}{2\tau}\int_0^\tau \tilde{\phi}(u)^2 du}(6\frac{\tau^2}{2}-\frac{12 u^2}{\tau}\int_0^\tau \int_0^{s} \int_0^{s_1} \tilde{\phi}(s_2)d s_2\tilde{\phi}(s_1)d s_1 ds+\\
    &\quad+\frac{12iu}{\sqrt{\tau}}\int_0^\tau \int_0^s \intphitildesqplus{s_1}\E[e^{iu\frac{\int_0^{s_1}\tilde{\phi}(u)dW_u}{\sqrt{\tau}}}W_{s_1}]d s_{1}\tilde{\phi}(s)ds
    \\
    &\quad+\frac{12(iu)^2}{\tau}\int_0^\tau \int_0^s \intphitildesqplus{s_1}\E[e^{iu\frac{\int_0^{s_1}\tilde{\phi}(u)dW_u}{\sqrt{\tau}}}W_{s_1}^2]\tilde{\phi}(s_1)d s_1\tilde{\phi}(s)ds
    ))\\
    &=-\frac{\beta_0^2 u^2}{8\sigma_0^2 \tau}(e^{-\frac{u^2}{2\tau}\int_0^\tau \tilde{\phi}(u)^2 du}(3\tau^2-\frac{12 u^2}{\tau}\int_0^\tau \int_0^{s} \int_0^{s_1} \tilde{\phi}(s_2)d s_2\tilde{\phi}(s_1)d s_1 ds+\\
    &\quad-\frac{12u^2}{\tau}\int_0^\tau \int_0^s \left(\int_0^{s_1}\tilde{\phi}(s_2)d s_2\right)d s_{1}\tilde{\phi}(s)ds
    \\
    &\quad-\frac{12u^2}{\tau}\int_0^\tau\int_0^s\left(s_1-2\frac{u^2}{\tau}\int_0^{s_1} \int_0^{s_2} \tilde{\phi}(s_3)d s_3\tilde{\phi}(s_2)d s_2\right)\tilde{\phi}(s_1)d s_1\tilde{\phi}(s)ds
    ))\\
    &=-\frac{\beta_0^2 u^2}{8\sigma_0^2 \tau}(e^{-\frac{u^2}{2\tau}\int_0^\tau \tilde{\phi}(u)^2 du}(3\tau^2-\frac{24 u^2}{\tau}\int_0^\tau \int_0^{s} \int_0^{s_1} \tilde{\phi}(s_2)d s_2\tilde{\phi}(s_1)d s_1 ds\\
    &\quad-\frac{12u^2}{\tau}\int_0^\tau\int_0^s\left(s_1-2\frac{u^2}{\tau}\int_0^{s_1} \int_0^{s_2} \tilde{\phi}(s_3)d s_3\tilde{\phi}(s_2)d s_2\right)\tilde{\phi}(s_1)d s_1\tilde{\phi}(s)ds
    )).
\end{align*}
Also,
\begin{align*}
    B_{\tau,3}&=-\frac{(\beta_0^\prime)^2 u^2}{2\sigma_0^2\tau}\E\left[\left(\int_0^\tau W_s^\prime d W_s\right)^2 e^{iu\intphitilde{s}}\right]\\
    &=-\frac{(\beta_0^\prime)^2 u^2}{2\sigma_0^2\tau}\left(2\E\left[\int_0^\tau(\int_0^{s}W_{s_1}^\prime d W_{s_1}) W_s^\prime d W_s e^{iu\intphitilde{s}}\right]+\E\left[\int_0^\tau(W_s^\prime)^2ds e^{iu\intphitilde{s}}\right]\right)\\
     &=-\frac{(\beta_0^\prime)^2 u^2}{2\sigma_0^2\tau}(2 \intphitildesqminus{s}\int_0^\tau \intphitildesqplus{s}\E[e^{iu\frac{\int_0^s\tilde{\phi}(u)dW_u}{\sqrt{\tau}}}(\int_0^s W_{s_1}^{\prime} d W_{s_1} W_s^\prime d W_s\\
     &\quad + iu\frac{\tilde{\phi}(s)}{\sqrt{\tau}} \int_0^s W_{s_1}^{\prime} d W_{s_1} W_s^\prime d W_s ds)]+\intphitildesqminus{s}\int_0^\tau \intphitildesqplus{s}\E[e^{iu\frac{\int_0^s\tilde{\phi}(u)dW_u}{\sqrt{\tau}}}(W_{s}^\prime)^2]ds)\\
      &=-\frac{(\beta_0^\prime)^2 u^2}{2\sigma_0^2\tau}(2 \frac{iu}{\sqrt{\tau}}\intphitildesqminus{s}\int_0^\tau \intphitildesqplus{s}\E[e^{iu\frac{\int_0^s\tilde{\phi}(u)dW_u}{\sqrt{\tau}}}(\int_0^s W_{s_1}d W_{s_1})W_s^\prime]\tilde{\phi}(s)ds\\
     &\quad+\intphitildesqminus{s}\int_0^\tau \int_0^s \intphitildesqplus{s_1}\E[e^{iu\frac{\int_0^{s_1}\tilde{\phi}(u)dW_u}{\sqrt{\tau}}} (2W_{s_1}^\prime dW_{s_1}^\prime + d s_1)])\\
&=-\frac{(\beta_0^\prime)^2 u^2}{2\sigma_0^2\tau}\intphitildesqminus{s}(2 \frac{iu}{\sqrt{\tau}}\int_0^\tau \intphitildesqplus{s}\E[e^{iu\frac{\int_0^s\tilde{\phi}(u)dW_u}{\sqrt{\tau}}}(\int_0^s W_{s_1}d W_{s_1})W_s^\prime]\tilde{\phi}(s)ds+\frac{\tau^2}{2})\\
&=-\frac{(\beta_0^\prime)^2 u^2}{2\sigma_0^2\tau}\intphitildesqminus{s}(2 \frac{(iu)^2}{\tau}\int_0^\tau   \int_0^s \intphitildesqplus{s_1}\E[e^{iu\frac{\int_0^{s_1}\tilde{\phi}(u)dW_u}{\sqrt{\tau}}}(W_{s_1}^\prime)^2]\tilde{\phi}(s_1)d s_1
\tilde{\phi}(s)ds+\frac{\tau^2}{2})\\
&=-\frac{(\beta_0^\prime)^2 u^2}{2\sigma_0^2\tau}\intphitildesqminus{s}(- \frac{2u^2}{\tau}\int_0^\tau   \int_0^s s_1\tilde{\phi}(s_1)d s_1
\tilde{\phi}(s)ds+\frac{\tau^2}{2}).
\end{align*}

\qed

\subsection{Proof of Corollary \ref{cor:expansion_piecewise}}\label{appendix:piecewise_shift}

The proof hinges on computing the integrals appearing in Theorem \ref{th:expansion} when $\phi(t)$ is of the form in Eq. \eqref{eqz:phi_piecewise}. We make use of the notation introduced in Eq. \eqref{eqz:delta_phi} under the assumption that $a_0=\tau_0=0$. Without loss of generality, we only focus on the largest tenor $\tau_n$. We have

\begin{align*}
    \int_0^{\tau_n}\tilde{\phi}^2 (s) ds = \int_0^{\tau_1}\tilde{\phi}^2(s) d s + \sum_{k=1}^{n-1} \int_{\tau_k}^{\tau_{k+1}}\tilde{\phi}^2(s)ds&=\tau_1+\sum_{k=1}^{n-1}\int_{\tau_k}^{\tau_{k+1}}\left(1+\frac{a_k}{\sigma_0}\right)^2 ds\\
    &=\tau_1+\sum_{k=1}^{n-1}\left(1+\frac{a_k}{\sigma_0}\right)^2
(\tau_{k+1}-\tau_k)=\langle\tilde{\Phi}_2^{(n)},\Delta\Tcal_1^{(n)}\rangle,
\end{align*}
\begin{align*}
    \int_0^{\tau_n}\tilde{\phi}(s)\left(\int_0^s \tilde{\phi}(s_1)d s_1\right) d s&= \int_0^{\tau_1}\tilde{\phi}(s)\left(\int_0^s \tilde{\phi}(s_1)d s_1\right) d s+ \sum_{k=1}^{n-1}\int_{\tau_k}^{\tau_{k+1}}\tilde{\phi}(s)\left(\int_0^s \tilde{\phi}(s_1)d s_1\right) d s\\
    &=\frac{\tau_1^2}{2}+\frac{1}{2}\sum_{k=1}^{n-1}\left(1+\frac{a_k}{\sigma_0}\right)^2 (\tau_{k+1}^2-\tau_{k}^2)=\frac{1}{2}\langle \tilde{\Phi}_2^{(n)},\Delta \Tcal_2^{(n)}\rangle,
\end{align*}
\begin{align*}
    \int_0^{\tau_n}\left(\int_0^s \tilde{\phi}(s_1)d s_1\right) d s&= \int_0^{\tau_1}\left(\int_0^s \tilde{\phi}(s_1)d s_1\right) d s+\sum_{k=1}^{n-1}\int_{\tau_k}^{\tau_{k+1}}\left(\int_0^s \tilde{\phi}(s_1)d s_1\right) d s\\
    &=\frac{\tau_1^2}{2}+\frac{1}{2}\sum_{k=1}^{n-1}\left(1+\frac{a_k}{\sigma_0}\right)(\tau_{k+1}^2-\tau_k^2)=\frac{1}{2}\langle \tilde{\Phi}_1^{(n)},\Delta \Tcal_2^{(n)}\rangle,
\end{align*}
\begin{align*}
    \int_0^{\tau_n} s \tilde{\phi}(s) ds = \int_0^{\tau_1} s \tilde{\phi}(s) ds + \sum_{k=1}^{n-1}\int_{\tau_k}^{\tau_{k+1}} s \tilde{\phi}(s) ds=& \\
    = \frac{\tau_1^2}{2}+\frac{1}{2}\sum_{k=1}^{n-1}\left(1+\frac{a_k}{\sigma_0}\right)(\tau_{k+1}^2-\tau_k^2)=\frac{1}{2}\langle \tilde{\Phi}_1^{(n)},\Delta \Tcal_2^{(n)}\rangle,
\end{align*}
\begin{align*}
    \int_0^{\tau_n} \int_0^s \int_{0}^{s_1}\tilde{\phi}(s_2) d s_2 d s_1 d s&=  \int_0^{\tau_1} \int_0^s \int_{0}^{s_1}\tilde{\phi}(s_2) d s_2 d s_1 d s+ \sum_{k=1}^{n-1} \int_{\tau_k}^{\tau_{k+1}} \int_0^s \int_{0}^{s_1}\tilde{\phi}(s_2) d s_2 d s_1 d s\\
    &=\frac{\tau_1^3}{6}+\frac{1}{6}\sum_{k=1}^{n-1}\left(1+\frac{a_k}{\sigma_0}\right)(\tau_{k+1}^3-\tau_k^3)=\frac{1}{6}\langle\tilde{\Phi}_1^{(n)},\Delta \Tcal_3^{(n)}\rangle,
\end{align*}
\begin{align*}
    \int_0^{\tau_n} \int_0^s \tilde{\phi}(s_1) d s_1 \tilde{\phi}(s) d s&=  \int_0^{\tau_1} \int_0^s \tilde{\phi}(s_1) d s_1 \tilde{\phi}(s) d s+ \sum_{k=1}^{n-1} \int_{\tau_k}^{\tau_{k+1}} \int_0^s\tilde{\phi}(s_1) d s_1 \tilde{\phi}(s) d s\\
    &=\frac{\tau_1^2}{2}+\frac{1}{2}\sum_{k=1}^{n-1}\left(1+\frac{a_k}{\sigma_0}\right)^2(\tau_{k+1}^2-\tau_k^2)=\frac{1}{2}\langle\tilde{\Phi}_2^{(n)},\Delta \Tcal_2^{(n)}\rangle,
\end{align*}
and
\begin{align*}
    \int_0^{\tau_n} \int_0^s \int_{0}^{s_1}\tilde{\phi}(s_2) d s_2 \tilde{\phi}(s_1)d s_1 d s&=\frac{\tau_1^3}{6}+\sum_{k=1}^{n-1}\int_{\tau_k}^{\tau_{k+1}} \int_0^s \int_{0}^{s_1}\tilde{\phi}(s_2) d s_2 \tilde{\phi}(s_1)d s_1 d s\\
    &=\frac{\tau_1^3}{6}+\frac{1}{6}\sum_{k=1}^{n-1}\left(1+\frac{a_k}{\sigma_0}\right)^2 (\tau_{k+1}^3-\tau_k^3)=\frac{1}{6}\langle\tilde{\Phi}_2^{(n)},\Delta \Tcal_3^{(n)}\rangle.
\end{align*}
Splitting, again, the integral $\int_0^{\tau_n}$ in terms of $\int_0^{\tau_1} + \sum_{k=1}^{n-1} \int_{\tau_k}^{\tau_{k+1}}$, we obtain
\begin{align*}
\int_0^{\tau_1} \int_0^s 
    \left( 
        s_1 - 2\frac{u^2}{\tau_1} \int_0^{s_1} \int_0^{s_2} \tilde{\phi}(s_3) \, ds_3 \, \tilde{\phi}(s_2) \, ds_2 
    \right) 
    \tilde{\phi}(s_1) \, ds_1 \, \tilde{\phi}(s) \, ds
&= \int_0^{\tau_1} \int_0^{s_1} \left( s_1 - 2\frac{u^2}{\tau_1} \cdot \frac{s_1^2}{2} \right) \, ds_2 \, ds_1 \\
&= \frac{\tau_1^3}{6} - \frac{u^2}{12\tau_1} \tau_1^4
\end{align*}
and, for any $k=1,\dots, n-1$,
\begin{align*}
& \int_{\tau_{k}}^{\tau_{k+1}}\int_0^s\left(s_1-2\frac{u^2}{\tau_n}\int_0^{s_1} \int_0^{s_2} \tilde{\phi}(s_3)d s_3\tilde{\phi}(s_2)d s_2\right)\tilde{\phi}(s_1)d s_1\tilde{\phi}(s)ds=\\
& = \int_{\tau_k}^{\tau_{k+1}}\frac{s^2}{2}\left(1+\frac{a_k}{\sigma_0}\right)^2-2\frac{u^2}{\tau_n}\frac{s^3}{6}\left(1+\frac{a_k}{\sigma_0}\right)^4 ds\\
&=\frac{1}{6}\Tcal_{3,k}^{(n)}\left(1+\frac{a_k}{\sigma_0}\right)^2-   \frac{u^2}{12\tau_n }\left(1+\frac{a_k}{\sigma_0}\right)^4 \Tcal_{4,k}^{(n)}\\
&=\frac{1}{6}\tilde{\Phi}_{2,k} \Delta \Tcal_{3,k}^{(n)} -  \frac{u^2}{12\tau_n }\tilde{\Phi}_{4,k}^{(n)}\Delta \Tcal_{4,k}^{(n)},
\end{align*}
giving
\begin{equation*}
\int_0^{\tau_n} \int_0^s \left(
    s_1 
    - 2\frac{u^2}{\tau_1} \int_0^{s_1} \int_0^{s_2} \tilde{\phi}(s_3) \, ds_3 \, \tilde{\phi}(s_2) \, ds_2
\right)
\tilde{\phi}(s_1) \, ds_1 \, \tilde{\phi}(s) \, ds
= \frac{1}{6} \langle \tilde{\Phi}_2^{(n)}, \Delta \mathcal{T}_3^{(n)} \rangle 
- \frac{u^2}{12\tau_n} \langle \tilde{\Phi}_4^{(n)}, \Delta \mathcal{T}_4^{(n)} \rangle.
\end{equation*}
Finally,
\begin{align*}
    \int_0^{\tau_{n}}\int_0^s s_1\tilde{\phi}(s_1)d s_1
\tilde{\phi}(s)ds&=\int_0^{\tau_{1}}\int_0^s s_1\tilde{\phi}(s_1)d s_1
\tilde{\phi}(s)ds+\sum_{k=1}^{n-1}\int_{\tau_{k}}^{\tau_{k+1}}\int_0^s s_1\tilde{\phi}(s_1)d s_1
\tilde{\phi}(s)ds\\
&=\frac{\tau_1^3}{6}+\sum_{k=1}^{n-1}\left(1+\frac{a_k}{\sigma_0}\right)^2(\tau_{k+1}^3-\tau_k^3)=\frac{1}{6}\langle \tilde{\Phi}_2^{(n)},\Delta \Tcal_3^{(n)}\rangle.
\end{align*}

\subsection{Proof of Theorem \ref{secondo}} \label{proof_secondo}
Consider the standardized log-return $Z_\tau^c = \frac{X^c_{\tau} - X^c_0 - \mu_0 \tau}{\sigma_0 \sqrt{\tau}},$ whose characteristic function admits the small-$\tau$ expansion in Corollary \ref{cor:expansion}. Express the expansion as follows:
\begin{eqnarray*}
\mathbb{E}[e^{iu Z_\tau^c}] &=& e^{-u^2/2}(1 + \Psi_{\tau}(u)) + O(\tau^{3/2})\psi(u).
\end{eqnarray*}
Write
\begin{eqnarray*}
\log \mathbb{E}[e^{iu Z_\tau^c}] = -\frac12u^2 + \log(1+\Psi_\tau(u)) +O(\tau^{3/2}),
\end{eqnarray*}
and linearize the logarithm on the right-hand side around zero. Because $\Psi_\tau(u)=O(\sqrt{\tau})$, and we are excluding orders higher than $\tau,$ the linearization is
\begin{eqnarray*}
\log(1+\Psi_\tau(u))=\Psi_\tau(u)-\frac12\Psi_\tau(u)^2+ o(\tau).
\end{eqnarray*}
The only contribution associated with $- \frac{1}{2}\Psi_{\tau}(u)^2$ is the square of the $\sqrt{\tau}$ term in $\Psi_{\tau}(u),$ namely $- i u^3\frac{\tilde\beta_0\rho_0}{2\sigma_0}\sqrt{\tau}.$ This contribution cancels exactly the $u^6$ term in $\Psi_{\tau}(u)$. Hence, collecting terms,
\begin{eqnarray*}
\log \mathbb{E}[e^{iu Z_\tau^c}] = -\frac12u^2 + a_2 u^2\tau - i a_3 u^3\sqrt{\tau} + a_4 u^4\tau + o(\tau),
\end{eqnarray*}
with
\begin{eqnarray*}
a_2 &=& -\frac{\alpha_0+\delta_0}{2\sigma_0} -\frac{\tilde\beta_0^2}{4\sigma_0^2}, \\
a_3 &=& \frac{\tilde\beta_0\rho_0}{2\sigma_0}, \\
a_4 &=& \frac{\eta_0}{6\sigma_0} + \frac{\tilde\beta_0^2}{6\sigma_0^2}(1+2\rho_0^2).
\end{eqnarray*}
Let us now compute \textit{cumulants}. We have
\begin{eqnarray*}
\log \mathbb{E}[e^{iu Z_\tau^c}] &=& \sum_{n\ge1}\frac{(iu)^n}{n!}\kappa_n(\tau). 
\end{eqnarray*}
Matching coefficients associated with the characteristic exponent $u$, now, gives
\begin{eqnarray*}
\kappa_1(\tau) & = & 0 \\
\kappa_2(\tau) &=& 1 - \left(\frac{\alpha_0+\delta_0}{\sigma_0} + \frac{\tilde\beta_0^2}{2\sigma_0^2}\right)\tau+o(\tau), \\
\kappa_3(\tau) &=& \underbrace{3\frac{\tilde\beta_0\rho_0}{\sigma_0}}_{\theta_3}\sqrt{\tau}+o(\sqrt{\tau}), \\
\kappa_4(\tau) &=& \underbrace{4\left(\frac{\eta_0}{\sigma_0}+\frac{\tilde\beta_0^2}{\sigma_0^2}(1+2\rho_0^2)\right)}_{\theta_4}\tau +o(\tau). 
\end{eqnarray*}
Thus, the cumulants of the price process, i.e., $X^c_{\tau} - X^c_0 = \mu_0 \tau + \sigma_0 \sqrt{\tau}Z^c_{\tau},$ are:
\begin{eqnarray*}
\kappa^{\star}_1(\tau) = \left(\sigma_0 \sqrt{\tau}\right)\kappa_1(\tau)+\mu_0\tau = \mu_0 \tau \\
\kappa^{\star}_2(\tau) = \left(\sigma_0 \sqrt{\tau}\right)^2\kappa_2(\tau) = \sigma^2_0\tau +O(\tau^2), \\
\kappa^{\star}_3(\tau) = \left(\sigma_0 \sqrt{\tau}\right)^3\kappa_3(\tau) = \sigma^3_0\theta_3\tau^2 +o(\tau^2), \\
\kappa^{\star}_4(\tau) = \left(\sigma_0 \sqrt{\tau}\right)^4\kappa_4(\tau) = \sigma^4_0\theta_4\tau^3 +o(\tau^3).
\end{eqnarray*}
From \citet{forde2009small}, the (short-time) implied-volatility smile is
\begin{eqnarray*}
I(x)=\frac{x}{\sqrt{2\Lambda^*(x)}},
\end{eqnarray*}
where $\Lambda^*(x)$ is the Legendre-Fenchel transform defined as 
\begin{eqnarray*}
\Lambda^*(x) = \sup_{p\in\mathbb{R}}\{px-\Lambda(p)\},
\end{eqnarray*}
with $\Lambda(p) := \lim_{\tau \rightarrow 0} \tau \log \mathbb{E}[e^{p\frac{(X_{\tau} - X_0)}{\tau}}]  =  \lim_{\tau \rightarrow 0}\Lambda_{\tau}(p)$. Like earlier, we may write
\begin{eqnarray*}
\log \mathbb{E}[e^{p\frac{(X_{\tau} - X_0)}{\tau}}] &=& \sum_{n\geq1} \frac{p^n}{n!} \frac{\kappa^{\star}_n(\tau)}{\tau^n} \\
&=& \frac{\mu_0\tau}{\tau} + \frac{\sigma^2_0\tau p^2}{2\tau^2}+\frac{\sigma^3_0 \theta^3 \tau^2 p^3}{6\tau^3} + \frac{\sigma^4_0 \theta_4 \tau^3 p^4}{24 \tau^4} + \sum_{n\geq5} \frac{p^n}{n!} \frac{\kappa^{\star}_n(\tau)}{\tau^n}
\end{eqnarray*}
which gives
\begin{eqnarray*}
\Lambda(p)=  \lim_{\tau \rightarrow 0} \Lambda_{\tau}(p)=  \lim_{\tau \rightarrow 0} \sum_{n\geq1} \frac{p^n}{n!} \frac{\kappa^{\star}_n(\tau)}{\tau^{(n-1)}} = \frac{\sigma^2_0p^2}{2}+\frac{\sigma^3_0\theta_3p^3}{6} +\frac{\sigma^4_0\theta_4p^4}{24} + \lim_{\tau \rightarrow 0} \sum_{n\geq5} \frac{p^n}{n!} \frac{\kappa^{\star}_n(\tau)}{\tau^{(n-1)}}.
\end{eqnarray*}
Because we are interested in a second-order expansion of $I(x)$, the remainder term can be dropped in what follows. Notice that the maximizer $p(x)$ satisfies
\begin{eqnarray*}
x = \Lambda'(p) = \sigma^2_0p + \frac{\sigma^3_0\theta_3p^2}{2} + \frac{\sigma^4_0\theta_4p^3}{6}.
\end{eqnarray*}
We consider an ansatz of the form:
\begin{eqnarray*}
p(x)=\gamma x+\alpha x^2+\beta x^3.
\end{eqnarray*}
Plugging now into the previous expression, and ignoring $o(x^3)$ terms, we have 
\begin{eqnarray*}
x &=& \sigma^2_0\left(\gamma x+\alpha x^2+\beta x^3\right) + \frac{\sigma^3_0\theta_3}{2}(\gamma x+\alpha x^2+\beta x^3)^2  + \frac{\sigma^4_0\theta_4}{6}(\gamma x+\alpha x^2+\beta x^3)^3  \\
&=& \sigma^2_0(\gamma x+\alpha x^2+\beta x^3)+ \frac{\sigma^3_0\theta_3}{2}\gamma^2 x^2+ \sigma^3_0\theta_3 \gamma \alpha x^3  + \frac{\sigma^4_0\theta_4}{6} \gamma^3 x^3 + O(x^4).
\end{eqnarray*}
Matching coefficients on the left-hand side and the right-hand side yields
\begin{eqnarray*}
\gamma &=& \frac{1}{\sigma^2_0}, \\
\alpha &=& -\frac{\theta_3}{2 \sigma^3_0}, \\
\beta &=& \frac{\theta_3^2}{2 \sigma_0^4}-\frac{\theta_4}{6 \sigma_0^2}.
\end{eqnarray*}
Finally, substituting back into $\sup_{p\in\mathbb{R}}\{px-\Lambda(p)\},$ we have 
\begin{eqnarray*}
\Lambda^*(x) &=& \underbrace{\frac{1}{\sigma^2_0} x^2 -\frac{\theta_3}{2 \sigma^3_0} x^3+\left(\frac{\theta_3^2}{2 \sigma_0^4}-\frac{\theta_4}{6 \sigma_0^2} \right) x^4}_{p(x)x} \\
&& \underbrace{-  \frac{\sigma^2_0}{2}\left(\frac{1}{\sigma^2_0}\right)^2x^2 -\frac{\sigma^2_0}{2}\left(-\frac{\theta_3}{2 \sigma^3_0}\right)^2x^4 - \sigma^2_0 \left(\frac{1}{\sigma^2_0}\right) \left(-\frac{\theta_3}{2 \sigma^3_0}\right) x^3  - \sigma^2_0\left(\frac{1}{\sigma^2_0}\right)\left(\frac{\theta_3^2}{2 \sigma_0^4}-\frac{\theta_4}{6 \sigma_0^2}\right)x^4}_{- \textrm{quadratic term in}~ \Lambda(p(x))} \\
&&\underbrace{- \frac{\sigma^3_0\theta_3}{6}\left(\frac{1}{\sigma^2_0}\right)^3x^3 + \frac{\sigma^3_0\theta_3}{6}3\left(\frac{\theta_3}{2\sigma^7_0}\right)x^4}_{- \textrm{cubic term in}~ \Lambda(p(x))}\\
&&\underbrace{- \frac{\sigma^4_0\theta_4}{24}\left(\frac{1}{\sigma^2_0}\right)^4x^4}_{-\textrm{quartic term in}~ \Lambda(p(x))}+O(x^5) \\
&=& \frac{1}{2\sigma^2_0}x^2 - \frac{\theta_3}{6\sigma^3_0}x^3 + \Big(\frac{\theta_3^2}{8\sigma^4_0}-\frac{\theta_4}{24\sigma^4_0}\Big)x^4
+ O(x^5).
\end{eqnarray*}
Now, define
\begin{eqnarray*}
\frac{x}{\sqrt{2\Lambda^*(x)}}=\frac{x}{\frac{|x|}{\sigma_0}\sqrt{1+ \underbrace{\left(- \frac{\theta_3}{3\sigma_0}x + \Big(\frac{\theta_3^2}{4\sigma^2_0}-\frac{\theta_4}{12\sigma^2_0}\Big)x^2
+ O(x^3)\right)}_{=:y}}},
\end{eqnarray*}
and linearize as follows
\begin{eqnarray*}
\frac{1}{\sqrt{1+y}}= 1-\frac{y}{2}+\frac{3y^2}{8}+O(y^3).
\end{eqnarray*}
Working with $x>0,$ and replacing $y$ with its definition, we have
\begin{eqnarray*}
\frac{x}{\sqrt{2\Lambda^*(x)}} &=& \sigma_0 - \sigma_0\frac{y}{2}+\sigma_0\frac{3y^2}{8}+O(y^3) \\
&=&\sigma_0 - \sigma_0\frac{\left(- \frac{\theta_3}{3\sigma_0}x + \Big(\frac{\theta_3^2}{4\sigma^2_0}-\frac{\theta_4}{12\sigma^2_0}\Big)x^2
+ O(x^3)\right)}{2}+\sigma_0\frac{3\left(- \frac{\theta_3}{3\sigma_0}x\right)^2}{8}+O(y^3).
\end{eqnarray*}
Thus,
\begin{eqnarray*}
I(x) = \sigma_0 \left[1+\frac{\theta_3}{6\sigma_0}x+\left(\frac{\theta_4}{24\sigma^2_0}-\frac{\theta_3^2}{12\sigma^2_0}\right)x^2+O(x^3)
\right].
\end{eqnarray*}
Finally,
\begin{eqnarray*} 
I(0)=\sigma_0,\qquad I'(0)=\frac{\theta_3}{6}\quad \textrm{and}\quad I''(0) = \frac{\theta_4}{12\sigma_0}-\frac{\theta_3^2}{6\sigma_0}.
\end{eqnarray*}

\qed

\end{document}